\documentclass[a4paper,11pt]{article} 
\usepackage[english]{babel}
\usepackage[a4paper]{geometry}
\usepackage{amsmath}
\usepackage{amsthm}
\usepackage{amssymb}
\usepackage{bbm}
\usepackage{color}
\usepackage{enumerate}
\usepackage{mathrsfs}
\makeatletter
\let\@fnsymbol\@arabic %alph %@arabic  
\makeatother     
\usepackage{hyperref}         %%%% click to contents
\def\bR{\mathbb{R}}

\def\bN{\mathbb{N}}

\def\bZ{\mathbb{Z}}

\def\cI{\mathcal{I}}

\def\cQ{\mathcal{Q}}

\def\cA{\mathcal{A}}
\def\cM{\mathcal{M}}
\def\cP{\mathcal{P}}

\def\cV{\mathcal{V}}
\def\cO{\mathcal{O}}

\def\cF{\mathcal{F}}
\def\cG{\mathcal{G}}
\def\cL{\mathcal{L}}
\def\cJ{\mathcal{J}}
\def\cZ{\mathcal{Z}}
\def\cN{\mathcal{N}}
\def\cE{\mathcal{E}}
\def\cK{\mathcal{K}}

\def\cS{\mathcal{S}}
\def\cT{\mathcal{T}}
\def\eps{\varepsilon}
\def\ph{\varphi}

\def\wt{\widetilde}

\def\ZZZ{\mathbb{Z}}

\def\indic{\hbox{\raise-2pt \hbox{\indbf 1}}}

\let\dpr=\partial

\let\==\equiv

\let\io=\infty
\let\0=\noindent

\def\*{{\hfill\break\null\hfill\break}}

\def\tende#1{\,\vtop{\ialign{##\crcr\rightarrowfill\crcr
             \noalign{\kern-1pt\nointerlineskip}
             \hskip3.pt${\scriptstyle #1}$\hskip3.pt\crcr}}\,}
\def\otto{\,{\kern-1.truept\leftarrow\kern-5.truept\to\kern-1.truept}\,}
\def\fra#1#2{{#1\over#2}}

\newtheorem{theorem}{Theorem}[section]  % use thm for %Theorems to keep numbering consistent

\newtheorem{prop}[theorem]{Proposition}
\newtheorem{lemma}[theorem]{Lemma}

\numberwithin{equation}{section}

\def\be{\begin{equation}}
\def\ee{\end{equation}}

\newcommand{\hc}{\mbox{h.c.}}

     \let\g=\gamma          
        \let\k=\kappa     \let\l=\lambda
                          
\let\s=\sigma \let\t=\tau         \let\ph=\varphi   
   \let\o=\omega     
 \let\D=\Delta       \let\L=\Lambda    
             
\let\O=\Omega

\def\aa{\mathfrak{a}}
\def \blue#1 {\textcolor{blue}{#1}}
\def \red#1 {\textcolor{red}{#1}}

\definecolor{lightblue}{rgb}{0, 0.33, 0.71}

\title{A second order upper bound for the ground state energy of a hard-sphere gas in the Gross-Pitaevskii regime} 

\author{Giulia Basti$^{*,}$\footnote{Electronic mail: giulia.basti@gssi.it, ORCID: 0000-0002-3745-6293}\;, Serena Cenatiempo$^{*,}$\footnote{Corresponding author: serena.cenatiempo@gssi.it, ORCID: 0000-0002-8667-8300}, Alessandro Olgiati$^{\circ,}$\footnote{Electronic mail: alessandro.olgiati@math.uzh.ch, ORCID: 0000-0002-0170-7935}, \\ Giulio Pasqualetti$^{\circ,}$\footnote{Electronic mail: giulio.pasqualetti@math.uzh.ch}, Benjamin Schlein$^{\circ,}\footnote{Electronic mail: benjamin.schlein@math.uzh.ch, ORCID: 0000-0003-0910-9136}$\\[0.2cm]
{\footnotesize $^{*}$Gran Sasso Science Institute, Viale Francesco Crispi 7, 67100 L'Aquila, Italy}\\
{\footnotesize $^{\circ}$Institute of Mathematics, University of Zurich, Winterthurerstrasse 190, 8057 Zurich, Switzerland}}

\begin{document}

\maketitle

\begin{abstract} 
We prove an upper bound for the ground state energy of a Bose gas consisting of $N$ hard spheres with radius $\frak{a}/N$, moving in the three-dimensional unit torus $\Lambda$. Our estimate captures the correct asymptotics of the ground state energy, up to errors that vanish in the limit $N \to \infty$. The proof is based on the construction of an appropriate trial state, given by the product of a Jastrow factor (describing two-particle correlations on short scales) and of a wave function constructed through a (generalized) Bogoliubov transformation, generating orthogonal excitations of the Bose-Einstein condensate and describing correlations on large scales. 
\end{abstract}

\section{Introduction and main result}

In \cite{LHY}, Lee-Huang-Yang predicted that the ground state energy per particle of a 
system of $N$ bosons moving in a box with volume $N/\rho $ and interacting through a potential with scattering length $\frak{a}$ is given, as $N \to \infty$, by 
\begin{equation}\label{eq:LHY}  e(\rho) = 4\pi \frak{a} \rho \Big[ 1 + \frac{128}{15\sqrt{\pi}} (\rho \frak{a}^3)^{1/2} + \dots \Big]  \end{equation} 
up to corrections that are small, in the low density limit $\rho \frak{a}^3 \ll 1$ (see \cite{SP-book, LSSY-book} for the heuristics behind this formula and its relation with the expected occurrence of Bose-Einstein condensation in dilute Bose gases). At leading order, the validity of (\ref{eq:LHY}) follows from the upper bound obtained in \cite{Dy} and from the matching lower bound established in \cite{LY}. Recently, also the second order term on the r.h.s. of (\ref{eq:LHY}) has been rigorously justified. The upper bound has been shown in \cite{YY} (through a clever modification of a quasi-free trial state proposed in \cite{ESY})  and (for a larger class of interactions and using a simpler trial state) in \cite{BCS}. As for the lower bound, it has been first obtained in \cite{FS1} for integrable potentials and then in \cite{FS2}, for particles interacting through general potentials, including hard-spheres. The upper bound for the case of hard-sphere potential is still an open question. An alternative approach to the study of the ground state energy of the zero temperature Bose gas, still not justified rigorously but possibly valid beyond the dilute regime, has been proposed in \cite{L63} and recently revived in \cite{SimpleEq1, SimpleEq2, SimpleEq3}.

Trapped Bose gases can be described as systems of $N$ bosons, confined by external fields in a volume of order one and interacting through a radial, repulsive potential $V$ with scattering length of the order $N^{-1}$; this scaling limit is known as the Gross-Pitaevskii regime (see \cite[Chapter 6]{LSSY-book} for an introduction, and \cite{Rougerie,S} for reviews of more recent results). Focussing for simplicity on systems trapped in the unit torus $\Lambda$, the Hamilton operator takes the form 
\begin{equation}\label{eq:HN1}
H_N = \sum_{j=1}^N -\Delta_{x_j} + \sum_{i<j}^N N^2 V (N (x_i - x_j)) 
\end{equation} 
and acts on $L^2_s (\L^N)$, the subspace of $L^2 (\L^N)$ consisting of functions that are symmetric w.r.t. permutations of the $N$ particles. Note that $x_i-x_j$ is here the difference between the position vectors of particles $i$ and $j$ on the torus. Equivalently, we can think of $x_i - x_j$ as the difference in $\bR^3$; however, in this case, $V$ has to be replaced by its periodisation. As proven in \cite{LY,LSY,NRS}, the ground state energy $E_N$ of (\ref{eq:HN1}) is given, to leading order, by 
\begin{equation}\label{eq:ENlead} E_N = 4\pi \frak{a} N + o (N) \end{equation} 
in the limit $N \to \infty$. For $V \in L^3 (\bR^3)$, more precise information on the low-energy spectrum of (\ref{eq:HN1}) has been determined in \cite{BBCS4}. Here, the ground state energy was proven to satisfy  
\begin{equation}\label{eq:gse1}
\begin{split}  E_N = \; & 4 \pi \aa (N  -1) \, +\,  e_\L \aa^2   \\
&-\frac 12 \sum_{p \in \L^*_+} \bigg[ p^2 + 8 \pi \aa - \sqrt{|p|^4 + 16 \pi \aa p^2} - \frac{(8 \pi \aa)^2}{2p^2}\bigg] +  \cO (N^{-1/4}) 
\end{split} \end{equation} 
where $\L^*_+ = 2 \pi \ZZZ^3 \setminus \{0\}$ and \begin{equation}\label{eq:eLambda} e_\L = 2 - \lim_{M \to \io} \sum_{\substack{p \in {\mathbb Z}^3 \setminus \{0\}:\\ |p_1|, |p_2|, |p_3| \leq M}} \frac{\cos(|p|)}{p^2}\,.\end{equation} 
Additionally, the spectrum of $H_N - E_N$ below a threshold $\zeta > 0$ was shown to consist of eigenvalues having the form 
\begin{equation}\label{eq:excited} \sum_{p \in 2\pi \bZ^3 \backslash \{ 0 \}} n_p \sqrt{|p|^4 + 16 \pi \frak{a} p^2} + \cO (N^{-1/4} \zeta^3)\,. \end{equation} 
A new and simpler proof of (\ref{eq:gse1}), (\ref{eq:excited}) was recently obtained in \cite{HST}, for $V \in L^2 (\Lambda)$. Moreover, these results have been also extended to the non-homogeneous case of Bose gases trapped by external fields in \cite{NT,BSS}. 

While the approach of \cite{NT} applies to $V \in L^1 (\bR^3)$, the validity of (\ref{eq:gse1}), (\ref{eq:excited}) for bosons interacting through non-integrable potentials is still an open question. The goal of this paper is to prove that (\ref{eq:gse1}) remains valid, as an upper bound, for particles interacting through a hard-sphere potential. 

We consider $N$ bosons in $\L=[-\tfrac 1 2, \tfrac 1 2]^3 \subset \bR^3$, with periodic boundary conditions. We assume particles to interact through a hard-sphere potential, with radius $\aa/N$, for some $\frak{a} > 0$.  We are interested in the ground state energy of the system, defined by 
\begin{equation}\label{eq:gs-en}
E^\text{hs}_N = \inf \, \Big\langle \Psi , \sum_{j=1}^N -\Delta_{x_j} \Psi \Big\rangle 
\end{equation}
where the infimum is taken over all normalized wave functions $\Psi \in L_s^2 (\Lambda^N)$ satisfying the hard-core condition 
\begin{equation}\label{eq:hard} \Psi (x_1, \dots , x_N ) = 0 \end{equation} 
almost everywhere on the set \[ \bigcup_{i<j}^N \big\{ (x_1, \dots , x_N) \in \bR^{3N} : |x_ i - x_j| \leq \frak{a}/N \big\} \, . \]

\begin{theorem} \label{thm:main}  Let $E_N^{\text{hs}}$ be defined as in \eqref{eq:gs-en}. 
There exist $C , \eps > 0$ such that 
\begin{equation}\label{eq:main} 
\begin{split} 
E^{\text{hs}}_N \leq \; & 4 \pi \aa(N  -1) \, +\,  e_\L \aa^2   \\
&-\frac 12 \sum_{p \in \L^*_+} \bigg[ p^2 + 8 \pi \aa - \sqrt{|p|^4 + 16 \pi \aa p^2} - \frac{(8 \pi \aa)^2}{2p^2}\bigg] +  C N^{-\eps}
\end{split} \end{equation} 
for all $N$ large enough, with $e_\Lambda$ defined as in (\ref{eq:eLambda}). 
\end{theorem}

{\it Remarks.}   

\medskip

 1) Theorem \ref{thm:main} establishes an upper bound for the ground state energy (\ref{eq:gs-en}). With minor modifications, it would also be possible to obtain upper bounds for low-energy excited eigenvalues, agreeing with (\ref{eq:excited}). To conclude the proof of the estimates (\ref{eq:gse1}), (\ref{eq:excited}) for particles interacting through hard-sphere potentials, we would need to establish matching lower bounds. A possible approach to achieve this goal (at least for the ground state energy) consists in taking the lower bound established in \cite{FS2}, for particles in the thermodynamic limit, and to translate it to the Gross-Pitaevskii regime. 

\medskip

2) We believe that the statement of Theorem \ref{thm:main} and its proof can also be extended to bosons in the Gross-Pitaevskii regime interacting through a larger class of potentials, combining a hard-sphere potential at short distances and an integrable potential at larger distances. This would require the extension of Lemma \ref{lm:hardcorescatt} to more general interactions. To keep our analysis as simple as possible, we focus here on hard-sphere bosons. 

\medskip

3) Theorem \ref{thm:main} and its proof could also be extended to systems of $N$ bosons interacting through a hard-sphere potential with radius of the order $N^{-1+\kappa}$ for sufficiently small $\kappa > 0$ (results for integrable potentials with scattering length of the order $N^{-1+\kappa}$ have been recently discussed in \cite{ABS,BCaS,F, Ba}). 

\bigskip

The proof of (\ref{eq:gse1}), (\ref{eq:excited}) obtained in \cite{BBCS4} is based on a rigorous version of Bogoliubov theory, developed in \cite{BBCS1,BBCS2,BBCS3}. The starting point of Bogoliubov theory is the observation that, at low energies, the Bose gas exhibits complete condensation; all particles, up to a fraction vanishing in the limit $N \to \infty$, can be described by the same zero-momentum orbital $\ph_0$ defined by $\ph_0 (x) = 1$, for all $x \in \L$. This, however, does not mean that the factorized wave function $\ph_0^{\otimes N}$ is a good approximation for the ground state of (\ref{eq:HN1}); in fact, its energy does not even approximate the ground state energy to leading order. To decrease the energy and approach (\ref{eq:ENlead}), correlations are crucial. The strategy developed in \cite{BBCS1,BBCS2,BBCS3,BBCS4} is based on the idea that most correlations can be inserted through the action of (generalized) Bogoliubov transformations, having the form
\begin{equation}\label{eq:BT1} T = \exp \left[ \frac{1}{2} \sum_{p \in \L^*_+} \eta_p \big( b_p^* b_{-p}^* - b_p b_{-p} \big) \right] \end{equation} 
where the (modified) creation and annihilation operators $b_p^*, b_p$ act on the Fock space of orthogonal excitations of the Bose-Einstein condensate; the precise definitions are given below, in Section \ref{sec:GN} (to be more precise, the action of (\ref{eq:BT1}) has to be corrected through an additional unitary operator, given by the exponential of a cubic, rather than quadratic, expression in creation and annihilation operators; see \cite{BBCS4} for details). An important feature of (generalized) Bogoliubov transformations of the form (\ref{eq:BT1}), which plays a major role in the derivation of (\ref{eq:gse1}), (\ref{eq:excited}), is the fact that their action on creation and annihilation operators is (almost) explicit. This makes computations relatively easy and it gives the possibility of including correlations also at very large length scales. 

Unfortunately, Bogoliubov transformations of the form (\ref{eq:BT1}) do not seem 
compatible with the hard-core condition (\ref{eq:hard}).
 As a consequence, they do not seem appropriate to construct trial states approximating the ground state energy of a system of particles interacting through a hard-sphere potential. A different class of trial states, for which  (\ref{eq:hard})   can be easily verified, consists of products having the form 
\begin{equation}\label{eq:jas} \Psi_N (x_1, \dots , x_N) = \prod_{j=1}^N  f (x_i - x_j) \,  \end{equation} 
for a function $f$ satisfying $f (x) = 0$, for all $|x| < \frak{a}/N$ (as mentioned after \eqref{eq:HN1}, also here $x_i -x_j$ is interpreted as difference on the torus). Such an ansatz was first used in the physics literature in \cite{B,D,J}; it is often known as Jastrow factor. In order for (\ref{eq:jas}) to provide a good approximation for the ground state energy, $f$ must describe two-particle correlations. Probably the simplest possible choice of $f$ is given by the solution 
\[ f(x) = \left\{ \begin{array}{ll} 0 &\quad \text{if } |x| < \frak{a}/N \\ 1 - \frac{\frak{a}}{N|x|} &\quad \text{if } |x| \geq \frak{a}/N \end{array} \right. \]
of the zero-energy scattering equation $-\Delta f = 0$, with the hard-core requirement $f(x) = 0$ for $|x| < \frak{a}/N$ and the boundary condition $f(x) \to 1$, as $|x| \to \infty$. The problem with this choice is the fact that $f$ has long tails; as a consequence, it is extremely difficult to control the product (\ref{eq:jas}). To make computations possible, we need to cutoff $f$ at some intermediate length scale $\frak{a}/N \ll \ell \ll 1$, requiring that $f (x) = 1$ for $|x| \geq \ell$ (the cutoff can be implemented in different ways; below, we will choose $f$ as the solution of a Neumann problem on the ball $|x| \leq \ell$ and we will keep it constant outside the ball). Choosing $\ell$ small enough (in particular, smaller than the typical distance among particles, which is of the order $N^{-1/3}$), the Jastrow factor becomes more manageable and it is not too difficult to show that its energy matches, to leading order, the ground state energy (\ref{eq:ENlead}). In the thermodynamic limit, this was first verified in \cite{Dy}, using a modification of (\ref{eq:jas}), considering only correlations among neighbouring particles. 

While Jastrow factors can lead to the correct leading order term in the ground state energy, it seems much more difficult to use (\ref{eq:jas}) to obtain an upper bound matching also the second order term on the r.h.s. of (\ref{eq:main}). The point is that the second order corrections are generated by correlations at much larger length scales; to produce the term on the second line of (\ref{eq:main}) we would need to take $\ell$ of order one, making computations very difficult. 

In order to prove Theorem \ref{thm:main}, we will therefore consider a trial state given by the product of a Jastrow factor (\ref{eq:jas}), describing correlations up to a sufficiently small length scale $1/N \ll \ell \ll 1$, and of a wave function $\Phi_N$, constructed through a Bogoliubov transformation, describing correlations on length scales larger than $\ell$. This allows us to combine the nice features of the Jastrow factor (in particular, the fact that it automatically takes care of the hard core condition (\ref{eq:hard})) and of the Bogoliubov transformation (in particular, their (almost) explicit action on creation and annihilation operators, which enables us to insert correlations at large length scales). 

The paper is organised as follows. In Section \ref{sec:trial_state}, we define our trial state $\Psi_N$ as the product of a Jastrow factor and an $N$-particle wave function $\Phi_N$, to be specified later on, and we compute its energy. One of the contributions to the energy of $\Psi_N$ is a three-body term; under certain conditions on $\Phi_N$ (see \eqref{eq:3D-assum}), we show that this term is negligible in Sect. \ref{sec:3body}. In Sect. \ref{sec:reduction} we then prove that  the remaining contributions to the energy can be reduced (again under suitable assumptions on $\Phi_N$; see  (\ref{eq:Delta-phi})) to the expectation of an effective Hamiltonian $H_N^{\mathrm{eff}}$, defined in \eqref{eq:Heff}. Sects. \ref{sec:GN} and \ref{sec:diag} are devoted to the study of $H_N^{\mathrm{eff}}$; the goal is to find $\Phi_N$ so that the expectation of $H_N^{\mathrm{eff}}$ produces the energy on the r.h.s of (\ref{eq:main}), up to negligible errors. Here, we use the approach developed in \cite{BBCS1,BBCS2,BBCS3}. In Sect. \ref{sec:bound}, we show that the chosen wave function $\Phi_N$ satisfies the bounds that were used in Sects. \ref{sec:3body} and \ref{sec:reduction}. Finally, in Sect. \ref{sec:main}, we put all ingredients together to conclude the proof of Theorem \ref{thm:main}. The proof of important properties concerning the solution of the scattering equations is deferred to Appendix \ref{sec:app}.

\section{The Jastrow factor and its energy}\label{sec:trial_state}

As explained in the introduction, our trial state involves a Jastrow factor, to describe short-distance correlations. To define the Jastrow factor, we choose $1/N \ll \ell \ll 1$ and we consider the ground state solution of the Neumann problem
\begin{equation} \label{eq:ev}
\left\{ \begin{array}{ll}
- \D f_\ell(x) =\l_\ell f_{\ell}(x)   \qquad &\text{for \;} \aa/N  \leq |x| \leq \ell \\
\; \; \, \partial_r f_\ell (x) = 0 , \qquad &\text{if \; } |x| = \ell 
\end{array} \right.
\end{equation}
on the ball $B_\ell = \{ x \in \bR^3 : |x| \leq \ell \}$, with the hard-core condition $f_\ell (x) = 0$ for $|x| \leq \frak{a}/N$ and the normalization $f_\ell (x) = 1$ for $|x| = \ell$ (we denote here by $\partial_r$ the radial derivative). We extend $f_\ell$ to $\L$ setting $f_\ell(x)=1$ for $|x|\in \L \backslash B_\ell$. We have 
\be \label{eq:fell}
- \D f_\ell(x) = \l_\ell \chi_\ell(x) f_{\ell}(x) 
\ee
where $\chi_\ell$ denotes the characteristic function of $B_\ell$. The following lemma establishes properties of $\lambda_\ell$, $f_\ell$, of the difference $\o_\ell(x)= 1- f_\ell(x)$ and of its Fourier coefficients 
\[
\widehat \o_\ell(p) =   \int e^{ip \cdot x} \o_\ell(x) dx 
\]
defined for $p \in \Lambda^* = 2\pi \bZ^3$ (since $\omega_\ell$ has compact support inside $[-1/2 ; 1/2]^3$, we can think of the integral as being over $\bR^3$).
\begin{lemma} \label{lm:hardcorescatt}
Let $\lambda_\ell$ denote the ground state eigenvalue appearing in (\ref{eq:ev}). Then 
\be  \label{eq:lambdaell}
\tan \big( \sqrt{\l_\ell}\, (\ell- \aa/N) \big) = \sqrt{\l_\ell} \,\ell \,.
\ee
For $N \ell \to \infty$, we find 
\begin{equation} \label{eq:lambdaell-exp}
\lambda_\ell =\frac{3 \mathfrak{a}}{N \ell^3} \left[ 1+\frac{9}{5}\frac{\mathfrak{a}}{N \ell}%+ \frac{459}{175}\frac{\mathfrak{a}^2}{(N \ell)^2} 
+  \mathcal{O}\Big(\frac{\mathfrak{a}^2}{N^2 \ell^2}\Big) \right] \,.
\end{equation}
The corresponding eigenvector $f_\ell$ is given by 
\begin{equation} \label{eq:fell-x}
f_\ell (x) = \frac{\ell}{|x|} \frac{\sin(\sqrt \l_\ell (|x|-\aa/N))}{\sin(\sqrt \l_\ell (\ell-\aa/N))} 
\end{equation} 
for all $\frak{a}/N \leq |x| \leq \ell$ ($f_\ell (x) = 0$ for $|x| \leq \aa /N$ and $f_\ell (x) = 1$ for $|x| > \ell$). We find 
\be  \label{eq:Vell-zero}
	N  \l_{\ell} \int \chi_{\ell} f^2_{\ell}\, dx = 4 \pi \aa  + \frac{24}{5}\pi \frac{\aa^2}{\ell N} + \cO \Big( \frac{\aa^2}{N^{2}\ell^2} \Big) \, .
\ee
%\be \label{eq:Vell-zero}
%\Big| N\l_\ell \int_\L \chi_\ell (x) f_\ell^2 (x) dx - 4 \pi \aa \Big| \leq \frac{C\aa^2}{N\ell}\,.
%\ee
With the notation $\o_\ell (x) = 1 - f_\ell (x)$, we have $\o_\ell (x) = 0$ for $|x| \geq \ell$ and, for $|x| \leq \ell$, the pointwise bounds  
\begin{equation}\label{eq:w-bds} 0 \leq \omega_\ell (x) \leq  \frac{C \aa}{N|x|}  , \qquad |\nabla \omega_\ell (x)| \leq  \frac{C \aa}{N|x|^2} 
\end{equation} 
for a constant $C>0$. Furthermore, there exists a constant $C>0$ so that
\be \label{eq:normsomega}
\Big| \|\o_\ell\|_1   - \frac {2} 5 \,\pi \mathfrak{a} \, \frac{\ell^2}{N} \Big| \leq C \frac{\aa^2 \ell}{N^2} \ee
and, for all $p \in [1, 3)$ and  $q \in [1, 3/2)$, 
	\begin{equation}  \label{eq:Lp-norms}
	\| \o_\ell \|_{p} \le C \ell^{\frac 3 p-1}N^{-1}\,, \qquad \| \nabla \o_\ell\|_{q} \le C \ell^{\frac 3 q-2}N^{-1}\, .
	\end{equation}
Finally, for $p \in \L^*$, let $\widehat{\o}_p$ denote the Fourier coefficients of $\o_\ell$. Then 
\be \label{eq:omegap}
|\widehat{\o}_\ell(p)| \leq C \min \left \{\frac{\ell^2 } N \,;\, \frac{1}{N |p|^2}\,;\, \frac{1}{|p|^3}\right\}\,.
\ee
\end{lemma}
We defer the proof of Lemma \ref{lm:hardcorescatt} to Appendix \ref{App:scatt}.

With the solution $f_\ell$ of the Neumann problem (\ref{eq:ev}), we consider trial states of the form  
\begin{equation}\label{eq:trial}
\Psi_N (x_1, \dots , x_N) = \Phi_N (x_1, \dots , x_N) \prod_{i<j}^N f_\ell (x_i - x_j) 
\end{equation} 
for $\Phi_N \in L^2_s (\Lambda^N)$ to be specified later on. Again, $x_i -x_j$ should be interpreted as difference on the torus (or $f_\ell$ should be replaced with its periodic extension). Note that a similar 
trial state has been used in \cite{LSY}. However, for us the wave function $\Phi_N$ serves a completely different purpose (in our analysis, $\Phi_N$ carries correlations on length scales larger than $\ell$; in \cite{LSY}, on the other hand, it was a product state, describing the condensate trapped in an external potential).

We compute
\[ \begin{split} 
\frac{-\Delta_{x_j} \Psi_N (x_1, \dots , x_N) }{\prod_{i<j}^N f_\ell (x_i - x_j) } = \; &\Big[ -\Delta_{x_j}  - 2 \sum_{i \not = j}^N \frac{\nabla f_\ell (x_j - x_i)}{f_\ell (x_j - x_i)} \cdot \nabla_{x_j} \Big] \Phi_N (x_1, \dots , x_N) \\ &+  \sum_{i \not = j}^N \frac{-\Delta f_\ell (x_j - x_i)}{f_\ell (x_j - x_i)}  \Phi_N (x_1, \dots , x_N) \\ &- \sum^N_{i,m,j} \frac{\nabla f_\ell (x_j - x_i)}{f_\ell (x_j. -x_i)} \cdot \frac{\nabla f_\ell (x_j - x_m)}{f_\ell (x_j - x_m)}  \Phi_N (x_1, \dots , x_N)
\end{split} \]
where the sum in the last term runs over $i,j, m \in \{1, \dots , N \}$ all different. Noticing that the operator on the first line is the Laplacian with respect to the measure defined by (the square of) the Jastrow factor, and using (\ref{eq:fell}) in the second line, we conclude that
\begin{equation}\label{eq:en-psi} \begin{split}  \langle \Psi_N , &\sum_{j=1}^N -\Delta_{x_j} \Psi_N \rangle \\ = \; &\sum_{j=1}^N \int |\nabla_{x_j} \Phi_N ({\bf x})|^2    \prod_{n<m}^N f^2_\ell (x_n - x_m)  d{\bf x} \\ &+ \sum_{i<j}^N 2\lambda_\ell \int \chi_\ell (x_i - x_j) |\Phi_N ({\bf x})|^2 \prod_{n<m}^N f^2_\ell (x_n - x_n	) d{\bf x} \\ &- \sum_{i,j,k} \int \frac{\nabla f_\ell (x_j - x_i)}{f_\ell (x_j - x_i)} \cdot \frac{\nabla f_\ell (x_j - x_k)}{f_\ell (x_j - x_k)} |\Phi_N ({\bf x})|^2 \prod_{n<m}^N f^2_\ell (x_n - x_m) d{\bf x} 
\end{split} \end{equation} 
where we introduced the notation ${\bf x} = (x_1, \dots ,x_N) \in \L^N$. 

\section{Estimating the three-body term} \label{sec:3body}

In the next proposition, we control the last term on the r.h.s. of (\ref{eq:en-psi}). To this end, we need to assume some regularity on the $N$-particle wave function $\Phi_N$, appearing in (\ref{eq:trial}) (we will later make sure that our choice of $\Phi_N$ satisfies these estimates).
\begin{prop} \label{lm:3body} 
Let $N^{-1+\nu} \leq \ell \leq N^{-1/2-\nu}$, for some $\nu > 0$. Suppose $\Phi_N \in L^2_s (\Lambda^N)$ is such that 
\begin{equation}\label{eq:3D-assum} \langle \Phi_N, (1-\Delta_{x_1}) (1- \Delta_{x_2})(1-\Delta_{x_3}) \Phi_N \rangle \leq C \left(1 + \frac{1}{N^2 \ell^3} \right) \end{equation} 
and define $\Psi_N$ as in (\ref{eq:trial}). Then, for every $\delta > 0$, there exists $C > 0$ such that 
\begin{equation}\label{eq:3b-claim} \begin{split}  \Big|  & \frac{1}{\| \Psi_N \|^2} \sum_{i,j,k} \int \frac{\nabla f_\ell (x_j - x_i)}{f_\ell (x_j - x_i)} \cdot \frac{\nabla f_\ell (x_j - x_k)}{f_\ell (x_j - x_k)} |\Phi_N ({\bf x})|^2  \prod_{n<m}^N f^2_\ell (x_n - x_m) d{\bf x} \Big| \\ &\hspace{9cm} \leq C N \ell^{2-\delta} \left(1 + \frac{1}{N^2 \ell^3} \right)\,. \end{split}  \end{equation} 
\end{prop} 

To prove this proposition, we will use the following lemma.  
\begin{lemma} \label{lemma:sobolev_bounds}
Let $W: \bR^3 \to \bR$, with $\mathrm{ supp }\, W \subset [-1/2 ; 1/2]^3$. Then $W$ can be extended to a periodic function (i.e. a function on the torus $\Lambda$) satisfying, on $L^2 (\Lambda) \otimes L^2 (\Lambda)$, the operator inequalities 
\begin{equation*}\label{eq:WL32} 
\begin{split} 
\pm W(x-y) \le\;& C \|W\|_{3/2} \left( 1-\Delta_x \right) \\
\pm W(x-y) \le\;& C \|W\|_{2} \left( 1-\Delta_x \right)^{3/4} 
\end{split} \end{equation*}
for a constant $C>0$, independent on $W$. Moreover, for every $\delta \in [0,1/2)$ there exists $C > 0$ such that
\be \label{eq:W1-L1}	 
\pm W(x-y) \le\; C \|W\|_{1} \Big\{ 1 + \left(-\Delta_x\right)^{3/4+\delta/2} \left(-\Delta_y\right)^{3/4+\delta/2} \Big\} \,.
\ee
Additionally, for any $r > 1$, there exists $C > 0$ such that 
\be \begin{split}
\pm W (x-y) W (x-z) \le\;& C \|W\|^2_r \, \left( 1-\Delta_x \right) \left( 1-\Delta_y\right) \left( 1-\Delta_z \right)
\end{split}\label{eq:Delta3}\,.
\ee
\end{lemma}
\begin{proof} 
The proof is an adaptation to the torus of arguments that are, by now, standard on $\bR^3$. For example, (\ref{eq:W1-L1}) follows by writing, in momentum space
\[ \begin{split} 
\big| \langle \ph, W(x-y) \ph \rangle \big|  &= \Big| \sum_{p_1, p_2, q_1 , q_2 \in \L^*} \widehat{W} (p_1 - q_1) \widehat{\ph} (p_1, p_2) \overline{\widehat{\ph} (q_1, q_2)} \, \delta_{p_1 + p_2, q_1 + q_2} \Big|  \\ &\leq C \| \widehat{W} \|_\infty \sup_{p \in \L^*} \sum_{q \in \L^*_+} \frac{1}{1+|q|^{3/2+ \delta} |p-q|^{3/2+\delta}} \\ &\hspace{3cm} 
\times \big\langle \ph, \big[ 1 + (-\Delta_x)^{3/4+\delta/2} (-\Delta_y)^{3/4+\delta/2}\big]  \ph \big\rangle \\ &\leq C \| W \|_1 \big\langle \ph, \big[ 1 + (-\Delta_x)^{3/4+\delta/2} (-\Delta_y)^{3/4+\delta/2} \big] \ph \big\rangle\,. \end{split} \]
To show (\ref{eq:Delta3}), we proceed similarly, writing 
\[ \begin{split} 
\big| \langle \ph, &W(x-y) W(x-z) \ph \rangle \big| \\ = \; & \Big|\sum \widehat{W} (p_2 - q_2) \widehat{W} (p_3 - q_3) \widehat{\ph} (p_1, p_2,p_3) \overline{\widehat{\ph} (q_1, q_2,q_3)} \, \delta_{p_1 + p_2+p_3 , q_1 + q_2+q_3} \Big| \\
\leq \; &C \sup_p \sum_{q_2, q_3 \in \L^*} \frac{|\widehat{W} (p_2 - q_2)| | \widehat{W} (p_3 - q_3)|}{(1+|p-q_2 - q_3|^2)(1+|q_2|^2) (1+|q_3|^2)}  \\ &\hspace{5cm} \times \langle \ph, (1-\Delta_x) ( 1-\Delta_y) ( 1- \Delta_z) \ph \rangle \\ 
\leq \; &C \| \widehat{W} \|_{r'}^2   \langle \ph, (1-\Delta_x) ( 1-\Delta_y) ( 1- \Delta_z) \ph \rangle \end{split} \]
where $1/r+ 1/r' = 1$ and where we used the bound
\[    \sum_{q_2, q_3 \in \L^*} \frac{1}{(1+|p-q_2 - q_3|^2)^r (1+|q_2|^2)^r (1+|q_3|^2)^r}  \leq C  \]
uniformly in $p$, for any $r > 1$. 
\end{proof}

We are now ready to show Proposition \ref{lm:3body}.
\begin{proof}[Proof of Prop. \ref{lm:3body}] Using the permutation symmetry, $0 \leq f_\ell \leq 1$ and then Lemma \ref{lemma:sobolev_bounds} (in particular, (\ref{eq:Delta3})), the bound (\ref{eq:Lp-norms}) and the assumption (\ref{eq:3D-assum}), we can estimate the numerator in (\ref{eq:3b-claim}) by 
\begin{equation}\label{eq:num-bd} \begin{split} \Big| \sum_{i,j,k} \int \frac{\nabla f_\ell (x_j - x_i)}{f_\ell (x_j - x_i)} &\cdot \frac{\nabla f_\ell (x_j - x_k)}{f_\ell (x_j - x_k)} |\Phi_N ({\bf x})|^2  \prod_{n<m}^N f^2_\ell (x_n - x_m) d{\bf x} \Big| \\ &\leq C N^3 \int |\nabla f_\ell (x_1 - x_2)||\nabla f_\ell (x_1 - x_3)| |\Phi_N  ({\bf x})|^2 d{\bf x} \\ &\leq C N^3 \| \nabla f_\ell \|_r^2 \langle \Phi_N , (1- \Delta_{x_1}) (1-\Delta_{x_2)} (1-\Delta_{x_3}) \Phi_N \rangle \\ &\leq C N \ell^{\frac{6}{r} -4} \Big( 1 + \frac{1}{N^2 \ell^3} \Big) \end{split} \end{equation} 
for any $r > 1$. As for the denominator in (\ref{eq:3b-claim}), we write $u_\ell = 1- f_\ell^2 = 2\o_\ell - \o_\ell^2$, with $\o_\ell$ defined after (\ref{eq:fell}), and we bound (see (\ref{eq:up-jas}) below for a justification of this inequality) 
\[ \prod_{n<m}^N f_\ell^2 (x_n - x_m) \geq 1 - \sum_{n<m}^N u_\ell (x_n - x_m)\,. \]
Using $\| \Phi_N \| = 1$, Lemma \ref{lemma:sobolev_bounds} (in particular, (\ref{eq:W1-L1})), the bound (\ref{eq:Lp-norms}) and again the assumption (\ref{eq:3D-assum}), we arrive at 
\[ \begin{split} \int |\Phi_N ({\bf x})|^2 \prod_{n<m}^N f_\ell^2 (x_n - x_m) d{\bf x} &\geq 1 - \sum_{n<m}^N \int |\Phi_N ({\bf x})|^2 u_\ell (x_n - x_m) d{\bf x} \\ &\geq 1 - C N^2 \| u_\ell \|_1 \langle \Phi_N, (1- \Delta_{x_1}) (1-\Delta_{x_2}) \Phi_N \rangle \\ &\geq 1 - C N \ell^2 \Big( 1 + \frac{C}{N^2 \ell^3} \Big) \geq 1 - C N \ell^2 - \frac{C}{N\ell} \geq 1/2 \end{split} \]
for $N^{-1} \ll \ell \ll N^{-1/2}$.  Combining this estimate with (\ref{eq:num-bd}) and choosing $r > 1$ so that $6/r - 4 > 2 - \delta$, we obtain the desired bound. 
\end{proof}

\section{Reduction to an effective Hamiltonian} \label{sec:reduction}

Let us introduce the notation 
\begin{equation} \label{eq:Ekinpot} \begin{split} E_\text{kin} (\Phi_N) &= \sum_{j=1}^N \int |\nabla_{x_j} \Phi_N (x_1 , \dots , x_N)|^2 \prod_{n<m}^N f_\ell^2 (x_n - x_m) dx_1 \dots dx_N \\ 
E_\text{pot} (\Phi_N) &= \sum_{i<j}^N 2\lambda_\ell \int \chi_\ell (x_i - x_j) |\Phi_N (x_1, \dots , x_N)|^2 \prod_{n<m}^N f_\ell^2 (x_n -x_m) dx_1 \dots dx_N\,. \end{split} \end{equation} 
It follows from (\ref{eq:en-psi}) and Prop. \ref{lm:3body} that 
\begin{equation}\label{eq:exp-Phi} \frac{1}{\| \Psi_N \|^2}  \langle \Psi_N , \sum_{j=1}^N -\Delta_{x_j} \Psi_N \rangle = \frac{1}{\| \Psi_N \|^2} \big[ E_\text{kin} (\Phi_N) + E_\text{pot} (\Phi_N) \big] + \cE \end{equation} 
where $\pm \cE \leq C N \ell^{2-\delta} (1+ 1/ (N^2\ell^3))$, provided $\Phi_N$ satisfies (\ref{eq:3D-assum}). 

The goal of this subsection is to rewrite the main term on the r.h.s. of (\ref{eq:exp-Phi}) as the expectation, in the state $\Phi_N \in L^2_s (\Lambda^N)$, of an effective $N$-particle Hamiltonian having the form 
\begin{equation}\label{eq:Heff}
H_N^\text{eff} = \sum_{j=1}^N - \Delta_{x_j} + 2 \sum_{i<j}^N \nabla_{x_j} \cdot u_\ell (x_i - x_j) \nabla_{x_j} + 2 \sum_{i<j}^N \lambda_\ell \chi_\ell (x_i - x_j) f_\ell^2 (x_i - x_j) 
\end{equation} 
where $u_\ell = 1 - f_\ell^2$. To achieve this goal, we will make use of the following regularity bounds on the wave function $\Phi_N$ (when we will define $\Phi_N$ in the next sections, we will prove that it satisfies these estimates):  
\begin{equation}
\label{eq:Delta-phi} 
\begin{split} 
\langle \Phi_N, (-\Delta_{x_1}) \Phi_N \rangle &\leq \frac{C}{N\ell}, \\ \langle \Phi_N, (-\Delta_{x_1})(-\Delta_{x_2}) \Phi_N \rangle &\leq \frac{C}{N^2 \ell^3} , \\ \langle \Phi_N, (-\Delta_{x_1})(- \Delta_{x_2})(- \Delta_{x_3})  \Phi_N \rangle &\leq \frac{C}{N^3 \ell^4} \\
\langle \Phi_N, (-\Delta_{x_1})(- \Delta_{x_2})(- \Delta_{x_3}) (-\Delta_{x_4})  \Phi_N \rangle &\leq \frac{C}{N^4 \ell^6} \\
\langle \Phi_N, (-\Delta_{x_1})^{3/4+\delta}  (-\Delta_{x_2})^{3/4+\delta} \dots (-\Delta_{x_n})^{3/4+ \delta} \Phi_N \rangle &\leq \frac{C}{N^n \ell^{\alpha_n}} \\
\langle \Phi_N, (-\Delta_{x_1}) (- \Delta_{x_2})^{3/4+\delta}  (-\Delta_{x_3})^{3/4+\delta}  \dots (-\Delta_{x_n})^{3/4+\delta} \Phi_N \rangle &\leq \frac{C}{N^n \ell^{\beta_n}} \\
\end{split} \end{equation} 
for all $n \leq 6$ and $\delta > 0$ small enough and for sequences $\alpha_n, \beta_n$ defined by $\alpha_n = (7/6+\delta) n -(4/9)(1- (-1/2)^n)$ and $\beta_n = \alpha_n + 1/2 -\delta$. In applications (in particular, in Prop.~\ref{prop:eff} below) we will only need the last two bounds in (\ref{eq:Delta-phi}) for $n = 2,4,6$ and, respectively, for $n = 3,4,5$. The relevant values of $\alpha_n, \beta_n$ are given by: $\alpha_2 = 2+2\delta$, $\alpha_4 = 17/4+4\delta$, $\alpha_6 = 105/16+6\delta$, $\beta_3 = 7/2+2\delta$, $\beta_4= 19/4+3\delta$, $\beta_5 =47/8+4\delta$.  
 
\begin{prop} \label{prop:eff} 
Consider a sequence $\Phi_N \in L^2_s (\Lambda^N)$ of normalized wave functions, satisfying  the bounds (\ref{eq:Delta-phi}) and such that $\langle \Phi_N, H_N^\text{eff} \Phi_N \rangle \leq 4 \pi \frak{a} N + C$, for a constant $C > 0$ (independent of $N$), and for all $N$ large enough. Suppose $N^{-1 + \nu} \leq \ell \leq N^{-3/4-\nu}$, for some $\nu > 0$. Then, there exists $C, \eps> 0$ such that 
\begin{equation} \label{eq:red-eff} 
\begin{split} &\frac{1}{\| \Psi_N \|^2} \big[ E_{\rm{kin}} (\Phi_N) + E_{\rm{pot}} (\Phi_N) \big] \leq \langle \Phi_N , H_N^\text{eff} \Phi_N \rangle  - \frac{N(N-1)}{2}\\
	&\qquad\times \Big\langle \Phi_N , \Big\{ \big[ H_{N-2}^\text{eff} - 4 \pi \frak{a} N \big] \otimes u_{\ell} (x_{N-1} - x_N) \Big\} \Phi_N \Big\rangle + C N^{-\eps} \,.\end{split}  \end{equation}
\end{prop} 

\textit{Remark}: we will later prove a lower bound for $H_{N-2}^\text{eff} - 4 \pi \frak{a} N$ which will allow us to show that the second term on the r.h.s. of  (\ref{eq:red-eff}) is negligible, in the limit $N \to \infty$.

\begin{proof}
Writing again $u_\ell = 1- f_\ell^2$, we can estimate 
\begin{equation}\label{eq:up-jas} 
\begin{split} 
\prod_{i<j}^N f_\ell^2 (x_i - x_j) &\geq 1- \sum_{i<j} u_\ell (x_i - x_j)  \\
\prod_{i<j}^N f_\ell^2 (x_i - x_j) &\leq 1 - \sum_{i<j} u_\ell (x_i - x_j) + \frac{1}{2} \sum_{\substack{i<j; \; m<n :\\ (i,j) \not = (m,n)}} u_\ell (x_i - x_j) u_\ell (x_m - x_n) \,.
\end{split} 
\end{equation} 
These bounds follow by setting $h(s) = \prod_{i<j}^N (1 - s u_\ell (x_i -x_j))$, for $s \in [0;1]$, and by proving that \[ h' (s) \geq - \sum_{i<j}^N u_\ell (x_i -x_j), \qquad h'' (s) \leq \sum_{\substack{i<j; m<n :\\ (i,j) \not = (m,n)}}  u_\ell (x_i - x_j) u_\ell (x_m - x_n) \] 
for all $s \in (0;1)$. Thus, we obtain the upper bound  
\begin{equation} \label{eq:Ekin1} \begin{split} E_\text{kin} &(\Phi_N) \\ \leq \; &N \int |\nabla_{x_1} \Phi_N (\textbf{x})|^2 d{\bf x} - N \int |\nabla_{x_1} \Phi_N ({\bf x})|^2 \sum_{i<j} u_\ell (x_i - x_j)  \, d {\bf x}  \\ &+ \frac{N}{2} \int |\nabla_{x_1} \Phi_N ({\bf x})|^2 \sum_{\substack{i<j ;\; m<n :\\ (i,j) \not = (m,n)}} u_\ell (x_i - x_j) u_\ell (x_m - x_n) \, d{\bf x}  \\ = \; &N \int |\nabla_{x_1} \Phi_N ({\bf x}) |^2 (1- (N-1) u_\ell (x_1 - x_2)) d{\bf x} \\ &- \frac{N (N-1)(N-2)}{2} \int |\nabla_{x_1} \Phi_N ({\bf x})|^2 (1 - (N-3) u_\ell (x_1 - x_2)) u_\ell (x_3 - x_4) \, d{\bf x} \\ &+ \cE_\text{kin} \end{split} \end{equation} 
where 
\begin{equation}\label{eq:cEkin} \begin{split} \cE_\text{kin} \leq  \; &C N^3 \int |\nabla_{x_1} \Phi_N ({\bf x})|^2  u_\ell (x_1 - x_2) u_\ell (x_1-x_3) d{\bf x}\\ &+C N^3 \int |\nabla_{x_1} \Phi_N ({\bf x})|^2  u_\ell (x_1 - x_2) u_\ell (x_2 - x_3) d{\bf x} \\ &+ C N^4 \int |\nabla_{x_1} \Phi_N ({\bf x})|^2 u_\ell (x_2 - x_3) u_\ell (x_2 - x_4) d{\bf x}  \\ &+ C N^5 \int |\nabla_{x_1} \Phi_N ({\bf x})|^2 u_\ell (x_2 - x_3) u_\ell (x_4 - x_5) d{\bf x} \\ = \; &\cE_\text{kin}^{(1)} + \cE_\text{kin}^{(2)} + \cE_\text{kin}^{(3)}+ \cE_\text{kin}^{(4)} \,.\end{split} \end{equation} 
Consider the first error term on the r.h.s. of (\ref{eq:cEkin}). Writing $\frak{p} = |\ph_0 \rangle \langle \ph_0|$ for the orthogonal projection onto the condensate wave function $\ph_0 (x) \equiv 1$, $\frak{p}_j$ for $\frak{p}$ acting on the $j$-th particle and $\frak{q}_j = 1 - \frak{p}_j$, we find  
\begin{equation}\label{eq:Ekin2} \begin{split} \cE_\text{kin}^{(1)} = \; &C N^3 \langle \nabla_{x_1} \Phi_N , u_\ell (x_1 - x_2) u_\ell (x_1 - x_3) \nabla_{x_1} \Phi_N \rangle \\ \leq \; &C N^3 \langle \nabla_{x_1} \frak{q}_3 \Phi_N , u_\ell (x_1 - x_2) u_\ell (x_1 - x_3) \nabla_{x_1} \frak{q}_3 \Phi_N \rangle \\ &+ C N^3 \| u_\ell \|_1  \langle \nabla_{x_1} \frak{p}_3 \Phi_N , u_\ell (x_1 - x_2)  \nabla_{x_1} \frak{p}_3 \Phi_N \rangle
\\ 
\leq \; &C N^3 \langle \nabla_{x_1} \frak{q}_2 \frak{q}_3 \Phi_N , u_\ell (x_1 - x_2) u_\ell (x_1 - x_3) \nabla_{x_1} \frak{q}_2 \frak{q}_3 \Phi_N \rangle \\ &+ C N^3 \| u_\ell \|_1  \langle \nabla_{x_1} \frak{q}_2 \frak{p}_3 \Phi_N , u_\ell (x_1 - x_2)  \nabla_{x_1} \frak{q}_2 \frak{p}_3 \Phi_N \rangle \\ &+ C N^3 \| u_\ell \|^2_1  \| \nabla_{x_1} \frak{p}_2 \frak{p}_3 \Phi_N \|^2 \, . \end{split} \end{equation} 
With Lemma \ref{lemma:sobolev_bounds}, and observing that, on the range of $\frak{q}$, $(1-\Delta) \leq - C \Delta$, we obtain 
\[  \begin{split} \cE_\text{kin}^{(1)} \leq \; &C N^3 \| u_\ell \|_{3/2}^2 \langle \Phi_N, (-\Delta_{x_1})(-\Delta_{x_2}) ( -\Delta_{x_3}) \Phi_N \rangle \\ &+ C N^3 \| u_\ell \|_1 \| u_\ell \|_{3/2} \langle \Phi_N, (-\Delta_{x_1})(-\Delta_{x_2}) \Phi_N \rangle + C N^3 \| u_\ell \|_1^2 \langle \Phi_N, (-\Delta_{x_1}) \Phi_N \rangle\, . \end{split} \]
The term $\cE^{(2)}_\text{kin}$ can be treated like $\cE^{(1)}_\text{kin}$. 
Proceeding analogously, we also find, with (\ref{eq:Delta3}), 
\[ \begin{split}  
\cE_\text{kin}^{(3)} \leq \; &C N^4 \| u_\ell \|_r^2 \langle \Phi_N , (-\Delta_{x_1}) (-\Delta_{x_2}) (-\Delta_{x_3}) (-\Delta_{x_4})  \Phi_N \rangle \\ &+ C N^4 \| u_\ell \|_1^2 \langle \Phi_N , (-\Delta_{x_1}) \big[ 1 + (\Delta_{x_2} \Delta_{x_3})^{3/4+\delta} \big] \Phi_N \rangle \end{split} \]
for any $r > 1$, and 
\[ \begin{split} 
\cE_\text{kin}^{(4)} \leq \; &CN^5 \| u_\ell \|_1^2  \big\langle \Phi_N, (-\Delta_{x_1}) \big[ 1 + (\Delta_{x_2} \Delta_{x_3})^{3/4+\delta} + (\Delta_{x_2} \Delta_{x_3} \Delta_{x_4} 
\Delta_{x_5})^{3/4+\delta} \big] \Phi_N \big\rangle \,.
\end{split} \]
From $u_\ell = 1-f_\ell^2 = 2\o_\ell - \o_\ell^2$, we obtain $0 \leq u_\ell \leq 2 \o_\ell$ and thus, with (\ref{eq:Lp-norms}), 
\begin{equation}\label{eq:u-est} \| u_\ell \|_r \leq C \ell^{\frac{3}{p}-1} / N
\end{equation} 
for any $p \geq 1$. From the assumption (\ref{eq:Delta-phi}), we find   
\[ \begin{split} \cE_\text{kin}^{(1)} , \cE_\text{kin}^{(2)} \leq \frac{C}{N^2 \ell^2}, \quad \cE_\text{kin}^{(3)} \leq \frac{C}{N^2 \ell^{2 +6 (1-1/r)}} , \quad \cE^{(4)}_\text{kin} \leq C N^2 \ell^3 + C \ell^{1/2-2\delta} + \frac{C}{N^2 \ell^{15/8+4\delta}}\,. \end{split} \] 
Choosing $\delta > 0$ sufficiently small and $r > 1$ sufficiently close to $1$, we conclude that there exist $C, \eps > 0$ such that $\cE_\text{kin} \leq C N^{-\eps}$, if $N^{-1+\nu} \leq \ell \leq N^{-2/3- \nu}$ for a $\nu > 0$, and $N \in \bN$ is large enough. 

Let us now consider the potential energy. From (\ref{eq:Ekinpot}), we can estimate 
\[  E_\text{pot} (\Phi_N) \leq N (N-1) \lambda_\ell \int \chi_\ell (x_1 - x_2) f_\ell^2 (x_1 - x_2) |\Phi_N ({\bf x})|^2 \prod_{3 \leq i<j}^N f_\ell^2 (x_i - x_j) d {\bf x} \, . \]
With (\ref{eq:up-jas}) (applied now to the product over $3 \leq i < j$), we obtain  
\begin{equation} \label{eq:Epot1} \begin{split} E_\text{pot} &(\Phi_N) \\ \leq \; &N (N-1) \lambda_\ell \int \chi_\ell (x_1 - x_2) f_\ell^2 (x_1 - x_2) |\Phi_N ({\bf x})|^2 d{\bf x} \\ &-  \frac{N (N-1)(N-2)(N-3)}{2} \lambda_\ell \int \chi_\ell (x_1 - x_2) f_\ell^2 (x_1- x_2) |\Phi_N ({\bf x})|^2 u_\ell (x_3 - x_4) d {\bf x} \\ &+ \cE_\text{pot} \end{split} \end{equation} 
where
\[ \begin{split} \cE_\text{pot} \leq \;&CN^6 \lambda_\ell \int \chi_\ell (x_1 - x_2) f_\ell^2 (x_1 - x_2) |\Phi_N ({\bf x})|^2 u_\ell (x_3 - x_4) u_\ell (x_5 - x_6) d{\bf x} \\
&+CN^5 \lambda_\ell \int \chi_\ell (x_1 - x_2) f_\ell^2 (x_1 - x_2) |\Phi_N ({\bf x})|^2 u_\ell (x_3 - x_4) u_\ell (x_4 - x_5) d{\bf x}  \\ = \;& \cE_\text{pot}^{(1)} + \cE_\text{pot}^{(2)} \,.\end{split} \]
Proceeding similarly to (\ref{eq:Ekin2}) (introducing the projections $\frak{p}_j, \frak{q}_j$), we can bound 
\[ \begin{split} 
\cE_\text{pot}^{(1)} \leq \; &CN^6 \lambda_\ell \| \chi_\ell \|_1 \| u_\ell \|_1^2 \Big[ 1 + \langle \Phi_N , (\Delta_{x_1} \Delta_{x_2})^{3/4+\delta} \Phi_N \rangle + \langle \Phi_N , (\Delta_{x_1} \dots \Delta_{x_4})^{3/4+\delta} \Phi_N \rangle \\ & \hspace{8cm} + \langle \Phi_N , (\Delta_{x_1} \dots \Delta_{x_6})^{3/4+\delta} \Phi_N \rangle \Big] \\
\cE_\text{pot}^{(2)} \leq \; &CN^5 \lambda_\ell  \| \chi_\ell \|_1 \| u_\ell \|_1 \| u_\ell \|_{3/2}  \\ &\hspace{.5cm} \times \Big[ \langle \Phi_N, (-\Delta_{x_1}) (\Delta_{x_2} \dots \Delta_{x_5})^{3/4+\delta} \Phi_N \rangle + \langle  \Phi_N, (-\Delta_{x_1}) (\Delta_{x_2} \Delta_{x_3})^{3/4+\delta} \Phi_N \rangle \Big] \\ &+ CN^5 \lambda_\ell \| \chi_\ell \|_1 \| u_\ell \|_1^2 \Big[ 1 + \langle \Phi_N , (\Delta_{x_1} \Delta_{x_2})^{3/4+\delta} \Phi_N \rangle + \langle \Phi_N , (\Delta_{x_1} \dots \Delta_{x_4})^{3/4+\delta} \Phi_N \rangle \Big] \,.
\end{split} 
\] 
From Lemma \ref{lm:hardcorescatt}, we have $\lambda_\ell \leq C/(N\ell^3)$. From the assumption  (\ref{eq:Delta-phi}) and from (\ref{eq:u-est}), we obtain
\[ \begin{split} \cE_\text{pot}^{(1)} &\leq C \Big[ N^3 \ell^4 + N \ell^{2-\delta} + \frac{1}{N \ell^{1/4 + \delta}} + \frac{1}{N^3 \ell^{41/16+ 6\delta}} \Big]  \\ \cE_\text{pot}^{(2)} &\leq C \Big[ N^2 \ell^4 + \ell^{2-\delta} + \frac{1}{N^2 \ell^{1/4 + 4\delta}} + \frac{1}{N \ell^{1/2 +\delta}} + \frac{1}{N^3 \ell^{23/8+4\delta}} \Big]\,. \end{split}  
 \] 
Thus, choosing $\delta > 0$ small enough, we can find $C, \eps >0$ such that $\cE_\text{pot} \leq C N^{-\eps}$, if $N^{-1+\nu} \leq \ell \leq N^{-3/4-\nu}$ for a $\nu > 0$, and $N \in \bN$ is large enough. 

Finally, we consider the denominator on the r.h.s. of (\ref{eq:exp-Phi}). With the lower bound in (\ref{eq:up-jas}) (and the assumption $\| \Phi_N \|_2 = 1$), we find 
\[ \int |\Phi_N ({\bf x})|^2 \prod_{i<j}^N f_\ell^2 (x_i - x_j) d{\bf x} \geq 1 - \frac{N(N-1)}{2} \int u_\ell (x_1 - x_2) |\Phi_N ({\bf x})|^2 d{\bf x} \,.\]
Observing that, by (\ref{eq:W1-L1}), (\ref{eq:u-est}) and by the assumption (\ref{eq:Delta-phi}),  
\[ \begin{split} \frac{N(N-1)}{2} \int |\Phi_N & ({\bf x})|^2 u_\ell (x_1 - x_2) d{\bf x} \\ &\leq C N^2 \| u_\ell \|_1 \big[ 1 + \langle \Phi_N, (\Delta_{x_1} \Delta_{x_2})^{3/4+\delta} \Phi_N \rangle \big] \leq C \Big[ N \ell^2 + \frac{1}{N \ell^{2\delta}} \Big] \end{split}  \]
we conclude, choosing $\delta > 0$ sufficiently small and recalling that $\ell \leq N^{-3/4-\nu}$, that 
\[ \frac{1}{ \int |\Phi_N ({\bf x})|^2 \prod_{i<j}^N f_\ell^2 (x_i - x_j) d{\bf x}} \leq 1 + \frac{N(N-1)}{2} \int u_\ell (x_1 - x_2) |\Phi_N ({\bf x})|^2 d{\bf x} + C N^{-1-\eps} \]
for $\eps > 0$ small enough. Combining the last equation with (\ref{eq:Ekin1}), (\ref{eq:Epot1}) we arrive at (recall the assumption $\langle \Phi_N, H_N^\text{eff} \Phi_N \rangle \leq 4 \pi \frak{a} N + C$) 
\[  \begin{split} &\frac{1}{\| \Psi_N \|^2} \big[ E_\text{kin} (\Phi_N) + E_\text{pot} (\Phi_N) \big] \\ &\hspace{.3cm} \leq \Big[  1 + \frac{N(N-1)}{2} \int u_\ell (x_1 - x_2) |\Phi_N ({\bf x})|^2 d{\bf x} + C N^{-3/2}  \Big] \\ &\hspace{2cm} \times  \Big[ \langle \Phi_N, H_N^\text{eff} \Phi_N \rangle - \frac{N(N-1)}{2} \langle \Phi_N , \big[ H_{N-2}^\text{eff} \otimes u_\ell (x_{N-1} - x_N) \big]  \Phi_N \rangle + C N^{-\eps}  \Big] \\
&\hspace{.3cm} \leq \langle \Phi_N , H_N^\text{eff} \Phi_N \rangle + \frac{N(N-1)}{2}  \langle \Phi_N , H_N^\text{eff} \Phi_N \rangle \int |\Phi_N ({\bf x})|^2 u_\ell (x_{N-1} - x_N) d{\bf x}\\ &\hspace{2cm}  - \frac{N(N-1)}{2} \langle \Phi_N, \big[ H_{N-2}^\text{eff} \otimes u_\ell (x_{N-1} - x_N) \big] \Phi_N \rangle + C N^{-\eps} \\
&\hspace{.3cm} \leq \langle \Phi_N , H_N^\text{eff} \Phi_N \rangle -  \frac{N(N-1)}{2} \Big\langle \Phi_N, \Big\{ \big[ H_{N-2}^\text{eff} - 4 \pi \frak{a} N \big] \otimes u_\ell (x_{N-1} - x_N) \Big\} \Phi_N \Big\rangle + C N^{-\eps} \,. \end{split} \]
\end{proof}

\section{Properties of the effective Hamiltonian} 
\label{sec:GN}

Motivated by the results of the last sections, in particular by (\ref{eq:en-psi}), by Prop.  \ref{lm:3body} and by Prop. \ref{prop:eff}, we would like to choose $\Phi_N \in L^2_s (\Lambda^N)$ as a good trial state for the effective Hamiltonian $H_N^\text{eff}$ defined in (\ref{eq:Heff}) (i.e. $\Phi_N$ should lead to a small expectation of $H_N^\text{eff}$ and, at the same time, it should satisfy the bounds (\ref{eq:Delta-phi})). Since $u_\ell = 1 - f_\ell^2$ is small, unless particles are very close, we can think of $H_N^\text{eff}$ as a perturbation of 
\begin{equation}\label{eq:Hbeta} H_{\ell, N} = \sum_{j=1}^N -\Delta_{x_j} + 2\lambda_\ell \sum_{i<j}^N \chi_\ell (x_i - x_j) \,.\end{equation} 
Keeping in mind that, by (\ref{eq:lambdaell-exp}), $\lambda_\ell \simeq 3\frak{a} / N \ell^3$ and that $1/N \ll \ell \ll 1$, (\ref{eq:Hbeta}) looks like the Hamilton operator of a Bose gas in an intermediate scaling regime, interpolating between mean-field and Gross-Pitaevskii limits. The validity of Bogoliubov theory in such regimes has been recently established in \cite{BBCS2}. The goal of this section is to apply the strategy of \cite{BBCS2} to the Hamilton 
operator (\ref{eq:Heff}). This will lead to bounds for the operator $H_N^\text{eff}$ and, eventually, to an ansatz for $\Phi_N$. While part of our analysis in this section can be taken over from \cite{BBCS2}, we need additional work to control the effect of the difference $u_\ell = 1- f_\ell^2$, appearing in the kinetic and the potential energy in the effective Hamiltonian (\ref{eq:Heff}).  

To determine the spectrum of (\ref{eq:Heff}), it is useful to factor out the condensate and to focus instead on its orthogonal excitations. To this end, following \cite{LNSS}, we define a unitary map $U_N : L^2_s (\Lambda^N) \to \cF^{\leq N}_+ = \bigoplus_{n=0}^N L^2_\perp (\Lambda)^{\otimes_s n}$, requiring that 
\begin{equation}\label{eq:UNdef} U_N \psi = \{ \alpha_0, \alpha_1, \dots , \alpha_N \} \in  \cF^{\leq N}_+ \end{equation} 
if 
\[ \psi = \alpha_0 \ph_0^{\otimes N} + \alpha_1 \otimes_s \ph_0^{\otimes (N-1)} + \dots + \alpha_N \,. \]
Here $\ph_0 (x) \equiv 1$ for all $x \in \Lambda$ denotes the condensate wave function, and $L^2_\perp (\Lambda)$ is the orthogonal complement of $\ph_0$ in $L^2 (\Lambda)$. The action of the unitary operator $U_N$ is determined by the rules 
\begin{equation}\label{eq:U-rules}
\begin{split} 
U_N \, a^*_0 a_0 \, U_N^* &= N- \cN_+  \\
U_N \, a^*_p a_0 \, U_N^* &= a^*_p \sqrt{N-\cN_+ } = \sqrt{N} b_p^* \\ 
U_N \, a^*_0 a_p \, U_N^* &= \sqrt{N-\cN_+ } \, a_p = \sqrt{N} b_p \\
U_N \, a^*_p a_q \, U_N^* &= a^*_p a_q \, . 
\end{split} \end{equation}
where $\cN_+$ denotes the number of particles operator on $\cF^{\leq N}_+$ (it measures therefore the number of excitations of the condensate) and where we introduced modified creation and annihilation operators $b^*_p , b_p$ satisfying the commutation relations  
\begin{equation}\label{eq:comm-bp} 
[ b_p, b_q^* ] = \left( 1 - \frac{\cN_+}{N} \right) \delta_{p,q} - \frac{1}{N} a_q^* a_p \, , \qquad [ b_p, b_q ] = [b_p^* , b_q^*] = 0 \end{equation} 
and 
\begin{equation} \label{eq:commaa}
[ a_r^* a_s , b^*_p ] = \delta_{p,s} b_r^* , \qquad [ a_r^* a_s , b_p ] = - \delta_{r,p} b_s   \, . 
\end{equation}
On the truncated Fock space $\cF_+^{\leq N}$, we can define the excitation Hamiltonian $\cL^\text{eff}_N = U_N H_N^\text{eff} U_N^*$. To compute $\cL^\text{eff}_N$, we first rewrite (\ref{eq:Heff}) in momentum space, using the formalism of second quantization, as  
\begin{equation} \label{eq:HNeff-F} \begin{split}
	H_N^\text{eff} =\;& \sum_{p \in \Lambda^*} p^2 a_p^* a_p - \hskip -0.1cm \sum_{p,q ,r \in \L^*} p \cdot (p+r) \widehat{u}_\ell (r) a_{p+r}^* a_{q-r}^* a_p a_q \\
	&+ \lambda_\ell \sum_{p,q,r \in \L^*} \widehat{\chi_\ell f_\ell^2} (r) a_{p+r}^* a_q^* a_{q+r} a_p\,.  \end{split} \end{equation}
Then, we apply (\ref{eq:U-rules}). This will produce a constant term, as well as contributions that are quadratic, cubic and quartic in (modified) creation and annihilation operators. Following Bogoliubov's method, we would like to eliminate cubic and quartic terms. This would reduce $\cL_N^\text{eff}$ to a quadratic expression, whose spectrum could be computed through diagonalization with a (generalized) Bogoliubov transformation. As explained in \cite{BBCS2}, though, cubic and quartic terms in $\cL^\text{eff}_N$ are not negligible (they contribute to the energy to order $\ell^{-1}$). Before proceeding with the diagonalization, we need to extract relevant contributions to the energy from cubic and quartic terms. As in \cite{BBCS2}, we do so by conjugating $\cL_N^\text{eff}$ with a (generalized) Bogoliubov transformation removing short-distance correlations characterising low-energy states. To reach this goal, we fix $\ell_0 \gg \ell$, small, but of order one, independent of $N$. Similarly as in (\ref{eq:ev}), we define $f_{\ell_0}$ to be the ground state solution of the Neumann problem for the hard sphere potential in the ball $B_{\ell_0}$. Extending $f_{\ell_0}$ to the box $\Lambda$, we find 
\[ -\Delta f_{\ell_0} (x) = \lambda_{\ell_0} \chi_{\ell_0} (x) f_{\ell_0} (x) \]
with $f_{\ell_0} (x) = 0$ for $|x| = \frak{a}/N$ (the eigenvalue $\lambda_{\ell_0}$ is approximately given by (\ref{eq:lambdaell-exp}), of course with $\ell$ replaced by $\ell_0$). For $\frak{a}/N \leq |x| \leq \ell_0$, we can then define $g_{\ell_0} (x) = f_{\ell_0} (x) / f_\ell (x)$. We can also extend $g_{\ell_0}$ to $\Lambda$, setting $g_{\ell_0} (x) = \lim_{|y| \downarrow \aa / N} g_{\ell_0} (y)$ for all $|x| \leq \aa /N$ and $g_{\ell_0} (x) = 1$ for all $x \in \L \backslash B_{\ell_0}$. A  simple computation shows that $g_{\ell_0}$ solves the equation 
\begin{equation}\label{eq:g0} 
-\nabla \big[ f_\ell^2 \nabla g_{\ell_0} \big] + \lambda_\ell \chi_\ell f_\ell^2 g_{\ell_0} = \lambda_{\ell_0} \chi_{\ell_0} f_\ell^2 g_{\ell_0} 
\end{equation} 
with the Neumann boundary condition $\partial_r g_{\ell_0} (x) = 0$ for $|x| = \ell_0$ (this follows easily from the observation that, for $\ell \leq |x| \leq \ell_0$, $g_{\ell_0} (x) = f_{\ell_0} (x)$). Conversely, it is interesting to observe that, integrating (\ref{eq:g0}) against $g_{\ell_0}$, we find 
\begin{equation}\label{eq:en-g} \int f_\ell^2 |\nabla g_{\ell_0}|^2 dx + \lambda_\ell \int \chi_\ell f_\ell^2 g_{\ell_0}^2 dx = \lambda_{\ell_0} \int \chi_{\ell_0} f_\ell^2 g_{\ell_0}^2 dx \,.
\end{equation} 
With (\ref{eq:ev}), we find 
\begin{equation}\label{eq:en-g2}  \int |\nabla (f_\ell g_{\ell_0})|^2  dx = \lambda_{\ell_0} \int \chi_{\ell_0} |(f_\ell g_{\ell_0})|^2 dx \end{equation} 
which implies that (\ref{eq:g0}) is solved by $g_{\ell_0} = f_{\ell_0}/f_\ell$. 

With $g_{\ell_0}$, we define $\check{\eta} (x) : = - N (1 - g_{\ell_0} (x))$. Some properties of $g_{\ell_0}, \check{\eta}$ and of their Fourier coefficients are collected in the next lemma, whose proof is deferred to Appendix~\ref{App:scatt}. We introduce here the notation 
\begin{equation} \label{eq:Vell}
 V_\ell (x) = 2 N \lambda_\ell  \chi_\ell (x) f_\ell^2 (x)\,.
\end{equation} 
\begin{lemma}\label{lm:driftscatt}
We have $\check{\eta} (x) = 0$ for $|x| \geq \ell_0$. For $|x| \leq \ell_0$, we have the bounds 
\be \label{eq:checketa}
|\check \eta(x)| \leq \frac{ C \aa}{|x| +\ell} \,, \qquad |\nabla \check{\eta}(x)| \leq \frac{C \aa}{(|x| + \ell)^2 }\, . 
\ee
Furthermore  
\be \label{eq:intVellg0}
\begin{split} 
\left| \int V_\ell (x) g_{\ell_0} (x) dx - 8 \pi \aa \right| &= \left | 2N \lambda_\ell \int \chi_{\ell} (x) f_\ell^2 (x) g_{\ell_0} (x) dx  - 8 \pi \aa \right|  \leq C N^{-1} \\
\left| \int V_\ell (x) g_{\ell_0} (x) e^{-ip \cdot x} dx \right| &= \left | 2N \lambda_\ell \int \chi_{\ell} (x) f_\ell^2 (x) g_{\ell_0} (x) e^{-ip \cdot x} dx \right|   \leq \frac{C  }{\ell^2 p^2} 
\end{split} \ee
and, analogously,  
\be \begin{split}  \label{eq:intchiell0}
\Big | 2N \l_{\ell_0} \int \chi_{\ell_0} (x)  f_{\ell}^2 (x) g_{\ell_0} (x) dx - 8 \pi \aa \Big|  &  \leq C N^{-1}\\
\Big| 2N \lambda_{\ell_0} \int \chi_{\ell_0} (x) f_\ell^2 (x) g_{\ell_0} (x) e^{-ip \cdot x} dx \Big| &\leq \frac{C }{ p^2}\,.
\end{split}\ee

Recall the definition $u_\ell = 1 - f_\ell^2$. For $p \in \L^*_+$, let 
\begin{equation}\label{eq:Dp-def} D_p = - \sum_{r \in \L^*} p \cdot (p+r) \hat{u}_\ell (r) \eta_{p+r} \end{equation}
and denote by $\eta_p$ the Fourier coefficients of $\check{\eta}$. Then (\ref{eq:g0}) takes the form 
\be  \label{eq:eta-scat0}
p^2 \eta_p +  D_p + N \l_\ell \big( \widehat{\chi_\ell f_\ell^2} * \widehat{g}_{\ell_0} \big) (p) =  N \l_{\ell_0} \big(\widehat{\chi_{\ell_0} f_{\ell}^2} \ast \widehat{g}_{\ell_0} \big)(p)\,.
\ee 
or, equivalently, with the definition (\ref{eq:Vell}), 
\be  \label{eq:eta-scat}
p^2 \eta_p +  D_p + \frac{1}{2} \widehat{V}_\ell (p) + \frac{1}{2N} \big( \widehat{V}_\ell * \eta \big) (p) = N \l_{\ell_0} \big(\widehat{\chi_{\ell_0} f_{\ell}^2} \ast \widehat{g}_{\ell_0} \big)(p)\,.
\ee 
We have
\be\label{eq:eta0}
   \eta_0  = -\frac{2}{5} \pi \aa \ell_0^2 + \cO \Big( \frac{\aa^2\ell_0}{N}\Big)+\cO\big(\aa\ell^2\big) 
\ee
and, for $p \in \L^*_+$, 
\be \label{eq:etap}
\begin{split} 
|\eta_p| &\leq C \min \left \{\frac{1}{|p|^2}\, ;  \frac{1}{\ell^2 |p|^4} \right\}\,,  \\
|D_p| &\leq   C \min \left \{\,\frac{1}{N \ell}\, ;  \frac{1}{\ell^2 |p|^2} \right\} \,.
\end{split} 
\ee
In particular, this implies 
\be  \label{eq:etaHr}
\sum_{q\in \Lambda^*_+} |q|^{r} |\eta_q|^2 \leq C \ell^{1-r} 
\ee
for all $1 < r < 5$.
\end{lemma}
Using the coefficients $\eta_p$, for $p\in \L^*_+$, we define now 
\begin{equation}\label{eq:defB} 
B(\eta) = \frac{1}{2} \sum_{p\in \Lambda^*_+} \eta_p \big( b_p^* b_{-p}^* -  b_p b_{-p} \big)  \, \end{equation}
and we introduce the renormalized excitation Hamiltonian 
\begin{equation}\label{eq:cGNeff}  \cG^\text{eff}_{N,\ell} = e^{-B(\eta)} U_N H_N^\text{eff} U_N^* e^{B(\eta)}\,. \end{equation} 
As explained in \cite{BBCS2}, conjugation with the generalized Bogoliubov transformation 
$e^{B(\eta)}$ models correlations up to scales of order one (determined by the radius $\ell_0$ of the ball used to define $g_{\ell_0}$). It  extracts important contributions to the energy from terms in $\cL_N^\text{eff}$ that are quartic in creation and annihilation operators. This will allow us to approximate $\cG_N^{\text{eff}}$ by the sum of a constant and of a quadratic expression in creation and annihilation operators, whose ground state energy will be computed by simple diagonalization (through a second Bogoliubov transformation). Unfortunately, conjugation with $e^{B(\eta)}$ also produces several error terms, which need to be bounded. For $1 < r <5$, we consider the positive operator 
\begin{equation}\label{eq:cPr} \cP^{(r)} = \sum_{p \in \Lambda^*_+} |p|^r a_p^* a_p \end{equation} 
acting on $\cF^{\leq N}_+$. The growth of $\cP^{(r)}$ (and of products of $\cP^{(r)}$ with moments of the number fo particles operator) under the action of $B(\eta)$ is controlled by the next lemma. 

\begin{lemma}\label{lm:N-Pgrow}
Let $B(\eta)$ be defined as in \eqref{eq:defB}. Then, for every $n \in \bN$ and  $r \in (1; 5)$ there is $C > 0$ such that, for all $t\in [0;1]$, 
\begin{equation}\begin{split} \label{eq:P_r}
e^{- tB(\eta)} (\cN_+ +1)^{n} e^{t B(\eta)} &\leq C (\cN_+ +1)^{n}   \\
e^{-t B(\eta)} \cP^{(r)}(\cN_+ +1)^{n} e^{tB(\eta)} &\leq C  \big( \cP^{(r)} + \ell^{1-r} \big) (\cN_+ +1)^n .
\end{split}\end{equation}
\end{lemma}

\begin{proof} The proof of the first bound in \eqref{eq:P_r} is standard and can be found for example in \cite[Lemma 6.1]{BS}. As for the second inequality, let us consider the case $n=0$. For any $\xi \in \cF_+^{\leq N}$ and $t \in [0;1]$ we write 
	\begin{equation} \label{eq:first_commut_D_r}
		\langle \xi, e^{-tB(\eta)} \mathcal{P}^{(r)} e^{tB(\eta)} \xi \rangle = \langle \xi, \mathcal{P}^{(r)} \xi \rangle +\int_0^t ds\, \langle \xi, e^{-sB (\eta)}\big[\mathcal{P}^{(r)},B(\eta)\big]e^{sB(\eta)} \xi \rangle  
	\end{equation}
where
	\begin{equation*}
			[\mathcal{P}^{(r)},B(\eta)]=\;\sum_{q\in \Lambda^*_+} |q|^{r} \eta_q b^*_{q} b^*_{-q}+\mathrm{h.c.} \, .
	\end{equation*}
By Cauchy-Schwarz's inequality and \eqref{eq:etaHr} we get
	\begin{equation*} 
	\begin{split} 
		\big|\langle \xi, [\mathcal{P}^{(r)},B(\eta)] \xi \rangle \big| &\le C \sum_{q \in \L^*_+} |q|^r |\eta_q| \| a_q \xi \| \| a_q^* \xi \| \\ &\leq C \sum_{q \in \L^*_+} |q|^r |\eta_q| \| a_q \xi \| \big [ \| a_q \xi \| + \| \xi \| \big] \leq C \langle \xi, \mathcal{P}^{(r)} \xi \rangle + \ell^{1-r} \| \xi \|^2  \end{split} 
	\end{equation*} 
Inserting this into \eqref{eq:first_commut_D_r} and using Gronwall's Lemma, we obtain the desired bound. The proof for $n\geq1$ is similar, we omit further details. \end{proof}

With Lemma \ref{lm:N-Pgrow} we are ready to establish the form of $\cG_{N,\ell}^\text{eff}$, up to errors which are negligible on our trial state. We use the notation (recall the definition (\ref{eq:Vell}) of $V_\ell$) 
\be \label{eq:cK-cVN}
\cK= \sum_{p\in \L^*_+} p^2 a^*_p a_p\,, \qquad \text{and}\qquad
\cV_\ell = \frac 1 {2N} \sum_{\substack{p,q \in \L^*_+,\, r \in \L^*:\\ r\neq -p, -q}} \widehat V_\ell(r) a^*_{p+r} a^*_q a_p a_{q+r}\,.
\ee
\begin{prop}  \label{prop:cGcJ} 
 Let $\cG^\text{eff}_{N,\ell}$ be defined as in \eqref{eq:cGNeff}, with $B(\eta)$ as in \eqref{eq:defB}, with $\ell \geq N^{-1+\nu}$ for some $\nu > 0$ and $\ell_0 > 0$ small enough (but fixed, independent of $N$). Let $\cP^{(r)}$ be defined as in (\ref{eq:cPr}). Then, for any $0 < \k < \nu /2$ we have
\be \label{eq:GN-JN-LB}
\cG^\text{eff}_{N,\ell} \geq  4 \pi \aa N - C(\cN_+ +1)  
- \frac{C}{N^\k}\,  \cP^{(2+\k)} (\cN_++1)\,.
\ee 
On the other hand, using the notation $\g_p=\cosh(\eta_p)$ and $\s_p=\sinh(\eta_p)$, let  
\begin{equation} \label{eq:def-CN}
	\begin{split}
	C_{N,\ell}= &\, \frac{(N-1)}{2}\widehat{V}_\ell (0) +\sum_{p\in \L^*_+}\Big[ p^2\sigma_p^2+ \widehat V_\ell(p)(\sigma_p^2+\sigma_p\gamma_p)+\frac 1 {2N}\sum_{q \ne0}\widehat V_\ell(p-q)\eta_p\eta_q +  D_p\eta_p  \Big]
	\end{split}
	\end{equation}
with $D_p$ defined in \eqref{eq:Dp-def}. Denote also
\begin{equation}\label{eq:cQNell} 
\cQ_{N,\ell} =  \sum_{p \in \L^*_+}\left[ F_p a^*_p a_p+\dfrac{1}{2}G_p\left( b^*_p b^*_{-p}+b_pb_{-p}\right) \right]
\end{equation} 
with
	\begin{equation} \label{eq:FpGp}
	\begin{split}
	F_p=\;&p^2(\s_p^2+\g_p^2)+ (\widehat{V_\ell} \ast \widehat g_{\ell_0})(p) (\g_p +\s_p)^2\\
	G_p=\;&2p^2 \g_p \sigma_p +(\widehat{V_\ell} \ast \widehat g_{\ell_0})(p) (\g_p+\s_p)^2+ 2D_p\,.
\end{split}
\end{equation}
Then
	\begin{equation} \label{eq:cG-cJ}
\cG^\text{eff}_{N, \ell} = C_{N,\ell} +\cQ_{N,\ell}+\cE_{N,\ell}
	\end{equation}
where, 
\[  \label{eq:Eps-err}
\pm \cE_{N,\ell} \leq  \frac{C}{\sqrt{N \ell}} \,(\cK + \cV_{\ell} +\cP^{(5/2)})(\cN_++1) \,,
\]
and $\cK$ and $\cV_\ell$ are defined in \eqref{eq:cK-cVN}.
\end{prop}

\begin{proof}  
According to (\ref{eq:HNeff-F}) we can decompose 
\[ \cG^\text{eff}_{N,\ell} = \cG_{N,\ell} + \cJ_{N,\ell} \]
with 
\begin{equation}\label{eq:cGNell} \cG_{N,\ell} = e^{-B(\eta)} U_N \left[ \sum_{p \in \Lambda^*} p^2 a_p^* a_p + \frac{1}{2N} \sum_{p,q,r \in \L^*} \widehat{V}_\ell (r) a^*_{p+r} a_q^* a_{q+r} a_p \right] U_N^* e^{B(\eta)}  \end{equation} 
and 
\begin{equation}\label{eq:cJNell}  \cJ_{N,\ell} = - e^{-B(\eta)} U_N \left[ \sum_{p,q,r \in \L^*} p \cdot (p+r) \widehat{u}_\ell (r) a_{p+r}^* a_{q-r}^* a_p a_q \right] U_N^* e^{B(\eta)}\,. \end{equation} 
We can compute $\cG_{N,\ell}^\text{eff}$ with tools developed in \cite{BBCS2}. From Propositions 7.4 -- 7.7 of \cite{BBCS2}, we obtain, on the one hand, the lower bound
\begin{equation}\label{eq:GN0-fin} \begin{split} \cG_{N,\ell} \geq \; &\frac{(N-1)}{2} \widehat V_\ell(0)   + \sum_{p \in \L^*_+} \big[ p^2 \eta_p  + \widehat V_\ell(p) + \frac{1}{2N} (\widehat{V}_\ell * \eta) (p) \big] \eta_p  \\
& + \sum_{p \in \L^*_+} \big[ p^2 \eta_p + \frac 12 \widehat V_\ell (p) +\frac 1 {2N} (\widehat{V}_\ell * \eta) (p) \big] \big( b_p b_{-p} + b^*_p b^*_{-p} \big) - C (\cN_+ + 1) \end{split} \end{equation} 
and, on the other hand, the approximation 
\be \begin{split} \label{eq:GNl-final}
\cG_{N,\ell} =\;& \frac{(N-1)}{2} \widehat V_\ell(0)   + \sum_{p \in \L^*_+} \big[ p^2 \s_p^2  + \widehat V_\ell(p) \s_p^2 + \widehat V_\ell(p) \g_p \s_p +\frac 1 {2N} \sum_{q \in \L^*_+}\widehat V_\ell(p-q) \eta_p \eta_q \big] \\
&  + \sum_{p \in \L^*_+}   \big[2 p^2 \s_p^2 + \widehat V_\ell(p) (\g_p+\s_p)^2 \big]  b^*_p b_p + \cK + \cV_\ell \\
& + \sum_{p \in \L^*_+} \big[ p^2 \s_p \g_p +\frac 12 \widehat V_\ell(p) (\g_p + \s_p)^2  +\frac 1 {2N} \sum_{q \in \L^*_+}\widehat V_\ell(p-q) \eta_q \big]\big( b_p b_{-p} + b^*_p b^*_{-p} \big) + \cE_\cG
\end{split}\ee
where
\[
\pm \cE_\cG \leq  \frac{C}{\sqrt{N \ell}} \,(\cK + \cV_{\ell} +1)(\cN_++1) \,.
\]
Some care is required  here when we apply results from \cite{BBCS2}. First of all, the interaction potential considered in \cite{BBCS2} has the form $N^{3\beta} W (N^\beta x)$, for some $0 < \beta < 1$. The potential $V_\ell (x) = 2 N \lambda_\ell \chi_\ell (x) f_\ell^2 (x)$ appearing in (\ref{eq:cGNell}) has this form only if we approximate $f_\ell \simeq 1$ and $\lambda_\ell \simeq 3\frak{a} / (N \ell^3)$. A closer inspection to \cite{BBCS2} shows, however, that (\ref{eq:GNl-final}) does not rely on the precise form of the interaction potential but instead only on the bounds 
\[ \label{eq:Vell-eta}
\sup_{q \in \L^*_+}\sum_{\substack{r \in \L^*\\ r \neq -q}} \frac{|\widehat V_\ell(r)|}{|q+r|^2} \leq C\ell^{-1}\,, \qquad  \sum_{\substack{r \in \L^*,\, q \in \L^*_+\\ r \neq -q}} \frac{|\widehat V_\ell(r)|}{|q+r|^2 |q|^2} \leq C\ell^{-2}\,,
\]
which are the analog of \cite[Eq. (7.5) and (7.75)]{BBCS2} and follow from  $\|\widehat V_\ell\|_\io \leq C$ and $\| \widehat V_\ell \|_2 \leq C \ell^{-3/2}$. Moreover, the estimate (\ref{eq:GNl-final}) was proven in \cite{BBCS2} under the assumption that $W = \lambda V$, for a sufficiently small $\lambda > 0$. 
This assumption was used in \cite{BBCS2} to make sure that the $\ell^2$-norm of $\eta$
is sufficiently small. As later shown in \cite{BBCS4}, smallness of $\| \eta \|$ can also be achieved by choosing the parameter $\ell_0$ small enough, with no restriction on the size of the interaction potential \footnote{In \cite{BBCS2}, smallness of the potential was more importantly used to establish Bose-Einstein condensation for low-energy states; here, we do not need to show Bose-Einstein condensation, because we are only interested in an upper bound on the energy.}. Finally, in \cite{BBCS2}, the choice of $\eta$ was slightly different from the definition given after (\ref{eq:g0}) (the presence of the second term on the r.h.s. of (\ref{eq:HNeff-F}) affects the choice of $\eta$, as we will see shortly). However, the derivation of (\ref{eq:GNl-final}) does not depend on the exact form of $\eta$, but rather on bounds, proven in Lemma \ref{lm:driftscatt}, that holds for both choices of $\eta$. This explains why (\ref{eq:GNl-final}) holds true, for sufficiently small values of $\ell_0$. 

Let us now consider (\ref{eq:cJNell}). With (\ref{eq:U-rules}) we find 
\[ U_N \Big[ - \sum_{p,q,r \in \L^*} p \cdot (p+r) \widehat{u}_\ell (r) a_{p+r}^* a_{q-r}^* a_p a_q \Big] U_N^* = Z_1 + Z_2 + Z_3 \]
with 
\begin{equation} \label{eq:Z1-3} \begin{split}
Z_1 =\; &- (N-\cN_+)\, \widehat u_\ell(0) \sum_{p \in \L^*_+} p^2 a^*_p a_p \\
Z_2 = \; &- \sqrt{N}   \sum_{\substack{p, r \in \L^*_+: \\
p+r \neq 0}}   p \cdot (p+r)\, \widehat u_\ell(r) \big(  b^*_{p+r} a^*_{-r} a_{p}  + \text{h.c.} \big) \\
Z_3 =\;  &- \sum_{\substack{r \in \L^*,\, p,q \in \L^*_+: \\
r \neq -p,q}}   p \cdot (p+r)\, \widehat u_\ell(r)  a^*_{p+r} a^*_{q-r} a_{p}a_q \,.
\end{split}\end{equation} 
Using Lemma \ref{lm:hardcorescatt} to bound $\| \widehat{u}_\ell \|_\infty \leq \| u_\ell\|_1 \leq C \| \o_\ell\|_1 \leq C \ell^2 N^{-1}$ and Lemma \ref{lm:N-Pgrow} (in particular, the second inequality in (\ref{eq:P_r}), with $r = 2$), we find 
\[ 
|\langle \xi, e^{-B(\eta)} Z_1 e^{B(\eta)}\xi \rangle| \leq  C  N \| u_\ell\|_1  \| (\cK + 1)^{1/2} e^{B(\eta)}\xi \|^2 \leq C \ell  \| (\cK+1)^{1/2} \xi \|^2 \,,
\]
because $\cN_+ \leq C \cK$. As for the term $Z_{2}$, we have, from $\| u_\ell \|_2 \leq C \| \o_\ell \|_2 \leq C \ell^{1/2} / N$ and by Lemma \ref{lm:N-Pgrow},  
\[  \begin{split}
&|\langle e^{B(\eta)}\xi, Z_{2} e^{B(\eta)}\xi \rangle|\\
& \;\leq \sqrt{N} \Big( \sum_{\substack{p, r \in \L^*_+: \\
p+r \neq 0}}  | p+r|^2 \| a_{p+r}a_{-r} e^{B(\eta)}\xi \|^2  \Big)^{1/2}  \Big( \sum_{\substack{p, r \in \L^*_+: \\
p+r \neq 0}}   |\widehat u_\ell(r)|^2  |p|^2 \,\| a_p e^{B(\eta)}\xi \|^2 \Big)^{1/2} \\
& \;\leq C  \sqrt N \, \| u_\ell\|_2  \| \cK^{1/2} \cN_+^{1/2} e^{B(\eta)}\xi \| \| \cK^{1/2} e^{B(\eta)} \xi \|    \\ & \;\leq C (\ell N)^{-1/2} \| (\cK+1)^{1/2} (\cN_++1)^{1/2}\xi \|^2\,.
\end{split}\]
Hence, we obtain 
\[
\cJ_{N,\ell} = Z_3 + \int_0^1 e^{-tB(\eta)} [Z_3, B(\eta)] e^{t B(\eta)}   dt + \cE_1
\]
with 
\[
\pm \cE_{1} \leq C (N\ell)^{-1/2}  (\cK+1) (\cN_++1)\,.
\]

Using \eqref{eq:commaa} we find 
\[
[Z_3, B(\eta)]  = \sum_{i=1}^3 W_i
\]
with
\[ \begin{split}
W_1 =\; & \sum_{p \in \L^*_+} D_p  \big(b^*_p b^*_{-p} + b_p b_{-p} \big)\\
W_2 =\;   &- \sum_{\substack{r \in \L^*,\, p,q \in \L^*_+: \\
p+r,\, q-r \neq 0}}   p \cdot (p+r)\, \widehat u_\ell(r) \eta_{q-r} \big( b^*_{p} b^*_{r-q} a^*_{q}a_{p+r} +\hc \big) \\
W_3 =\;   &- \sum_{\substack{r \in \L^*,\, p,q \in \L^*_+: \\
p+r,\, q-r \neq 0}}   p \cdot (p+r)\, \widehat u_\ell(r) \eta_{p+r} \big( b^*_{-p-r} b^*_{q} a^*_{p}a_{q-r} +\hc \big) \,.
\end{split}\]
For any $t \in [0;1]$, we have (using again $\| u_\ell \|_2 \leq C \ell^{1/2}/N$ and $\| \eta \|_2 \leq C$, from (\ref{eq:etap})) 
\[ \begin{split}
|\langle & e^{tB(\eta)}\xi, W_2 e^{tB(\eta)}\xi \rangle| \\
& \leq  \Big[ \sum_{\substack{r \in \L^*,\, p,q \in \L^*_+: \\
p+r,\, q-r \neq 0}}  | p|^2 \| a_q b_{r-q}  b_{p} (\cN_+ + 1)^{-1/2} e^{tB(\eta)}\xi\|  \Big]^{1/2}   \\
 & \hskip 0.5cm \times \Big[ \sum_{\substack{r \in \L^*,\, p,q \in \L^*_+: \\
p+r,\, q-r \neq 0}} | \widehat u_\ell(r)|^2 |\eta_{q-r}|^2  |p+r|^2 \|a_{p+r} (\cN_++1)^{1/2}e^{tB(\eta)}\xi\|\Big]^{1/2} \\
& \leq  \| u_\ell\|_2 \| \eta\|_2 \|  \cK^{1/2} (\cN_++1)^{1/2} e^{tB(\eta)}\xi\|^2   \leq C \ell^{1/2}  (N\ell)^{-1}  \|  \cK^{1/2} (\cN_++1)^{1/2}\xi\|^2\,.
\end{split}\]
The contribution of $W_3$ can be bounded similarly. Hence, 
\[ \label{eq:Z4-comm1}
\cJ_{N,\ell}  = Z_3 + W_1 + \int_0^1 dt \int_0^t ds \, e^{-s B(\eta)} \big[ W_1, B(\eta) \big]  e^{s B(\eta)}   + \cE_2
\]
with 
\[
\pm \cE_2  \leq \frac{C}{\sqrt{N\ell}} \, (\cK + 1)  (\cN_++1)\,.
\]

With  \eqref{eq:comm-bp}, we compute 
\[ \big[ W_1 , B(\eta ) \big] = \sum_{i=1}^4 X_i + \text{h.c.} \]
where 
\[ \begin{split}
X_1 =\;&  \sum_{p\in \L^*_+} D_p \eta_p \\
X_2 =\;&  \sum_{p\in \L^*_+} D_p \eta_p  \, \bigg[ \left(1- \frac{\cN_+}N \right)\left(1- \frac{\cN_++1}N   \right) -1 \bigg] \\
X_3 =\;&   2\sum_{p\in \L^*_+} D_p \eta_p\, a^*_p   \left(1 -  \frac{\cN_++2}N \right) \left(1 -  \frac{\cN_++1}N \right)   a_p\\
X_4 =\; & - \frac 1 {N} \sum_{p,q \in \L^*_+} D_p \eta_q  \, a^*_p a^*_{-p} \bigg[ 2 \left(1- \frac{\cN_+}{N} \right) - \frac 3 N \bigg] a_{q} a_{-q} \,.
\end{split}\]
With \eqref{eq:etap}, we find, for any $t \in [0;1]$, 
\[ \begin{split}
|\langle e^{tB(\eta)}\xi, X_2 e^{tB(\eta)}\xi \rangle| 
& \leq \frac C N  \Big[ \sum_{\substack{p \in \L^*_+ \\ |p|\leq \ell^{-1}}} \frac{1}{|p|^2} + \sum_{\substack{p \in \L^*_+ \\ |p|\geq \ell^{-1}}} \frac{1}{\ell^2 |p|^4}  \Big] \| (\cN_+ +1)^{1/2}e^{tB(\eta)}\xi\|^2 \\
&  \leq C (N \ell)^{-1}  \| (\cN_+ + 1)^{1/2}\xi\|^2\,. 
\end{split}\] 
Again from (\ref{eq:etap}), we have $|D_p \eta_p| \leq C (N\ell)^{-1}$ for all $p \in \L^*_+$. Thus 
\[ \begin{split}
|\langle e^{tB(\eta)}\xi, X_3 e^{tB(\eta)}\xi \rangle|  & \leq C (N\ell)^{-1} \| \cN_+^{1/2} e^{t B (\eta)} \xi \|^2 \leq C (N\ell)^{-1} \| (\cN_+ + 1)^{1/2} \xi \|^2 \,.\end{split} \]
As for the expectation of $X_4$, using  \eqref{eq:etap} we obtain  
\[\begin{split}  \label{eq:X4-cNcK}
&|\langle e^{tB(\eta)}\xi, X_4 e^{tB(\eta)}\xi \rangle|  \\
& \leq \frac C N \Big( \sum_{p,q \in \L^*_+} \frac{|D_p|^2}{|p|^2} \| a_{-q} a_q\, e^{tB(\eta)} \xi\|^2 \Big)^{1/2} \\
& \hskip 3cm \times\Big( \sum_{p,q \in \L^*_+} |\eta_q |^2  |p|^2\| a_p a_{-p} \,e^{tB(\eta)}\xi\|^2 \Big)^{1/2} \\
& \leq \frac{C}{N \ell^{1/2}} \| \cK^{1/2}\cN_+^{1/2}e^{tB(\eta)}\xi\| \| \cN_+ e^{tB(\eta)}\xi\|  \leq \frac{C}{N \ell}\, \| (\cK + 1)^{1/2} (\cN_+ + 1) \xi\|^2 \, . 
\end{split}\]
We conclude that 
\[ \cJ_{N,\ell} = Z_3 + W_1 + X_1  + \cE_3 \]
with 
\[
\pm \cE_3  \leq \frac{C}{\sqrt{N\ell}} \, (\cK + 1)  (\cN_++1)\,.
\]

Let us now go back to control the term $Z_3$, as defined in (\ref{eq:Z1-3}). We can estimate, for any $\kappa > 0$,  
\[ \begin{split} \label{eq:Z4final}
 | \langle \xi, Z_3  \xi \rangle | &\leq  C \Big( \sum_{\substack{r \in \L^*,\, p,q \in \L^*_+: \\
p+r,\, q-r \neq 0}}   |p|^{2+\k} \frac{|\widehat \o_\ell(r)|}{|p+r|^\k} \, \| a_{p}a_q \xi\|^2 \Big)^{1/2}\\
& \hskip 1.5cm \times \Big( \sum_{\substack{r \in \L^*,\, p',q' \in \L^*_+: \\
p'-r,\, q'+r \neq 0}}   |p'|^{2+\k} \frac{|\widehat \o_\ell(r)|}{|p'-r|^\k} \, \| a_{p'}a_{q'}\xi\|^2 \Big)^{1/2} \\
& \leq \frac{C}{N^{\k}}\,  \langle \xi, \cP^{(2+\k)} \cN_+ \xi \rangle   \,,
\end{split}\]
where we used the change of variables $p'=p+r, q' = q-r$ and the bound 
\begin{equation}\label{eq:sum-o}
\sup_{p \in \L^*_+} \sum_{r \in \L^*: \,r\neq -p} \frac{|\widehat \o_\ell(r)|}{|p-r|^\k}  \leq C N^{-\k} ,
\end{equation}
valid for any $\kappa > 0$. To prove (\ref{eq:sum-o}), we use the bound (\ref{eq:omegap}) for $|\widehat{\o}_\ell (r)|$. More precisely, we consider separately the sets where i) $|p-r| < N$ and $|r| < N$ (here we use $|\widehat{\o}_\ell (r)| \leq C/(N |r|^2)$ and we estimate $|r|^{-2} |p-r|^{-\kappa} \lesssim |r|^{-2-\kappa} + |p-r|^{-2-\kappa}$), ii) $|p-r| \geq N$ and $|r| \geq N$ (here we apply $|\widehat{\o}_\ell (r)| \leq C / |r|^3$ and we use $|r|^{-3} |p-r|^{-\kappa} \lesssim |r|^{-3-\kappa} + |p-r|^{-3-\kappa}$), iii) $|p-r| < N$ and $|r| \geq N$ (here we estimate $|\widehat{\o}_\ell (r)| \leq CN^{-3}$), iv) $|p-r| \geq N$ and $|r| < N$ (here we use $|\widehat{\o}_\ell (r)| \leq C / (N |r|^2)$ and we estimate $|p-r|^{-\kappa} \leq C N^{-\kappa}$). 

Thus, for any $\kappa >0$, we arrive at 
\[ \cJ_{N,\ell} = \sum_{p\in \L^*_+} D_p \eta_p  +  \sum_{p \in \L^*_+} D_p  \big(b^*_p b^*_{-p} + b_p b_{-p} \big) + \cE_\cJ  \] 
where
\begin{equation} \label{eq:JNl-final} \pm \cE_\cJ \leq \frac{C}{\sqrt{N\ell}} \, (\cK + 1)  (\cN_++1) + \frac{C}{N^\kappa} \cP^{(2 + \kappa)} (\cN_+ +1) . \end{equation}

Combining the last estimate with (\ref{eq:GN0-fin}), we obtain 
\[ \begin{split} \cG_{N,\ell}^\text{eff} \geq  \; &\frac{(N-1)}{2} \widehat V_\ell(0)   + \sum_{p \in \L^*_+} \big[ p^2 \eta_p  + D_p + \widehat V_\ell(p) + \frac{1}{2N} (\widehat{V}_\ell * \eta) (p) \big] \eta_p  \\
& + \sum_{p \in \L^*_+} \big[ p^2 \eta_p + D_p + \frac 12 \widehat V_\ell (p) +\frac 1 {2N} (\widehat{V}_\ell * \eta) (p) \big] \big( b_p b_{-p} + b^*_p b^*_{-p} \big) \\ &- C (\cN_+ + 1) - \frac{C}{N^\kappa} \cP^{(2+\kappa)} (\cN_+ + 1)\,, \end{split} \]
now with the restriction $0 < \kappa < \nu /2$ (from $\ell \geq N^{-1 + \nu}$, it then follows that $N\ell \geq N^{\nu} \geq N^{2\kappa}$; thus, the first term on the r.h.s. of (\ref{eq:JNl-final}) can be controlled by the second). With the scattering equation (\ref{eq:eta-scat}) and using the bound on the second line of \eqref{eq:intchiell0},  
we obtain 
\[ \cG_{N,\ell}^\text{eff} \geq \frac{N}{2} \int V_\ell (x) g_{\ell_0} (x) dx - C (\cN_+ + 1) - \frac{C}{N^\kappa} \cP^{(2+\kappa)} (\cN_+ + 1) \]
for any $0< \kappa < \nu /2$. With \eqref{eq:intVellg0}, we find (\ref{eq:GN-JN-LB}). 

On the other hand, combining (\ref{eq:JNl-final}) with (\ref{eq:GNl-final}), we arrive at 
\[ \begin{split} \label{eq:GNleff-final}
&\cG^\text{eff}_{N,\ell} \\ & = \frac{(N-1)}{2} \widehat V_\ell(0)   + \sum_{p \in \L^*_+} \big[ p^2 \s_p^2  + \widehat V_\ell(p) \s_p^2 + \widehat V_\ell(p) \g_p \s_p +\frac 1 {2N} \sum_{q \in \L^*_+}\widehat V_\ell(p-q) \eta_p \eta_q + D_p \eta_p  \big] \\
&\qquad + \sum_{p \in \L^*_+} \big[ p^2 \s_p \g_p +\frac 12 \widehat V_\ell(p) (\g_p + \s_p)^2  +\frac 1 {2N} \sum_{q \in \L^*_+}\widehat V_\ell(p-q) \eta_q + D_p  \big]\big( b_p b_{-p} + b^*_p b^*_{-p} \big) \\
&\qquad + \sum_{p \in \L^*_+}   \big[2 p^2 \s_p^2 + \widehat V_\ell(p) (\g_p+\s_p)^2 \big]  b^*_p b_p + \cK + \cV_\ell + \wt{\cE}_{N,\ell}
\end{split}\]
where 
\[
\pm\, {\widetilde \cE}_{N,\ell} \leq  \frac{C}{\sqrt{N \ell}} \,(\cK + \cV_{\ell} +1) (\cN_+ + 1) + \frac{C}{N^\kappa} \cP^{(2+\kappa)} (\cN_++1) \,,
\] 
for any $0 < \kappa < 1$. Observing that 
\[ \begin{split}
\Big|\frac 1 N \sum_{p \in \L^*_+,\, q \in \L^* } \widehat V_\ell (p-q) \eta_q (\s_p + \g_p)^2  \langle \xi, b^*_p b_{p}\xi\rangle \Big| &\leq \frac {C}{N\ell} \| (\cN_++1)^{1/2}\xi\|^2 \\
\Big|\frac 1 N \sum_{p \in \L^*_+,\, q \in \L^* } \widehat V_\ell (p-q) \eta_q \big((\s_p + \g_p)^2 -1  \big) \langle \xi, b_p b_{-p}\xi\rangle \Big| &\leq \frac {C}{N\ell} \| (\cN_++1)^{1/2}\xi\|^2 \\
\Big| \sum_{p\in\L^*_+} p^2 \langle \xi, (b^*_p b_p - a^*_p a_p)\xi \rangle \Big| &\leq \frac C N \|\cK^{1/2} \cN_+^{1/2}\xi\|^2 
\end{split}\]
and that 
\[ \begin{split}  \langle \xi, \cV_\ell  \, \xi \rangle &= \frac{1}{2N} \sum_{\substack{p,q \in \L^*_+, r \in \L^* \\ r \not = -p , -q}} \widehat{V}_\ell (r)  \langle \xi, a_{p+r}^* a_q^* a_{q+r} a_p \xi \rangle \\ &\leq \frac{1}{2N} \sum_{\substack{p,q \in \L^*_+, r \in \L^* \\ r \not = -p , -q}} \frac{|\widehat{V}_\ell (r)}{|q+r|^2} |p+r|^2 \| a_{p+r} a_q \xi \|^2 \leq \frac{C}{N\ell} \| \cK^{1/2} \cN_+^{1/2} \xi \|^2 \end{split} \]
we arrive at (\ref{eq:cG-cJ}), choosing $\kappa = 1/2$.

\end{proof}

\section{Diagonalization of the effective Hamiltonian} 
\label{sec:diag}

According to Prop. \ref{prop:cGcJ}, we need to find a good ansatz for the ground state of the quadratic Hamiltonian $\cQ_{N,\ell}$, defined in (\ref{eq:cQNell}). To this end, we are going to conjugate $\cG_{N,\ell}^\text{eff}$ with a second generalized Bogoliubov transformation, diagonalizing $\cQ_{N,\ell}$. In order to define the appropriate Bogoliubov transformation, we first need to establish some properties of the coefficients $F_p, G_p$, defined in (\ref{eq:FpGp}).
\begin{lemma}\label{lm:FpGp}
Suppose $\ell \geq N^{-1+\nu}$, for some $\nu > 0$. Then there exists a constant $C > 0$ such that  
\begin{equation*}
p^2 /2 \leq F_p \leq C (1+p^2) , \, \qquad |G_p| \leq \frac{C }{p^2}\,, \qquad |G_p| < F_p
\end{equation*}  
for all $N \in \bN$ large enough. 
\end{lemma}
\begin{proof}  Recall the notations $\g_p=\cosh(\eta_p)$ and $\s_p=\sinh(\eta_p)$. 
With $(\s^2_p + \g^2_p) \leq C$ (from the boundedness of $\eta_p$) and \eqref{eq:intchiell0} in Lemma \ref{lm:driftscatt}, 
we immediately obtain $F_p \leq C (1+ p^2)$. To prove the lower bound for $F_p$, let us first consider $|p| > \ell^{-1/2}$. With $(\s^2_p + \g^2_p)=\cosh(2\eta_p) \geq 1$, we find $F_p \geq p^2 - C \geq p^2/2$, if $N$ is large enough (so that $\ell$ is small enough). For $|p| \leq \ell^{-1/2}$, we use $(\widehat{\chi_\ell f_\ell^2} * \widehat{g}_{\ell_0} ) (0) > 0$ to estimate  
\[
\left( \widehat{\chi_\ell f_\ell^2} * \widehat{g}_{\ell_0}  \right)\left( p\right)
> \left( \widehat{\chi_\ell f_\ell^2} * \widehat{g}_{\ell_0}  \right)\left( p\right) - \left( \widehat{\chi_\ell f_\ell^2} * \widehat{g}_{\ell_0}  \right)\left( 0\right) \,.
\]
With 
\[
\left| \left( \widehat{\chi_\ell f_\ell^2} * \widehat{g}_{\ell_0}  \right)\left( p\right) - \left( \widehat{\chi_\ell f_\ell^2} * \widehat{g}_{\ell_0}  \right)\left( 0\right) \right| \leq C |p|  \int |x| \chi_\ell (x) f_\ell^2 ( x) g_{\ell_0} (x) dx \leq C \ell^{\frac{7}{2}}  
\]
we conclude that 
\[
F_p \geq p^2 - C\ell^{1/2} \geq \frac{p^2}{2} .
\]

Next, we show $|G_p| \leq C/p^2$. With the scattering equation \eqref{eq:eta-scat}, we obtain
\[ G_p = 2N \lambda_{\ell_0} (\widehat{\chi_{\ell_0} f_\ell^2} * \widehat{g}_{\ell_0}) (p) + 2p^2 (\g_p\s_p - \eta_p) + (\widehat{V}_\ell \ast \widehat{g}_{\ell_0} )(p) \big[ (\gamma_p+\sigma_p)^2-1\big]\,.  \]
Since 
\be\label{eq:gamma_sigma}\begin{split}
&\left|\gamma_p \sigma_p - \eta_p \right|=\Big|\frac12\sinh(2\eta_p)-\eta_p\Big| \leq C \left| \eta_p\right|^3\leq \frac{C}{|p|^6} \: , \\
&\big| (\gamma_p + \sigma_p)^2-1 \big|=\Big|\sinh(2\eta_p)+\cosh(2\eta_p)-1\Big| \leq C |\eta_p| \leq \frac{C}{p^2} 
\end{split}\ee
and using \eqref{eq:intchiell0} we obtain $|G_p| \leq C / p^2$, as claimed. 

It remains to show $| G_p | \leq F_p $. To this end, we write 
\begin{align*}
F_p-G_p&=p^2\left(\gamma_p-\sigma_p\right)^2 -2D_p , \\
F_p+G_p&= \big[ p^2 + 2 (\widehat{V}_\ell * \widehat{g}_{\ell_0})  (p) \big] (\gamma_p+\sigma_p )^2 + 2 D_p .
\end{align*}
By Lemma \ref{lm:driftscatt} we have $\left| D_p\right|\leq C/ (N\ell)$. Hence, we find, for $N$ large enough, $F_p - G_p \geq p^2 - C /(N\ell) \geq 0$ and, similarly as in the proof of $F_p \geq p^2/2$ (distinguishing small and large $|p|$), $F_p + G_p \geq C p^2 - C/ (N\ell) > 0$. This shows that $F_p > |G_p|$ and concludes the proof of the lemma.
\end{proof}

With Lemma \ref{lm:FpGp}, using in particular the bound $|G_p| < F_p$, we can define, for every $p \in \L^*_+$, $\tau_p \in \bR$ through the identity 
\[ \tanh (2\tau_p) = - \frac{G_p}{F_p}\,. \]
Equivalently,
\begin{equation}\label{eq:taup} \tau_p = \frac{1}{4} \log \frac{1- G_p/F_p}{1+G_p/F_p}. \end{equation} 
From  Lemma \ref{lm:FpGp} we obtain
\begin{equation}\label{eq:tau-dec} |\tau_p| \leq C \frac{|G_p|}{F_p} \leq \frac{C}{|p|^{4}} \end{equation}
for all $p \in \Lambda^*_+$. With the coefficients $\tau_p$, we define the antisymmetric operator
\begin{equation}\label{eq:Btau} B(\tau) = \frac{1}{2} \sum_{p \in \L^*_+} \tau_p \big( b_p^* b_{-p}^* - b_p b_{-p} \big)  \end{equation} 
and we consider the generalized Bogoliubov transformation $e^{B(\tau)}$. 
\begin{lemma} \label{lm:action-tau}
Let $\t_p$ be defined as  in \eqref{eq:taup}. Then, for every $n \in \bN$ any $r\in (0 ; 5)$ there exists a constant $C>0$ such that 
\be  
e^{-B(\tau)}   (\cK+\cV_\ell+\cP^{(r)}+1)(\cN_+ +1)^n e^{B(\tau)}  \leq   C (\cK+\cV_\ell + \cP^{(r)} + 1)(\cN_+ +1)^n\;.
\ee
\end{lemma}
\begin{proof} 
Proceeding as in \cite[Lemma 5.4]{BBCS2} and using that, by \eqref{eq:tau-dec}, $\| \t\|_1$, $\| \t\|_2$ and $\| \tau \|_{H^2}$ are all bounded uniformly in $\ell$ and $N$, we find 
\[ e^{-B(\tau)}  (\cK + \cV_\ell +1)(\cN_+ +1) e^{B(\tau)}  \leq C (\cK+\cV_\ell + 1)(\cN_+ +1) \,. \]
The growth of $\cP^{(r)}(\cN_++1)$ can be controlled as in Lemma \ref{lm:N-Pgrow}, with the only difference that now $\sum_{q \in \Lambda^*_+} |q|^{r} |\t_q|^2 \leq C$, for all $0 < r < 5$. For $n\geq 1$, we can proceed similarly. 
\end{proof} 

The reason why we are interested in the Bogoliubov transformation $e^{B(\tau)}$ is that it diagonalizes the quadratic operator $\cQ_{N,\ell}$ defined as in Prop. \ref{prop:cGcJ}. 
\begin{lemma}\label{lm:diago}
Let $\cQ_{N,\ell}$ be defined as in (\ref{eq:cQNell}), and  $\tau_p$ as in (\ref{eq:taup}). Then, we have   
\[ 
e^{-B(\tau)} \cQ_{N,\ell} e^{B(\tau)} = \frac{1}{2} \sum_{p \in \Lambda^*_+} \left[ -F_p + \sqrt{F_p^2 - G_p^2} \right] + \sum_{p \in \Lambda^*_+} \sqrt{F_p^2 - G_p^2} \; a_p^* a_p + \delta_{N, \ell} \]
where  
\begin{equation*} \label{eq:pmdelta} \pm \delta_{N,\ell} \leq  \frac{C}{N}  (\cK+1)(\cN+1)\,. \end{equation*}
\end{lemma}

\begin{proof} The proof of Lemma \ref{lm:diago} follows exactly as in \cite[Lemma 5.3]{BBCS4}, using 
Lemma \ref{lm:FpGp} (which implies $\| \t\|_1 \leq C$), Lemma \ref{lm:N-Pgrow} and Lemma \ref{lm:action-tau}.
\end{proof} 

With the generalized Bogoliubov transformation $e^{B(\eta)}$, we define a new excitation Hamiltonian $\cM^\text{eff}_{N,\ell} : \cF_+^{\leq N} \to \cF_+^{\leq N}$, setting 
\footnote{Instead of considering first (in Sect. \ref{sec:GN}) the action of $B(\eta)$ and then (here in Sect. \ref{sec:diag}) the action of $B(\tau)$, we could have combined both unitary maps into a single Bogoliubov transformation $\exp (B(\rho))$, with $\rho$ interpolating between $\eta$, for large momenta, and $\tau$, for small momenta. We chose to keep the two transformations apart, because this allowed us to apply several results from \cite{BBCS2}.}
\be \label{eq:cM}
\cM^\text{eff}_{N, \ell} =  e^{-B(\tau)} \cG_{N,\ell}^\text{eff}  e^{B(\tau)}\,.
\ee
Since the generalized Bogoliubov transformation $e^{B(\tau)}$ diagonalizes the quadratic part of $\cG_{N,\ell}^\text{eff}$, the vacuum vector $\Omega \in \cF_+^{\leq N}$ is a good trial state for $\cM^\text{eff}_{N,\ell}$. This correspond to the trial state $\Phi_N = U_N^* e^{B(\eta)} e^{B(\t)}\O \in L^2_s (\L^N)$ for the Hamiltonian $H_N^\text{eff}$.
\begin{prop}\label{prop:cM}
Let $\cM^\text{eff}_{N,\ell}$  be as defined in (\ref{eq:cM}), with $B(\tau)$ as in (\ref{eq:Btau}) and $\cG_{N,\ell}^\text{eff}$ as in (\ref{eq:cGNeff}), with $\ell \geq N^{-1+\nu}$ for some $\nu > 0$ and $\ell_0 > 0$ small enough. Then, we have 
\begin{equation*}\label{eq:CN+bd} \begin{split} 
\langle \O, \cM^\text{eff}_{N,\ell}\O \rangle 
= \; &4 \pi \aa(N  -1) \, +\,  e_\L \aa^2  \\
&-\frac 12 \sum_{p \in \L^*_+} \bigg[ p^2 + 8 \pi \aa - \sqrt{|p|^4 + 16 \pi \aa p^2} - \frac{(8 \pi \aa)^2}{2p^2}\bigg]  + \cO(N^{-\nu/2})
 \end{split} 
\end{equation*}
with $e_\Lambda$ defined as in (\ref{eq:eLambda}). 
\end{prop}

\begin{proof} With \eqref{eq:cG-cJ} and  Lemma \ref{lm:action-tau}, we have 
\[
\cM^\text{eff}_{N,\ell} =  C_{N,\ell} + e^{-B(\tau)} \cQ_{N,\ell} e^{B(\tau)} + \cE'_{N,\ell}
\]
with 
\[
\pm \cE'_{N,\ell} \leq \frac{C}{(N\ell)^{1/2}} (\cK+ \cV_\ell + \cP^{(5/2)} + 1) (\cN_++1) \,.
\]
With Lemma \ref{lm:diago} and the assumption $\ell \geq N^{-1+\nu}$, we obtain
\be \label{eq:MN-vacuum}
 \langle \O, \cM^\text{eff}_{N,\ell}\O \rangle = C_{N,\ell} + \frac{1}{2} \sum_{p \in \Lambda^*_+} \left[ - F_p + \sqrt{F_p^2 - G_p^2} \right] + \cO( N^{-\nu/2})
\ee
with $C_{N,\ell}$, $F_p$ and $G_p$ defined as in \eqref{eq:def-CN} and \eqref{eq:FpGp}.
We rewrite 
\begin{equation} \label{eq:CN0}
	\begin{split}
	C_{N,\ell}= &\, \frac{(N-1)}{2}\widehat{V}_\ell (0) +\sum_{p\in \L^*_+} \Big[ p^2 \eta_p^2 + \widehat{V}_\ell (p) \eta_p + \frac{1}{2N} (\widehat{V}_\ell * \eta) (p) \eta_p +D_p \eta_p \Big]  \\
& + \sum_{p \in \L^*_+} \Big[ p^2 (\sigma_p^2 -\eta_p^2)  + \widehat V_\ell(p) \big( \s^2_p +\g_p \s_p - \eta_p \big) - \frac 1 {2N} \widehat V_\ell(p)\eta_p \eta_0  \Big]\,. \end{split} \end{equation} 
With the scattering equation \eqref{eq:eta-scat} we find 
\begin{equation*} \label{eq:const1}
	\begin{split}
	C_{N,\ell}  =\; &  \frac{(N-1)}{2}\widehat{V}_\ell (0)   +  \sum_{p \in \L^*_+}  \Big[\,\frac 12\widehat V_\ell(p)\eta_p + N \l_{\ell_0} \big(\widehat{(\chi_{\ell_0}f_\ell^2)} \ast \widehat g_{\ell_0}\big)(p) \eta_p \,\Big] \\
& + \sum_{p \in \L^*_+} \Big[ p^2 (\sigma_p^2 -\eta_p^2)  + \widehat V_\ell(p) \big( \s^2_p +\g_p \s_p - \eta_p \big) - \frac 1 {2N} \widehat V_\ell(p)\eta_p \eta_0  \Big] \,.\end{split} \end{equation*} 
Recalling that $V_\ell = 2N \lambda_\ell \chi_\ell f_\ell^2$ we obtain, switching to position space,  
\begin{equation*}\begin{split}  \label{slit} 
C_{N,\ell} =\; &N (N-1) \lambda_\ell \int \chi_\ell (x) f_\ell^2 (x) dx + N \lambda_\ell \int \chi_\ell (x) f_\ell^2 (x) \check{\eta} (x) dx \\ &+ N \lambda_{\ell_0} \int \chi_{\ell_0} (x) f_\ell^2 (x) g_{\ell_0} (x) \check{\eta}  (x) dx - N \l_\ell \widehat{(\chi_{\ell} f_\ell^2)}(0)\eta_0  -N \l_{\ell_0} \big(\widehat{\chi_{\ell_0}f_\ell^2} \ast \widehat g_{\ell_0}\big)(0) \eta_0  \\
&  + \sum_{p \in \L^*_+} \Big[ p^2 (\sigma_p^2 -\eta_p^2)  + \widehat V_\ell(p) (\s^2_p+\g_p \s_p - \eta_p)  - \frac 1 {2N} \widehat V_\ell(p)\eta_p \eta_0   \Big] \,.
	\end{split}
	\end{equation*}
With $\check{\eta} = N (g_{\ell_0} - 1)$, we arrive at 
\[ \begin{split} 
C_{N,\ell} =\; &N (N-1) \lambda_\ell \int \chi_\ell (x) f_\ell^2 (x) dx + N^2 \lambda_\ell \int \chi_\ell (x) f_\ell^2 (x) g_{\ell_0} (x) dx \\ &- N^2 \lambda_\ell \int \chi_\ell (x) f_\ell^2 (x) dx+ N^2 \lambda_{\ell_0} \int \chi_{\ell_0} (x) f_\ell^2 (x) g^2_{\ell_0} (x) dx \\ &- N^2 \lambda_{\ell_0} \int \chi_{\ell_0} (x) f_\ell^2 (x) g_{\ell_0} (x) dx -N \l_\ell \widehat{(\chi_{\ell} f_\ell^2)}(0)\eta_0  -N \l_{\ell_0} \big(\widehat{\chi_{\ell_0}f_\ell^2} \ast \widehat g_{\ell_0}\big)(0) \eta_0  \\
&  + \sum_{p \in \L^*_+} \Big[ p^2 (\sigma_p^2 -\eta_p^2)  + \widehat V_\ell(p) (\s^2_p+\g_p \s_p - \eta_p)  - \frac 1 {2N} \widehat V_\ell(p)\eta_p \eta_0   \Big] \,.
	\end{split}\]
With \eqref{eq:g0} and since $g_{\ell_0}$ satisfies Neumann boundary conditions, we notice that
\[
 \l_{\ell} \int_{B_{\ell_0}} \hskip -0.2cm \chi_{\ell}(x) f_\ell^2(x)  g_{\ell_0}(x) dx -\,   \l_{\ell_0} \int_{B_{\ell_0}} \hskip -0.2cm\chi_{\ell_0}(x) f_\ell^2(x) g_{\ell_0}(x)dx = \int_{B_{\ell_0}} \hskip -0.2cm \nabla \big( f_\ell^2(x) \nabla g_{\ell_0}(x) \big) dx = 0\,.
\]
Thus, using $f_{\ell_0} = f_\ell g_{\ell_0}$, we conclude that 
\footnote{Instead of applying the scattering equation on the first line of (\ref{eq:CN0}), we could have switched to position space and argued as in (\ref{eq:en-g}) to reconstruct the term on the r.h.s. of (\ref{eq:en-g2}); this would have given an alternative derivation of (\ref{eq:Delta1}).} 
\begin{equation}\label{eq:Delta1}  \begin{split} 
C_{N,\ell} = \; &N (N-1) \lambda_{\ell_0} \int \chi_{\ell_0} (x) f^2_{\ell_0} (x) dx \\ &+ N \lambda_{\ell_0} 
\int \chi_{\ell_0} (x) f^2_{\ell_0} (x) dx - N \lambda_\ell \int \chi_\ell (x) f_\ell^2 (x) dx \\
& -N \l_\ell \widehat{(\chi_{\ell} f_\ell^2)}(0)\eta_0  -N \l_{\ell_0} \big(\widehat{\chi_{\ell_0}f_\ell^2} \ast \widehat g_{\ell_0}\big)(0) \eta_0  \\
&  + \sum_{p \in \L^*_+} \Big[ p^2 (\sigma_p^2 -\eta_p^2)  + \widehat V_\ell(p) (\s^2_p+\g_p \s_p - \eta_p)  - \frac 1 {2N} \widehat V_\ell(p)\eta_p \eta_0  \Big] \,.
	\end{split}\end{equation} 
To bound the terms on the second line of (\ref{eq:Delta1}), we use Lemma \ref{lm:hardcorescatt} to show that 
\[ \begin{split}
\Big | N \l_{\ell_0} \int \chi_{\ell_0}(x) f_{\ell_0}^2 (x)  dx - 4 \pi \aa \Big| & \leq \frac{C}{N\ell_0}  \\
\Big | N \l_{\ell} \int \chi_{\ell}(x) f_\ell^2(x)  dx - 4 \pi \aa \Big| & \leq \frac{C}{N\ell} \,.
\end{split}\]
Similarly, we find 
\[ \label{eq:const2}
 - N \l_\ell \widehat{(\chi_{\ell} f_\ell^2)}(0)\eta_0  -N \l_{\ell_0} \big(\widehat{\chi_{\ell_0}f_\ell^2} \ast \widehat g_{\ell_0}\big)(0) \eta_0 =  - 8\pi \aa \eta_0 + \cO ((N\ell)^{-1}) \,.
\]
As for the terms on the fourth line, the last contribution can be bounded, using that $|\eta_0 | \leq C$, by 
\[ \Big| \frac{1}{2N} \sum_{p \in \L^*_+} \widehat{V}_\ell (p) \eta_p \eta_0 \Big| \leq \frac{C}{N\ell}\,. \]
To handle the other terms on the fourth line of (\ref{eq:Delta1}), we combine them with the first term in the sum on the r.h.s. of (\ref{eq:MN-vacuum}). Recalling (\ref{eq:FpGp}), we find (using again $\ell \geq N^{-1+\nu}$) 
\be \begin{split} \label{eq:const3}
\sum_{p \in \L^*_+}& \bigg[ p^2 (\sigma_p^2 -\eta_p^2)  +   \widehat V_\ell(p) (\s^2_p +\g_p \s_p - \eta_p) - \frac 12  F_p \bigg]  \\
&= -\sum_{p \in \L^*_+}  \left[\frac{p^2}{2} +\frac 12 \,(\widehat{V_\ell} \ast \widehat g_{\ell_0})(p)  + p^2 \eta_p^2 +(\widehat{V_\ell} \ast \widehat g_{\ell_0})(p)  \eta_p \right]+ \cO (N^{-\nu}) 
\end{split}\ee
where we bounded, using $|\s_p^2 + \g_p \s_p - \eta_p|\leq C|\eta_p|^2 \leq C /|p|^4$ (see \eqref{eq:gamma_sigma}), 
\[ \Big| \sum_{p \in \L^*_+} \big( \widehat V_\ell(p) - (\widehat{V_\ell} \ast \widehat g_{\ell_0})(p) \big)(\s^2_p +\g_p \s_p - \eta_p ) \Big| \leq \frac{C}{N\ell}\,. \]
As for the remaining term on the r.h.s. of (\ref{eq:MN-vacuum}), we can write 
\[
F_p^2 - G_p^2= |p|^4 +2p^2  (\widehat{V_\ell} \ast \widehat g_{\ell_0})(p)  + A_p
\]
with the notation 
\[ \label{eq:Ap}
A_p = -4 D_p \Big( (\widehat{V_\ell} \ast \widehat g_{\ell_0})(p)  (\g_p+\s_p)^2 + D_p +  2 p^2 \g_p \s_p \Big) \,.
\]
From  (\ref{eq:etap}), we have $|D_p|\leq C/ (N\ell)$. Thus,  with $(\g_p + \s_p)^2 \leq C$ and $|\g_p \s_p|\leq C|p|^{-2}$,  we obtain $|A_p| \leq C / (N\ell)$. Using this bound and the observation that $|p|^4 +2p^2 (\widehat{V_\ell} \ast \widehat g_{\ell_0})(p) $ and $|p|^4 +2p^2  (\widehat{V_\ell} \ast \widehat g_{\ell_0})(p)  + A_p$ are positive and bounded away from zero  we write
\[ \begin{split}
\sqrt{F_p^2 - G_p^2} &= \sqrt{|p|^4 +2p^2 (\widehat{V_\ell} \ast \widehat g_{\ell_0})(p)} \\
& \hskip 2cm+ \frac{A_p}{ \sqrt{|p|^4 +2p^2 (\widehat{V_\ell} \ast \widehat g_{\ell_0})(p) + A_p} + \sqrt{|p|^4 +2p^2 (\widehat{V_\ell} \ast \widehat g_{\ell_0})(p)} }\,.
\end{split}\]
Expanding the square roots in the denominator around $p^2$, we easily find (using again $|A_p| \leq C / (N\ell)$),
\[ \sum_{p \in \L^*_+} \frac{A_p}{ \sqrt{|p|^4 +2p^2 (\widehat{V_\ell} \ast \widehat g_{\ell_0})(p) + A_p} + \sqrt{|p|^4 +2p^2 (\widehat{V_\ell} \ast \widehat g_{\ell_0})(p)} } = \sum_{p \in \L^*_+} \frac{A_p}{2p^2} + \cO (N^{-\nu})\,. \]
Combining the last two equations with (\ref{eq:MN-vacuum}), (\ref{eq:Delta1}), (\ref{eq:const3}), we find 
\be \begin{split} \label{eq:const5}
 \langle \O, &\cM^\text{eff}_{N,\ell}\O \rangle \\ =\; & N (N-1) \l_{\ell_0} \int \chi_{\ell_0}  (x) f^2_{\ell_0} (x) dx - 8\pi \aa \eta_0  \\
& + \frac 1 2 \sum_{p\in \L^*_+} \Big[ -p^2 -(\widehat{V_\ell} \ast \widehat g_{\ell_0})(p)  + \sqrt{|p|^4 +2p^2 (\widehat{V_\ell} \ast \widehat g_{\ell_0}) (p) } + \frac{\big((\widehat{V_\ell} \ast \widehat g_{\ell_0})(p) \big)^2}{2p^2} \Big] \\
& -  \sum_{p \in \L^*_+}  \left[ p^2 \eta^2_p +   (\widehat{V_\ell} \ast \widehat g_{\ell_0})(p)\eta_p - \frac{A_p}{ 4p^2} +  \frac{\big( (\widehat{V_\ell} \ast \widehat g_{\ell_0})(p)\big)^2}{4p^2}\right]  +\cO(N^{-\nu/2})\,.
\end{split}\ee
Estimating $|(\g_p + \s_p)^2 -1 | \leq C |\eta_p| \leq C/|p|^2$ and $|\g_p \s_p - \eta_p| \leq C \eta_p^3 \leq C / |p|^6$  (see \eqref{eq:gamma_sigma}), we obtain 
\[ \sum_{p \in \L^*_+} \frac{A_p}{4p^2} = - \sum_{p \in \L^*_+} \frac{D_p}{p^2} \left[ 2p^2 \eta_p + (\widehat{V}_\ell * \widehat{g}_{\ell_0}) (p) + D_p  \right]+\cO(N^{-\nu})\,. \]
Solving the scattering equation \eqref{eq:eta-scat} for $D_p$, we obtain 
\[ \sum_{p \in \L^*_+} \frac{A_p}{4p^2} = - \sum_{p \in \L^*_+} \frac{D_p}{p^2} \left[ p^2 \eta_p + \frac{1}{2} (\widehat{V}_\ell * \widehat{g}_{\ell_0}) (p) + N \lambda_{\ell_0} \big( \widehat{\chi_{\ell_0} f_\ell^2} * \widehat{g}_{\ell_0} \big) (p) \right]+\cO(N^{-\nu}) \,.\]
Inserting this bound in the last line of (\ref{eq:const5}), we get 
\[ \begin{split} 
\sum_{p \in \L^*_+}  &\left[ p^2 \eta^2_p +   (\widehat{V_\ell} \ast \widehat g_{\ell_0})(p)\eta_p - \frac{A_p}{ 4p^2} +  \frac{\big( (\widehat{V_\ell} \ast \widehat g_{\ell_0})(p)\big)^2}{4p^2}\right] \\
= \; &\sum_{p \in \L^*_+} \Big[ p^2 \eta_p  +\frac{1}{2} (\widehat{V}_\ell * \widehat{g}_{\ell_0}) (p) + D_p \Big] \eta_p + \sum_{p \in \L^*_+} \frac{1}{2}  (\widehat{V}_\ell * \widehat{g}_{\ell_0}) (p) \Big[ \eta_p + \frac{(\widehat{V}_\ell * \widehat{g}_{\ell_0}) (p)}{2p^2} + \frac{D_p}{p^2} \Big] \\ &+ \sum_{p \in \L^*_+} \frac{N  \lambda_{\ell_0} \big( \widehat{\chi_{\ell_0} f_\ell^2} * \widehat{g}_{\ell_0} \big) (p)}{p^2} D_p +\cO(N^{-\nu})   \,.
\end{split} \]
With the scattering equation \eqref{eq:eta-scat}, we find 
\[ \begin{split} 
\sum_{p \in \L^*_+}  \Big[ p^2 \eta^2_p +   (\widehat{V_\ell} \ast &\widehat g_{\ell_0})(p)\eta_p - \frac{A_p}{ 4p^2} +  \frac{\big( (\widehat{V_\ell} \ast \widehat g_{\ell_0})(p)\big)^2}{4p^2}\Big] \\
= \; &\sum_{p \in \L^*_+}   \frac{N  \lambda_{\ell_0} \big( \widehat{\chi_{\ell_0} f_\ell^2} * \widehat{g}_{\ell_0} \big) (p)}{p^2} \left[ p^2 \eta_p + \frac{1}{2} (\widehat{V_\ell} \ast \widehat g_{\ell_0})(p) + D_p \right] +\cO(N^{-\nu}) \\ = \; &\sum_{p \in \L^*_+}  \frac{\big[ N  \lambda_{\ell_0} \big( \widehat{\chi_{\ell_0} f_\ell^2} * \widehat{g}_{\ell_0} \big)(p)\big]^2}{p^2} = \frac{9 \frak{a}^2}{\ell_0^6} \sum_{p \in \L^*_+} \frac{\widehat{\chi}_{\ell_0}^2 (p)}{p^2} +\cO(N^{-\nu})
\end{split} \]
where in the last step we used Lemma \ref{lm:hardcorescatt} and Lemma \ref{lm:driftscatt}. From (\ref{eq:const5}), we conclude that 
\be \begin{split} \label{eq:const8}
 \langle \O, \cM^\text{eff}_{N,\ell}\O \rangle =\; & N (N-1) \l_{\ell_0} \int \chi_{\ell_0} (x) f^2_{\ell_0} (x) dx - 8\pi \aa \eta_0  - \frac{9\aa^2}{\ell_0^6} \,\sum_{p \in \L^*_+} \frac{  \widehat{\chi}_{\ell_0}(p)^2}{p^2} \\
&-\frac 1 2 \sum_{p \in \L^*_+} e_{N}(p)  +\cO(N^{-\nu/2})
\end{split}\ee
where we introduced the notation
\[\label{eq:E-Bog-N}
e_{N}(p) = p^2 +(\widehat{V_\ell} \ast \widehat g_{\ell_0})(p)  - \sqrt{|p|^4 +2p^2 (\widehat{V_\ell} \ast \widehat g_{\ell_0}) (p) } - \frac{\big((\widehat{V_\ell} \ast \widehat g_{\ell_0})(p) \big)^2}{2p^2} \, . 
\]
Expanding the square root, we find that $|e_{N}(p)|\leq C |p|^{-4}$, uniformly in $N$ and $\ell$.  This allows us to cut the sum to $|p| \leq \ell^{-1}$, with a negligible error. For $|p| \leq \ell^{-1}$, we can then compare $(\widehat{V}_\ell * \widehat{g}_{\ell_0}) (p)$ with $(\widehat{V}_\ell * \widehat{g}_{\ell_0}) (0)$ and then with $\widehat{V}_\ell (0)$. Proceeding similarly to \cite[Eq. (5.26)-(5.27)]{BBCS4}, we conclude that 
\begin{equation}\label{eq:bog1} \sum_{p \in \L^*_+} e_{N}(p) = \sum_{p \in \L^*_+} \Big[ p^2 + 8\pi \frak{a} - \sqrt{|p|^4 + 16 \pi \frak{a} p^2} - \frac{(8\pi \frak{a})^2}{2p^2} \Big] + \cO (\ell \log \ell)\,. \end{equation} 
Finally, let us compute the last term on the first line on the r.h.s. of (\ref{eq:const8}). Using the expressions (see \cite[Eq. (5.5), (5.29) and (5.33)]{BBCS4}):
\[ \begin{split}
\widehat \chi_{\ell_0}(p) &= \frac{4 \pi \ell_0}{|p|^2}  \left( \frac{\sin(\ell_0 |p|)}{\ell_0|p|} -   \cos(\ell_0 p)\right) \\
%%%
\widehat{(\chi_{\ell_0} |\cdot|^2)}(p) &= \frac{4 \pi \ell_0^3}{|p|^2}  \left( - \frac{6\sin(\ell_0 |p|)}{\ell_0^3 |p|^3} +  \frac{6\cos(\ell_0 p)}{\ell_0^2|p|^2} +\frac{3\sin(\ell_0 |p|)}{\ell_0 |p|} -  \cos(\ell_0 p)\right) \\
%%%
\widehat{(\chi_{\ell_0}|\cdot|^{-1})}(p) &= \frac{4 \pi }{|p|^2}  \Big( 1-  \cos(\ell_0 p)\Big) \\
\end{split}\]
we can rewrite
\be \begin{split} \label{eq:chisquare}
- \frac{9\aa^2}{\ell_0^6} \sum_{p\in \L^*_+} \frac{  \widehat{\chi}_{\ell_0}(p)^2}{p^2}  =\; & - 12 \pi \frac{\aa^2}{\ell_0^3} \sum_{p\in \L^*_+} \frac{\widehat{\chi}_{\ell_0}(p)}{|p|^2}  + \frac{3\aa^2}{2\ell_0^6} \sum_{p\in \L^*_+}\widehat \chi_{\ell_0}(p) \cdot \widehat{(\chi_{\ell_0}|\cdot|^2)}(p) \\
& - \frac 9 2 \frac{\aa^2}{\ell_0^4} \sum_{p\in \L^*_+} \widehat \chi_{\ell_0}(p)^2  + \frac{3\aa^2}{\ell_0^3} \sum_{p\in \L^*_+}\widehat \chi_{\ell_0}(p) \cdot \widehat{(\chi_{\ell_0}|\cdot|^{-1})}(p)\,.
\end{split}\ee
From \cite[Eq. (5.31)]{BBCS4} we have
\be \label{eq:I-def}
 - 12 \pi \frac{\aa^2}{\ell_0^3} \sum_{p\in \L^*_+}\frac{\widehat{\chi}_{\ell_0}(p)}{|p|^2} = 6 \pi \aa^2 \big( I_0- \frac 1 {\ell_0} - \frac{4}{15}\pi \ell_0^2 \big)
\ee
where 
\[
I_0 = \frac{1}{3 \pi} - \frac{2}{3\pi} \lim_{M \to \io} \sum_{\substack{p \in \L^*_+: \\|p_i|\leq M}} \frac{\cos (|p|)}{p^2}\,.
\]
Computing the different terms on the r.h.s. of \eqref{eq:chisquare} and using \eqref{eq:I-def} we obtain
\[
- \frac{9\aa^2}{\ell_0^6} \sum_{p\in \L^*_+} \frac{  \widehat{\chi}_{\ell_0}(p)^2}{p^2} = 6 \pi \aa I_0 - \frac{24}{5} \pi \frac{\aa^2}{\ell_0}  - \frac{16}5 \pi^2 \aa^2 \ell_0^2\,.
\]
Inserting \eqref{eq:bog1}, \eqref{eq:Vell-zero}, \eqref{eq:eta0} and the last equation in \eqref{eq:const8}, we conclude that 
\[ \begin{split} \langle \Omega, &\cM^\text{eff}_{N,\ell} \Omega \rangle \\ &=  4 \pi \aa (N-1) +  e_\L \aa^2 -\frac{1}{2} \sum_{p \in \L^*_+} \big[ p^2 + 8\pi \frak{a} - \sqrt{|p|^4 + 16 \pi \frak{a} p^2} - \frac{(8\pi \frak{a})^2}{2p^2} \Big] + \cO (N^{-\nu/2})\, \end{split} \]
with $e_\L$ defined as in (\ref{eq:eLambda}). 
\end{proof}

\section{Bounds on the trial state} \label{sec:bound}

We introduce some operators to control the regularity of our trial state. First of all, we recall the definition of the operator $\cP^{(r)}$, defined in (\ref{eq:P_r}) for $1 < r < 5$. Furthermore, we need some observables acting of several particles. For $n \in \bN$, we define 
\begin{equation} \label{eq:cTnd} 
\cT_n =\sum_{p_1, \dots , p_n \in \L^*_+} p_1^2 \dots p_n^2 \, a_{p_1}^* \dots a_{p_n}^* a_{p_n} \dots a_{p_1} \,.
\end{equation} 
Since $\eta$ has limited decay in momentum space (see (\ref{eq:etap})), we will only be able to control the expectation of $\cT_n$ for $n = 2,3,4$. To control some error terms, it is also important to use less derivatives on each particle. We define, for $\delta > 0$ small enough (we will later impose the condition $\delta \in (0;1/6)$), 
\begin{equation}\label{eq:cAnd} 
\cA^{(\delta)}_{n} = \sum_{p_1, \dots , p_n \in \L^*_+} |p_1|^{3/2+ \delta} \dots |p_n|^{3/2+ \delta} a_{p_1}^* \dots a_{p_n}^* a_{p_n} \dots a_{p_1} \,.
\end{equation} 
We will be able to control the expectation of $\cA^{(\delta)}_n$, for all $n \in \bN$. Additionally, we will also need the observable 
\begin{equation}\label{eq:cSand}
\cS^{(\eps,\delta)}_n = \sum_{p_1, \dots , p_n \in \L^*_+} |p_1|^{2+ \eps} |p_2|^{3/2+\delta} \dots |p_n|^{3/2+ \delta} a_{p_1}^* \dots a_{p_n}^* a_{p_n} \dots a_{p_1} \,. 
\end{equation}  
All these operators act on the excitation Fock space $\cF^{\leq N}_+$. In order to bound their  expectation on our trial state, we need to control their growth under the action of $B(\eta)$, similarly as we did in Lemma \ref{lm:N-Pgrow} for $\cP^{(r)}$. 
\begin{lemma} \label{lm:regu} 
For $n \in \bN \backslash \{ 0 \}$ and $0< \delta < 1/6$, we consider $\cA_n^{(\delta)}$ as in (\ref{eq:cAnd}). We define recursively the sequence $\alpha_n$ (depending on the parameter $\delta$) by setting $\alpha_1 = 1/2 + \delta$, $\alpha_2 = 2+2\delta$ and 
\begin{equation}\label{eq:aln}  \alpha_n = \big[ \alpha_{n-1} + \alpha_{n-2} \big] /2 + 7/4 + 3\delta/2 \,.\end{equation} 
Then, for every $k \in \bN$, there exists a constant $C > 0$ (depending also on $n$ and $\delta$) such that 
\begin{equation}\label{eq:Adn} \langle e^{B(\eta)} \xi , \cA^{(\delta)}_{n} (\cN_+ + 1)^k e^{B(\eta)}  \xi \rangle \leq C \ell^{-\alpha_{n}} \Big\{ \| (\cN_+ + 1)^{k/2} \xi \|^2 + \sum_{j=1}^n \langle \xi, \mathcal{A}_{j}^{(\delta)} (\cN_+ + 1)^k \xi \rangle \Big\} \end{equation}
for all $\xi \in \cF^{\leq N}_+$. 

For $n \in \mathbb{N}\backslash \{ 0 \}$, let \[ \cI_n = \{ ( \eps ; \delta ) \in (-1;3) \times (0;1/6)  : \eps+2\delta < 3/2^{(n-1)} \}\,. \]
For $(\eps ; \delta) \in \cI_n$ we consider $\cS^{(\eps,\delta)}_n$ as in (\ref{eq:cSand}). Moreover, we define the sequence $\beta^\eps_n = \alpha_n + 1/2 + \eps - \delta$, with $\alpha_n$ as in (\ref{eq:aln}) (the sequence $\beta_n^\eps$ depends also on $\delta$; since this dependence does not play an important role in the proof, we do not make it explicit in the notation). Then, for every $k \in \bN$, there exists a constant $C > 0$ (depending also on $n,\eps,\delta$) such that 
\begin{equation}\label{eq:Sedn} 
\begin{split} \langle e^{B(\eta)} \xi , &\cS^{(\eps,\delta)}_{n} e^{B(\eta)} \xi \rangle \leq C \ell^{-\beta^\eps_{n}} \Big\{ \|  \xi \|^2 + \sum_{j=1}^n \sup_{\eps,\delta \in \cI_j}  \; \langle \xi, \mathcal{S}_{j}^{(\eps,\delta)}  \xi \rangle \Big\} \end{split} 
\end{equation}
for all $\xi \in \cF^{\leq N}_+$.

For $n \in \{2,3,4 \}$, we can also control the growth of the operator $\cT_n$, defined in (\ref{eq:cAnd}). We find 
\begin{equation}\label{eq:T2} 
\begin{split}  
\langle e^{B(\eta)} \xi , \cT_2 e^{B(\eta)} \xi \rangle &\leq C \ell^{-3} \langle \xi, \big( 1 + \cP^{(4)} + \cT_2  \big)  \xi \rangle \\
\langle e^{B(\eta)} \xi , \cT_3 e^{B(\eta)} \xi \rangle &\leq C \ell^{-4} \langle \xi, \big( 1 +  \cP^{(4)} + \cZ_{4,2} + \cT_3 \big)  \xi \rangle \\
\langle e^{B(\eta)} \xi , \cT_4 e^{B(\eta)} \xi \rangle &\leq C \ell^{-6} \langle \xi, \big( 1 +  \cP^{(4)} + \cZ_{4,4} + \cZ_{4,2,2} + \cT_4 \big)  \xi \rangle
\end{split}
\end{equation} 
for every $\xi  \in \cF^{\leq N}_+$. Here we introduced the notation (for $m =2,4$) 
\begin{equation}\label{eq:Zdef}\begin{split}  \cZ_{4, 2} &= \sum_{\substack{p_1, p_2 \in \L^*_+ : \\ p_1 \not = \pm p_2}}  |p_1|^4 p_2^2 \, a^*_{p_1} a_{p_2}^* a_{p_2} a_{p_1}, \qquad  \cZ_{4, 4} = \sum_{\substack{p_1, p_2 \in \L^*_+ : \\ p_1 \not = \pm p_2}}  |p_1|^4 |p_2|^{4} \, a^*_{p_1} a_{p_2}^* a_{p_2} a_{p_1} \\  \cZ_{4,2,2} &= \sum_{\substack{p_1, p_2, p_3 \in \L^*_+ :\\  p_1 \not = \pm p_2, \pm p_3}}  |p_1|^4  p_2^2 p_3^2 \, a^*_{p_1} a_{p_2}^* a_{p_3}^* a_{p_3} a_{p_2} a_{p_1} \,.  \end{split}  \end{equation}  
Finally, we will also need an improvement of (\ref{eq:Sedn}), for $n=3$. For $\eps > -1$, $0 < \delta < 1/6$ with $\eps + \delta < 1$, we find 
\begin{equation}\label{eq:impr-S3} \langle e^{B(\eta)} \xi , \cS_3^{(\eps,\delta)} e^{B(\eta)} \xi \rangle \leq C \ell^{-3-\eps-2\delta} \Big\{ \langle \xi , \big[ 1 + \cP^{(4)} + \cZ_{4,2} \big] \xi \rangle+ \sup_{(\eps, \delta) \in \cI_3} \langle \xi , S^{(\eps,\delta)}_3 \xi \rangle \Big\}  \end{equation} 
for all $\xi \in \cF_+^{\leq N}$ (observe that, in (\ref{eq:Sedn}), $\beta_3^\eps = 7/2+\eps + 2\delta$). 
\end{lemma}

{\it Remark.} The sequence $\alpha_n$ defined in (\ref{eq:aln}) is given explicitly by  
\begin{equation}\label{eq:aln2} \alpha_n = \left( \frac{7}{6} + \delta \right) n - \frac{4}{9} \left( 1- \left(-\frac{1}{2} \right)^n \right)\,. \end{equation} 
 
\begin{proof} 
We begin with (\ref{eq:Adn}). We consider $k=0$; the case $k > 0$ can be handled similarly. For $n \geq 1$ and $0 < \delta < 1/6$, we set 
\[ F^{(\delta)}_n (t) = \langle e^{t B(\eta)} \xi , \cA_{n}^{(\delta)} e^{t B(\eta)} \xi \rangle \,.\]
For $n \geq 2$, we compute
\begin{equation*} \label{eq:derF} \begin{split} \frac{dF^{(\delta)}_n}{dt} (t) &= \langle e^{t B(\eta)} \xi , \big[  \cA_{n}^{(\delta)} ,  B(\eta) \big]  e^{t B(\eta)} \xi \rangle \\ &= \sum_{p_1, \dots ,p_n \in \L^*_+} |p_1|^{3/2+\delta} \dots |p_n|^{3/2+\delta} \langle e^{t B(\eta)} \xi , \big[ a_{p_1}^* \dots a_{p_n}^* a_{p_n} \dots a_{p_1}, B(\eta) \big]  e^{t B(\eta)} \xi \rangle\,. \end{split} \end{equation*} 
With the identity 
\begin{equation*} \label{eq:nth_commutator}
\left[ a^*_{p_1}\dots a^*_{p_n} a_{p_n}\dots a_{p_1} , b^*_q \right]= \sum_{j=1}^n \delta_{q,p_j} b^*_{p_j} a^*_{p_1}\dots a^*_{p_{j-1}}a^*_{p_{j+1}}\dots a^*_{p_n}  a_{p_n}\dots a_{p_{j+1}} a_{p_{j-1}}\dots a_{p_1}
\end{equation*}
we find  
\[ \begin{split}
\big[ a^*_{p_1}\dots a^*_{p_n} &a_{p_n}\dots a_{p_1} , b^*_q b_{-q}^* \big] \\ = \; &\sum_{j=1}^n \delta_{p_j, -q} b_q^* b_{p_j}^* a_{p_1}^* \dots a_{p_{j-1}}^* a_{p_{j+1}}^*  \dots a_{p_n}^* a_{p_n} \dots a_{p_{j+1}} a_{p_{j-1}} \dots a_{p_1} \\ &+\sum_{j=1}^n \delta_{p_j, q} b_{p_j}^* a_{p_1}^* \dots a_{p_{j-1}}^* a_{p_{j+1}}^*  \dots a_{p_n}^* a_{p_n} \dots a_{p_{j+1}} a_{p_{j-1}} \dots a_{p_1} b_{-q}^* \,.\end{split} \]
Thus
\begin{equation}\label{eq:comm} \begin{split}
\big[ a^*_{p_1} &\dots a^*_{p_n} a_{p_n}\dots a_{p_1} , b^*_q b_{-q}^* \big] \\ = \; & \sum_{j=1}^n (\delta_{p_j, -q} + \delta_{p_j, q}) \, b_q^* b_{-q}^* a_{p_1}^* \dots a_{p_{j-1}}^* a_{p_{j+1}}^*  \dots a_{p_n}^* a_{p_n} \dots a_{p_{j+1}} a_{p_{j-1}} \dots a_{p_1} \\ &+  \sum_{j=1}^n \sum_{i \not = j} \delta_{p_j, q} \delta_{p_i, -q} b_q^* b_{-q}^* a_{p_1}^* \dots a_{p_{j-1}}^* a_{p_{j+1}}^*  \dots a_{p_{i-1}}^* a_{p_{i+1}}^* \dots a_{p_n}^* \\ &\hspace{3.9cm} \times a_{p_n} \dots a_{p_{j+1}} a_{p_{j-1}} \dots a_{p_{i+1}} a_{p_{i-1}} \dots a_{p_1}\,. \end{split} \end{equation} 
Therefore, we can bound 
\begin{equation}\label{eq:Fd1} \begin{split}  \Big| \frac{dF^{(\delta)}_n}{dt} (t)  \Big| \leq \; &C \sum_{p_1, \dots , p_{n-1} , q \in \L^*_+} |\eta_q | |q|^{3/2+\delta} |p_1|^{3/2 + \delta} \dots |p_{n-1}|^{3/2+ \delta} \\ &\hspace{2cm} \times  \| a_q a_{p_1} \dots a_{p_{n-1}} e^{tB(\eta)} \xi \| \| a_{-q}^* a_{p_1} \dots a_{p_{n-1}}  e^{tB(\eta)} \xi \| 
\\ &+ C \sum_{p_1, \dots , p_{n-2} , q \in \L^*_+} |\eta_q | |q|^{3+2\delta} |p_1|^{3/2 + \delta} \dots |p_{n-2}|^{3/2+ \delta} \\ &\hspace{2cm} \times \| a_q a_{p_1} \dots a_{p_{n-2}} e^{tB(\eta)} \xi \| \| a_{-q}^* a_{p_1} \dots a_{p_{n-2}}  e^{tB(\eta)} \xi \| \end{split} \end{equation} 
for a constant $C$ depending on $n$. Estimating $\| a^*_{-q} \zeta \| \leq \| a_{-q} \zeta \| + \| \zeta \|$ and applying Cauchy-Schwarz's inequality, we obtain, for any $n \geq 3$,  
\[ \begin{split}  \Big| \frac{dF^{(\delta)}_n}{dt} (t)  \Big| \leq \; &C \| \eta \|_\infty  F^{(\delta)}_n (t) +C F^{(\delta)}_n (t)^{\frac{1}{2}} F^{(\delta)}_{n-1} (t)^{\frac{1}{2}} \Big[ \sum_{q \in \L^*_+} |q|^{3/2+\delta} \eta_q^2 \Big]^{\frac{1}{2}} \\ &+ C \big[ \sup_{q\in \L^*_+} |q|^{3/2+ \delta} \eta_q \big]  \, F^{(\delta)}_{n-1} (t) + C F^{(\delta)}_{n-1} (t)^{\frac{1}{2}}  F^{(\delta)}_{n-2} (t)^{\frac{1}{2}} \Big[ \sum_{q \in \L^*_+}  |q|^{9/2 + 3\delta} \eta_q^2 \Big]^{\frac{1}{2}} \,.\end{split} \]
With Lemma \ref{lm:driftscatt}, we arrive at 
\begin{equation} \label{eq:Fn-gro}  \Big| \frac{dF^{(\delta)}_n}{dt} (t)  \Big| \leq C F^{(\delta)}_n (t) + C \ell^{-1/2-\delta} F^{(\delta)}_{n-1} (t) + C \ell^{-7/4-3\delta/2} F^{(\delta)}_{n-1} (t)^{\frac{1}{2}} \,  F^{(\delta)}_{n-2} (t)^{\frac{1}{2}}\,. \end{equation} 
This bound is also valid for $n=2$, setting $F^{(\delta)}_0 (t) = \| \xi \|^2$. If $n=1$, we can use (\ref{eq:P_r}) to estimate 
\[ F^{(\delta)}_1 (t) \leq C \ell^{-1/2-\delta}  \big[ \| \xi \|^2 + \langle \xi, \cA_1^{(\delta)} \xi \rangle \big]  \]
for all $t \in [0;1]$. Inserting this bound on the r.h.s. of (\ref{eq:Fn-gro}) (with $n=2$), we obtain 
\[ F^{(\delta)}_2 (t) \leq C F^{(\delta)}_2 (0) + C \ell^{-2-2\delta} \langle \xi,\big[ \| \xi \|^2 + \cA_1^{(\delta)} \big]  \xi \rangle  \leq C \ell^{-2-2\delta} \langle \xi , \big[ 1 + \cA_1^{(\delta)} + \cA_2^{(\delta)}  \big] \xi \rangle\,. \]
Defining the coefficients $\alpha_n$ iteratively, as in (\ref{eq:aln}), by simple induction we conclude from (\ref{eq:Fn-gro}) that, for all $n \in \bN$, there exists a constant $C > 0$ such that 
\begin{equation}\label{eq:Fbd}  F^{(\delta)}_n (t) \leq C \ell^{-\alpha_n} \big\langle \xi, \big[ 1 + \sum_{j=1}^n \cA_j^{(\delta)} \big] \xi \big\rangle \, . \end{equation} 

Let us consider (\ref{eq:Sedn}), again for $k=0$. For $n \geq 1$, $(\eps ;\delta) \in \cI_n$, $t \in [0;1]$, we define 
\[ G^{(\eps,\delta)}_n (t) = \langle e^{t B(\eta)} \xi, \cS^{(\eps,\delta)}_n e^{t B(\eta)} \xi \rangle\,. \]
Proceeding similarly to (\ref{eq:Fn-gro}) we find, for $n \geq 2$ (with the convention that $G_0^{(\eps,\delta)} (t) = 0$ and $F_0^{(\delta)} (t) = \| \xi \|^2$ for all $t \in [0;1]$),   
\begin{equation} \label{eq:derG} \begin{split} \Big| \frac{dG_n^{(\eps,\delta)} (t)}{dt} \Big| \leq \; &C G_n^{(\eps,\delta)} (t) + C \ell^{-1/2-\delta} G_{n-1}^{(\eps,\delta)} (t) + C \ell^{-7/4-3\delta/2}  \, G_{n-1}^{(\eps,\delta)} (t)^{\frac{1}{2}} \, G_{n-2}^{(\eps,\delta)} (t)^{\frac{1}{2}} \\ &+ C \ell^{-1-\eps} F_{n-1}^{(\delta)} (t) + C \ell^{-2-\delta-\eps/2+\theta/2} \, G_{n-1}^{(\eps+\theta, \delta)} (t)^{\frac{1}{2}} \, F_{n-2}^{(\delta)} (t)^{\frac{1}{2}} \end{split}  \end{equation}
for a $\theta > \eps + 2\delta$. The second line arises from the contributions to the commutator (\ref{eq:comm}) where $q$ coincides 
with the variable raised to the power $2+\eps$. In fact, the contribution from the first term in (\ref{eq:comm}) can be estimated by 
\[ \begin{split}  \sum_{q, p_1, \dots ,p_{n-1} \in \L^*_+} &|\eta_q| |q|^{2+\eps} |p_1|^{3/2+\delta} \dots |p_{n-1}|^{3/2 + \delta} \\ & \quad \times\| a_q a_{p_1} \dots a_{p_{n-1}} e^{t B(\eta)} \xi \| \| a_{-q}^* a_{p_1} \dots a_{p_{n-1}} e^{t B(\eta)} \xi \| \\  &\leq C \| \eta \|_\infty G_n^{(\eps, \delta)} (t) + C G_n^{(\eps,\delta)} (t)^{\frac{1}{2}} F_{n-1}^{(\delta)} (t)^{\frac{1}{2}} \Big[ \sum_q \eta_q^2 |q|^{2+ \eps} \Big]^{\frac{1}{2}} \\ &\leq C G_n^{(\eps,\delta)} (t) + C \ell^{-1-\eps} F_{n-1}^{(\delta)} (t)\,. \end{split} \]
The contribution from the second term on the r.h.s. of (\ref{eq:comm}), on the other hand, can be bounded by 
\[ \begin{split} \sum_{q, p_1, \dots ,p_{n-2} \in \L^*_+} &|\eta_q | |q|^{7/2+\eps + \delta} |p_1|^{3/2+\delta} \dots |p_{n-2}|^{3/2 + \delta} \\ &\quad \times  \| a_q a_{p_1} \dots a_{p_{n-2}} e^{t B(\eta)} \xi \| \| a_{-q}^* a_{p_1} \dots a_{p_{n-2}} e^{t B(\eta)} \xi \|  \\ &\leq C \Big[ \sup_q |\eta_q| |q|^{3/2+\delta} \Big] G_{n-1}^{(\eps,\delta)} (t) + C G_{n-1}^{(\eps + \theta, \delta)} (t)^{\frac{1}{2}} F_{n-2}^{(\delta)} (t)^{\frac{1}{2}} \Big[ \sum |q|^{5+\eps + 2 \delta -\theta} \eta_q^2 \Big]^{\frac{1}{2}} \\ &\leq C G_{n-1}^{(\eps,\delta)} (t) + C \ell^{-2-\delta-\eps/2 +\theta/2} G_{n-1}^{(\eps + \theta, \delta)} (t)^{\frac{1}{2}} F_{n-2}^{(\delta)} (t)^{\frac{1}{2}} \end{split}  \]
for a $\theta > \eps + 2\delta$ (this condition is needed to apply (\ref{eq:etaHr}), in Lemma \ref{lm:driftscatt}). 

If $n=1$, we use again (\ref{eq:P_r}) to estimate 
\[ G_1^{(\eps,\delta)} (t) \leq C \ell^{-1-\eps} \Big\{ \| \xi \|^2 + \langle \xi, \cS^{(\eps,\delta)}_1 \xi \rangle \Big\} \leq C \ell^{-\beta^\eps_1} \Big\{ \| \xi \|^2 + \sup_{(\eps,\delta) \in \cI_1} \langle \xi, \cS^{(\eps,\delta)}_1 \xi \rangle \Big\} \]
for all $\eps < 3$ ($G_1^{(\eps,\delta)}$ does not depend on $\delta$). Inserting this bound in (\ref{eq:derG}), we arrive at 
\[ \Big| \frac{dG_2^{(\eps,\delta)} (t)}{dt} \Big| \leq C G_2^{(\eps,\delta)} (t) + C \ell^{-5/2-\eps-\delta}  \Big\{ \| \xi \|^2 + \sup_{(\eps,\delta) \in \cI_1} \langle \xi, \cS^{(\eps,\delta)}_1 \xi \rangle \Big\}  \]
if we can find $\theta > 0$ such that $\theta > \eps + 2\delta$ and $\eps + \theta < 3$, i.e. if $\eps + \delta < 3/2$ (this condition is certainly true, if $\eps + 2\delta < 3/2$). By Gronwall's lemma (noticing that $\beta_2^\eps = 5/2+\eps+\delta$), we conclude that 
\[ \begin{split} G^{(\eps,\delta)}_2 (t) &\leq C \ell^{-\beta^\eps_2}\Big\{ \| \xi \|^2 +  \sup_{(\eps,\delta) \in \cI_1} \langle \xi, \cS^{(\eps,\delta)}_1 \xi \rangle  + \sup_{(\eps,\delta) \in \cI_2} \langle \xi, \cS^{(\eps,\delta)}_2 \xi \rangle \Big\} \end{split} \]
for all $\delta \in (0;1/6)$, $\eps \in (-1;3)$ such that $\eps + 2\delta < 3/2$. Now, we proceed by induction. We fix $n \in \bN$ and we assume that for all $j \leq n-1$ there exists a constant $C > 0$ such that 
\[ G_j^{(\eps, \delta)} (t) \leq C \ell^{-\beta^\eps_j} \Big\{ \| \xi \|^2 + \sum_{i=1}^j \sup_{(\eps,\delta) \in \cI_i} \langle \xi, \cS^{(\eps,\delta)}_i \xi \rangle  \Big\}  \]
for all $\delta \in (0;1/6)$ and all $\eps \in (-1;3)$ with $\eps + \delta < 3/2^{(j-1)}$ and all $t \in [0;1]$. Then, using also (\ref{eq:Fbd}), (\ref{eq:derG}) implies that 
\begin{equation}\label{eq:derG2}  \Big| \frac{dG_n^{(\eps,\delta)} (t)}{dt} \Big| \leq C G_n^{(\eps,\delta)} (t) + C \ell^{-\beta^\eps_n} \Big\{ \| \xi \|^2 + \sum_{i=1}^{n-1} \sup_{(\eps,\delta) \in \cI_i} \langle \xi, \cS^{(\eps,\delta)}_i \xi \rangle  \Big\}  \end{equation} 
if we can show that   
\begin{equation} \label{eq:beta-cond} \left\{ \begin{array}{ll} 
\beta_n^{\eps} \geq \beta_{n-1}^{\eps} + \delta -1/2 \\
\beta_n^{\eps} \geq 7/4 + 3\delta/2 + (\beta^{\eps}_{n-1} + \beta^{\eps}_{n-2})/2 \\
\beta_n^{\eps} \geq 1 + \eps + \alpha_{n-1} \\
\beta_n^{\eps} \geq 2 + \delta + \eps/2 - \theta/2 + \beta^{\eps+ \theta}_{n-1}/2 + \alpha_{n-2}/2 \\
\end{array} \right. \end{equation} 
and if we can find $\theta \in \bR$ such that $\theta > \eps + 2 \delta$ and $\eps + \theta + 2\delta <  3/2^{(n-2)}$, ie. if $\eps + 2\delta < 3/2^{(n-1)}$. To verify (\ref{eq:beta-cond}), we use that $\beta_n^{\eps} = \alpha_n + 1/2 + \eps - \delta$. The first and the third conditions in (\ref{eq:beta-cond}) are equivalent to 
\[ \alpha_n \geq \alpha_{n-1} + 1/2 +\delta \]
which follows easily from the explicit formula (\ref{eq:aln2}). The second and the fourth conditions are immediate consequences of the recursive definition (\ref{eq:aln}) of the coefficients $\alpha_n$. From (\ref{eq:derG2}), by Gronwall's lemma we conclude that 
\begin{equation}\label{eq:Gnfin} G_n^{(\eps,\delta)} (t) \leq C \ell^{-\beta^\eps_n} \Big\{ \| \xi \|^2 + \sum_{i=1}^{n} \sup_{(\eps,\delta) \in \cI_i} \langle \xi, \cS^{(\eps,\delta)}_i \xi \rangle  \Big\} \end{equation}
for all $\delta \in (0;1/6), \eps \in (-1;3)$ with $\eps+ 2\delta < 3/2^{(n-1)}$.

Next, we show (\ref{eq:T2}). For $t \in [0;1]$ and for $n = 2,3,4$, we set 
\[ H_n (t)  = \langle e^{t B(\eta)} \xi , \cT_n e^{t B(\eta)} \xi \rangle \, . \] 
Proceeding as in the proof of (\ref{eq:Fd1}), we find 
\begin{equation}\label{eq:derH2} \begin{split} \Big| \frac{dH_2 (t)}{dt} \Big| \leq \; &C \sum_{p,q \in \L^*_+} |\eta_q| \, q^2 p^2 \, \| a_q a_p e^{t B(\eta)} \xi \| \| a_{-q}^* a_p e^{t B(\eta)} \xi \|\\ & + C \sum_{q \in \L^*_+} |\eta_q| \, |q|^4 \| a_q e^{t B(\eta)} \xi \| \| a_{-q}^*e^{t B(\eta)} \xi \|\,. \end{split}  \end{equation} 
Using $\| a_{-q}^* \zeta \| \leq \| a_{-q} \zeta \| + \| \zeta \|$ and Cauchy-Schwarz's inequality we obtain, with (\ref{eq:P_r}) and (\ref{eq:etaHr}),
\[ \Big| \frac{dH_2 (t)}{dt} \Big| \leq C H_2 (t) + C \ell^{-2} \langle \xi, (1+\cP^{(2)} )\xi \rangle + C  \ell^{-3}  \langle \xi , (1+ \cP^{(4)} ) \xi \rangle\,. \]
By Gronwall's lemma, we conclude that 
\begin{equation}\label{eq:H2fin} H_2 (t) \leq C \ell^{-3} \langle \xi, \big( 1 + \cT^{(2)} + \cP^{(4)} \big) \xi \rangle \end{equation} 
for all $t \in [0;1]$. 

Analogously to (\ref{eq:derH2}), we find 
\[ \begin{split} \Big| \frac{dH_3 (t)}{dt} \Big| \leq \; &C \sum_{q,p_1, p_2 \in \L^*_+} |\eta_q | \, q^2 p_1^2 p_2^2 \, \| a_q a_{p_1} a_{p_2} e^{t B(\eta)} \xi \| \| a_{-q}^* a_{p_1} a_{p_2} e^{t B(\eta)} \xi \|\\ & + C \sum_{q, p  \in \L^*_+} |\eta_q| \, |q|^4 p^2 \| a_q a_p e^{t B(\eta)} \xi \| \| a_{-q}^* a_p e^{t B(\eta)} \xi \|\,. \end{split}  \]
Thus 
\begin{equation}\label{eq:derH3} \begin{split} \Big| \frac{dH_3 (t)}{dt} \Big| \leq \; &C H_3 (t) + C \ell^{-1} H_2 (t)  + C \sum_{q,p \in \L^*_+} |\eta_q| |q|^4 p^2 \| a_q a_p e^{t B(\eta)} \xi \| \| a_p e^{t B(\eta)} \xi \|\,. \end{split} \end{equation} 
To control the last term, we distinguish the contribution 
\begin{equation}\label{eq:qpsam} \begin{split} \sum_{q\in \L^*_+} |\eta_q| |q|^6 \| a_q^2  e^{t B(\eta)} \xi \| \| a_{q} e^{t B(\eta)} \xi \| &\leq C  \langle e^{t B(\eta)} \xi , \cP^{(4)} (\cN_+ + 1) e^{t B(\eta)} \xi \rangle \\ &\leq C \ell^{-3} \langle \xi, \cP^{(4)} (\cN_+ + 1) \xi \rangle \end{split} \end{equation} 
arising from terms with $p= q$, a similar contribution from terms with $p=-q$ and the contribution arising from terms with $p \not = -q, q$, which can be bounded, with Cauchy-Schwarz's inequality, by \begin{equation}\label{eq:T3-not} \begin{split} \sum_{q,p \in \L^*_+ : p \not = -q,q} |\eta_q| |q|^4 p^2 \|a_q a_p e^{t B(\eta)} \xi \| \| a_p e^{t B(\eta)} \xi \| &\leq C \ell^{-3/2} W_{4,2}^{1/2} (t)  \langle e^{t B(\eta)} \xi , \cP^{(2)} e^{t B(\eta)} \xi \rangle^{1/2} \\ & \leq W_{4,2} (t) + C \ell^{-4} \langle \xi, ( 1 +\cP^{(2)}) \xi \rangle  \end{split} \end{equation} 
where we applied (\ref{eq:P_r}) and we defined 
\begin{equation}\label{eq:Wdef} W_{4,2} (t) = \sum_{p_1, p_2 \in \L^*_+ : p_1 \not = -p_2, p_2} |p_1|^4 p_2^2 \, \| a_{p_1} a_{p_2} e^{t B(\eta)} \xi \|^2 \,.\end{equation} 
To compute the derivative of $W_{4,2}$, we proceed once again as in (\ref{eq:Fd1}), noticing however that, because of the restriction to $p_1 \not = -p_2, p_2$, the contribution from the second term on the r.h.s. of (\ref{eq:comm}) vanishes. We find, with (\ref{eq:P_r}), 
\[ \begin{split} \Big| \frac{dW_{4,2} (t)}{dt} \Big| &\leq C \sum_{q,p \in \L^*_+ : p \not = -q} | \eta_q | |q|^4 p^2 \| a_q a_p e^{t B(\eta)} \xi \| \| a_{-q}^* a_p e^{tB(\eta)} \xi \| \\ & \hspace{.4cm} + C \sum_{q,p \in \L^*_+ : p \not = -q} | \eta_q | q^2 |p|^4 \| a_q a_p e^{t B(\eta)} \xi \| \| a_{-q}^* a_p e^{tB(\eta)} \xi \| \\ &\leq C W_{4,2} (t) + C \ell^{-3} \langle e^{t B(\eta)} \xi , \cP^{(2)} e^{t B(\eta)} \xi \rangle + C \ell^{-1} \langle e^{t B(\eta)} \xi , \cP^{(4)} e^{t B(\eta)} \xi \rangle \\ &\leq C W_{4,2} (t) + C \ell^{-4} \langle \xi, \big(1 +  \cP^{(4)} \big) \xi \rangle \,. \end{split} \]
By Gronwall's lemma, we conclude (recalling (\ref{eq:Zdef})) that 
\begin{equation} \label{eq:Wfin} W_{4,2} (t) \leq C \ell^{-4} \langle \xi, \big( 1 + \cP^{(4)} + \cZ_{4,2} \big) \xi \rangle \end{equation} 
for all $t \in [0;1]$. Inserting this estimate in (\ref{eq:T3-not}), and then, together with (\ref{eq:qpsam}), in (\ref{eq:derH3}), we obtain (using also that $\cP^{(4)} \cN_+ \leq \cZ_{4,2}$) 
\[ \Big| \frac{dH_3 (t)}{dt} \Big| \leq C H_3 (t) + C \ell^{-1} H_2 (t) + C \ell^{-4} \langle \xi, \big( 1 +  \cP^{(4)} + \cZ_{4,2}  \big)  \xi  \rangle\,. \]
With (\ref{eq:H2fin}) and Gronwall's lemma, we conclude that  
\begin{equation}\label{eq:H3fin} H_3 (t) \leq C \ell^{-4} \langle  \xi, \big( 1 + \cT^{(3)} + \cZ_{4,2} + \cP^{(4)} \big) \xi \rangle \, . \end{equation} 
To control $H_4$, we proceed again as we did to show (\ref{eq:derH3}) and we bound 
\[ \begin{split} 
\Big| \frac{dH_4 (t)}{dt} \Big| \leq \; &C H_4 (t) + C \ell^{-1} H_3 (t) \\ &+ C \sum_{q , p_1, p_2 \in \L^*_+} |\eta_q| |q|^4 p_1^2 p_2^2 \| a_q a_{p_1} a_{p-2} e^{t B(\eta)} \xi \| \| a_{p_1} a_{p_2} e^{t B(\eta)} \xi \| \,.\end{split} \]
In the last term, if $q = \pm p_1$ or $q = \pm p_2$, we find terms that can be bounded using (\ref{eq:etap}) and (\ref{eq:H3fin}) (and the trivial estimate $\cT_2 \cN_+ \leq \cT_3$) by 
\begin{equation}\label{eq:pqsame}  \begin{split} 
\sum_{q , p \in \L^*_+} &|\eta_q| |q|^6 p^2 \| a_q a_{\pm q} a_p e^{t B(\eta)} \xi \| \| a_{\pm q} a_p e^{t B(\eta)} \xi \| \\  &\leq C \ell^{-2} \sum_{q , p \in \L^*_+} q^2 p^2 \| a_q  a_p \cN_+^{1/2} e^{t B(\eta)} \xi \| \| a_{\pm q} a_p e^{t B(\eta)} \xi \| \\ &\leq C \ell^{-2} \langle e^{t B(\eta)} \xi , \cT_2 ( \cN_+  + 1) e^{t B(\eta)} \xi \rangle  \leq C \ell^{-6} \langle \xi, (1 + \cT_3 + \cZ_{4,2} +  \cP^{(4)} ) \xi \rangle \,.
\end{split} \end{equation} 
Contributions from terms with $q \not = \pm p_1, \pm p_2$, on the other hand, can be estimated (with (\ref{eq:H2fin})) by \begin{equation}\label{eq:T4-not} \begin{split} \sum_{q,p_1,p_2 : q \not = \pm p_1, \pm p_2} |\eta_q| |q|^4 p_1^2 p_2^2 \|a_q a_{p_1} &a_{p_2}  e^{t B(\eta)} \xi \| \| a_{p_1} a_{p_2}  e^{t B(\eta)} \xi \| \\ &\leq C \ell^{-3/2}  W_{4,2,2}^{1/2} (t) \langle e^{t B(\eta)} \xi , \cT_2  \, e^{t B(\eta)} \xi \rangle^{1/2} \\ &\leq C W_{4,2,2} (t) + C \ell^{-6} \langle \xi, ( 1 + \cT_2 + \cP^{(4)} ) \xi \rangle  \end{split} \end{equation} 
where we defined 
\[ W_{4,2,2} (t) = \sum_{p_1,p_2,p_3 \in \L^*_+ : p_1 \not = \pm p_2, \pm p_3} |p_1|^4 p_2^2 \, p_3^2 \, \| a_{p_1} a_{p_2} a_{p_3} e^{t B (\eta)} \xi \|^2 \,.\]
We compute 
\[ \begin{split}  \Big| \frac{dW_{4,2,2} (t)}{dt} \Big| \leq \; &C \sum_{q, p_2,p_3 \in \L_+^* : q \not = \pm p_2, \pm p_3} |\eta_q| |q|^4 p_2^2 p_3^2 \| a_q a_{p_1} a_{p_2} e^{t B(\eta)} \xi \| \| a_{-q}^* a_{p_1} a_{p_2} e^{t B(\eta)} \xi \| \\ &+ C \sum_{q,p_1, p_2 \in \L^*_+ : p_1 \not = \pm q, \pm p_2} q^2 |\eta_q| |p_1|^4 p_2^2 \| a_q a_{p_1} a_{p_2} e^{t B(\eta)} \xi \| \| a_{-q}^* a_{p_1} a_{p_2} e^{t B (\eta)} \xi \| \\ &+ C \sum_{q,p \in \L^*_+ : q \not = \pm p} |p|^4 |q|^4 |\eta_q| \, \| a_q a_p e^{t B(\eta)} \xi \| \| a_{-q}^* a_p e^{t B(\eta)} \xi \|   \end{split}\  \]
which leads to 
\begin{equation}\label{eq:W422}  \Big| \frac{dW_{4,2,2} (t)}{dt} \Big| \leq C W_{4,2,2} (t) + C \ell^{-6} \langle \xi, (1 + \cP^{(4)} + \cZ_{4,2}) \xi \rangle + C W_{4,4} (t) \end{equation}  
where 
\[ W_{4,4} (t) = \sum_{q,p \in \L^*_+ : q \not = \pm p} |q|^4 |p|^4 \| a_q a_p e^{t B(\eta)} \xi \|^2 \]
satisfies the estimate 
\[ \begin{split} \Big| \frac{dW_{4,4} (t)}{dt} \Big| &\leq C \sum_{q \not = p} |p|^4 |q|^4 \eta_q \| a_q a_p e^{t B(\eta)} \xi \| \| a_{_q}^* a_p e^{t B(\eta)} \xi \| \\ &\leq C W_{4,4} (t) + C \ell^{-6} \langle \xi, (1 + \cP^{(4)}) \xi \rangle \,.\end{split} \]
Thus, recalling the definition (\ref{eq:Zdef}), we find 
\[ W_{4,4} (t) \leq C \ell^{-6} \langle \xi, (1 + \cP^{(4)}+ \cZ_{4,4}) \xi \rangle \,.\]
Inserting this bound in (\ref{eq:W422}), we obtain 
\[ W_{4,2,2} (t) \leq C \ell^{-6} \langle \xi, (1 + \cP^{(4)} + \cZ_{4,2} + \cZ_{4,2,2} ) \xi \rangle \,. \]
Plugging the last equation in (\ref{eq:T4-not}) and using (\ref{eq:pqsame}), we arrive at
\[ H_4 (t) \leq C \ell^{-6} \langle \xi, (1 + \cP^{(4)} + \cZ_{4,4} + \cZ_{4,2,2} + \cT_4 ) \xi \rangle \,. \]

Finally, we prove (\ref{eq:impr-S3}). For $\eps > -1$, $\delta \in (0;1/6)$ with $\eps + \delta < 1$, we define 
\[ J^{(\eps,\delta)} (t) = \langle e^{t B(\eta)} \xi, \cS^{(\eps,\delta)}_3 e^{t B(\eta)} \xi \rangle \,.\]
Proceeding as in the proof of (\ref{eq:derG}), we find 
\[\begin{split}  \Big| \frac{dJ^{(\eps,\delta)} (t)}{dt} \Big| \leq \; &C J^{(\eps,\delta)} (t) + C \ell^{-1-\eps} F^{(\delta)}_2 (t)  + C \ell^{-1/2-\delta}  G_2^{(\eps,\delta)} (t) \\ &+ \sum_{p,q \in \L^*_+} |\eta_q| |q|^{7/2+\eps + \delta} |p|^{3/2+\delta} \| a_p a_q e^{t B(\eta)} \xi \| \| a_p e^{tB(\eta)} \xi \| \\ &+  \sum_{p,q \in \L^*_+} |\eta_q| |q|^{3+ 2\delta} |p|^{2+\eps} \| a_p a_q e^{t B(\eta)} \xi \| \| a_p e^{tB(\eta)} \xi \| \,.\end{split} \] 
 Recalling the definition (\ref{eq:Wdef}), we can estimate (distinguishing $p = q$ from $p \not = q$) 
 \[ \begin{split} \sum_{p,q \in \L^*_+} &|\eta_q| |q|^{7/2+\eps + \delta} |p|^{3/2+\delta} \| a_p a_q e^{t B(\eta)} \xi \| \| a_p e^{tB(\eta)} \xi \| \\  \leq \; &C \langle e^{t B(\eta)} \xi , \cP^{(3+\eps+2\delta)} (\cN_+ + 1)  e^{tB(\eta)} \xi \rangle \\ &+ C W_{4,2} (t)^{\frac{1}{2}} \Big( \sum_{q \in \L^*_+} \eta_q^2 |q|^{3+2\eps+2\delta} \Big)^{\frac{1}{2}}   \langle e^{t B(\eta)} \xi , \cP^{(1+2\delta)} e^{t B(\eta)} \xi \rangle^{\frac{1}{2}} 
  \\ \leq \; &C \langle e^{t B(\eta)} \xi , \cP^{(3+\eps+2\delta)} e^{tB(\eta)} \xi \rangle + C \ell^{-1-\eps-\delta} W_{4,2} (t)^{\frac{1}{2}} \langle e^{t B(\eta)} \xi , \cP^{(1+2\delta)} e^{t B(\eta)} \xi \rangle^{\frac{1}{2}} 
  \end{split} \]
 and, similarly,  
 \[ \begin{split} \sum_{p,q \in \L^*_+} &|\eta_q| |q|^{3+ 2\delta} |p|^{2+\eps} \| a_p a_q e^{t B(\eta)} \xi \| \| a_p e^{tB(\eta)} \xi \| \\  \leq \; &C \langle e^{t B(\eta)} \xi , \cP^{(3+\eps+2\delta)} (\cN_+ + 1)  e^{tB(\eta)} \xi \rangle + C \ell^{-\frac{1}{2} - 2\delta} W_{4,2} (t)^{\frac{1}{2}} \langle e^{t B(\eta)} \xi , \cP^{(2+2\eps)} e^{t B(\eta)} \xi \rangle^{\frac{1}{2}}\,. \end{split} \]
With Lemma \ref{lm:N-Pgrow}, with (\ref{eq:Fbd}), (\ref{eq:Gnfin}) and (\ref{eq:Wfin}), we conclude that
\[ \begin{split}    \Big| \frac{dJ^{(\eps,\delta)} (t)}{dt} \Big| \leq \; &C J^{(\eps,\delta)} (t) \\ &+ C \ell^{-3-\eps- 2\delta} \Big\{ \langle \xi, \big(1+ \cP^{(4)} +  \cZ_{4,2} \big) \xi \rangle + \sup_{(\eps, \delta) \in \cI_2} \langle \xi, \cS_2^{(\eps,\delta)} \xi \rangle \Big\}  \end{split}\]
 for all $t \in [0;1]$. By Gronwall's lemma, we obtain (\ref{eq:impr-S3}). 
\end{proof}

\section{Proof of Theorem \ref{thm:main}}\label{sec:main}

With the unitary operator $U_N$ as in (\ref{eq:UNdef}), with $\eta$ as introduced after (\ref{eq:g0}) and $\tau$ as in (\ref{eq:taup}), we define $\Phi_N \in L^2_s (\L^N)$ setting 
\begin{equation}\label{eq:PhiN-def} \Phi_N = U_N^* e^{B(\eta)} e^{B(\tau)} \Omega \,. \end{equation} 
We recall that we assumed $N^{-1+\nu}\leq\ell\leq N^{-3/4-\nu}$ (see Prop. \ref{prop:eff}) and $\ell_0>0$ small enough (independent of $N$).
From Prop. \ref{prop:cM}, we find that 
\begin{equation}\label{eq:en-Phi}
\begin{split} 
\langle \Phi_N, H_N^\text{eff} \Phi_N \rangle &= \langle \Omega, \cM^\text{eff}_{N,\ell} \Omega \rangle \\ &=  4\pi \frak{a} (N-1) + e_\Lambda \frak{a}^2 \\ &\hspace{.4cm} -\frac{1}{2} \sum_{p \in \L^*_+} \Big[ p^2 + 8\pi \frak{a} - \sqrt{|p|^4 + 16 \pi \frak{a} p^2} - \frac{(8\pi \frak{a})^2}{2p^2} \Big] + \cO (N^{-\eps}) \end{split}
\end{equation}
for a sufficiently small $\eps > 0$. 

Additionally, with Lemma \ref{lm:regu} we obtain important regularity estimates for $\Phi_N$. 
From (\ref{eq:T2}) (and from (\ref{eq:P_r}) in Lemma \ref{lm:N-Pgrow}), we find $C > 0$ such that  
\begin{equation}\label{eq:apri1}  \begin{split} 
\langle \Phi_N, (-\Delta_{x_1}) \Phi_N \rangle &\leq \frac{C}{N\ell} \\ \langle \Phi_N, (-\Delta_{x_1})(-\Delta_{x_2}) \Phi_N \rangle &\leq \frac{C}{N^2 \ell^3}  \\ \langle \Phi_N, (-\Delta_{x_1})(- \Delta_{x_2})(- \Delta_{x_3})  \Phi_N \rangle &\leq \frac{C}{N^3 \ell^4} \\
\langle \Phi_N, (-\Delta_{x_1})(- \Delta_{x_2})(- \Delta_{x_3})(-\Delta_{x_4}) \Phi_N \rangle &\leq \frac{C}{N^4 \ell^6}\,.
 \end{split} \end{equation} 
From (\ref{eq:Adn}) we find, for $n \in \bN$ and $0 < \delta < 1/6$, a constant $C > 0$ such that 
\begin{equation}\label{eq:apri2} \begin{split} 
\langle \Phi_N, (-\Delta_{x_1})^{3/4+\delta/2} \dots (-\Delta_{x_n})^{3/4+ \delta/2} \Phi_N \rangle &\leq \frac{C}{N^n \ell^{\alpha_n}} \,.
\end{split} \end{equation} 
From (\ref{eq:Sedn}) in Lemma \ref{lm:regu}, we find, for $n \in \bN$ and for every $\eps \in (-1;3)$, $\delta \in (0;1/6)$ such that $\eps + 2\delta < 3/2^{n-1}$, a constant $C > 0$ such that 
\begin{equation}\label{eq:apri3} \begin{split} 
\langle \Phi_N, (-\Delta_{x_1})^{1+\eps/2} (- \Delta_{x_2})^{3/4+\delta/2}  \dots (-\Delta_{x_n})^{3/4+\delta/2} \Phi_N \rangle &\leq \frac{C}{N^n \ell^{\beta^\eps_n}} \,.
\end{split} 
\end{equation}

Let us prove (\ref{eq:apri3}), the other bounds can be shown similarly. First of all, we symmetrize the expectation on the l.h.s. of (\ref{eq:apri3}), writing 
\[ \begin{split} \langle \Phi_N, &(-\Delta_{x_1})^{1+\eps/2} (- \Delta_{x_2})^{3/4+\delta/2}  \dots (-\Delta_{x_n})^{3/4+\delta/2} \Phi_N \rangle \\ &= \frac{1}{{N \choose n}} \sum_{1 \leq i_1 < \dots < i_n \leq N} \langle \Phi_N, (-\Delta_{x_{i_1}})^{1+\eps/2} (- \Delta_{x_{i_2}})^{3/4+\delta/2}  \dots (-\Delta_{x_{i_n}})^{3/4+\delta/2} \Phi_N \rangle \,. \end{split} \]
Next, we express the observable in second quantized form and we apply the rules (\ref{eq:U-rules}). We find
\[  \begin{split} \langle \Phi_N, &(-\Delta_{x_1})^{1+\eps/2} (- \Delta_{x_2})^{3/4+\delta/2}  \dots (-\Delta_{x_n})^{3/4+\delta/2} \Phi_N \rangle \\ &\leq \frac{C}{N^n} \sum_{p_1, \dots , p_n \in \L^*_+} |p_1|^{2+\eps} |p_2|^{3/2+\delta} \dots |p_n|^{3/2+\delta} \\ &\hspace{3cm} \times \langle e^{B(\eta)} e^{B(\tau)} \Omega , a_{p_1}^* \dots a_{p_n}^* a_{p_n} \dots a_{p_1} e^{B(\eta)} e^{B(\tau)} \Omega \rangle \,.\end{split} \]
With (\ref{eq:Sedn}), we conclude that 
\[  \begin{split} \langle \Phi_N, &(-\Delta_{x_1})^{1+\eps/2} (- \Delta_{x_2})^{3/4+\delta/2}  \dots (-\Delta_{x_n})^{3/4+\delta/2} \Phi_N \rangle \\ &\hspace{3cm} \leq \frac{C}{N^n \ell^{\beta_n^\eps}}  \Big\{ 1 + \sum_{j=1}^n \sup_{\eps, \delta \in \cI_j} \langle e^{B(\tau)} \Omega, S_j^{(\eps,\delta)} e^{B(\tau)} \Omega \rangle \Big\}\,.  \end{split} \]
To control the growth of $S_j^{(\eps,\delta)}$, we can proceed exactly as in the proof of Lemma \ref{lm:regu}; the difference is that, by (\ref{eq:tau-dec}), $|\tau_p| \leq C /|p|^4$, uniformly in $N,\ell$
(this should be compared with the bound (\ref{eq:etap}), for the coefficients $\eta_p$). As a consequence, for $0 < r< 5$, we find 
\[ \sum_{p \in \L^*_+} |p|^r |\tau_p|^2 \leq C \]
and thus the analog of the bounds in Lemma \ref{lm:regu}, with $B(\eta)$ replaced by $B(\tau)$, 
holds uniformly in $\ell$. This observation leads to (\ref{eq:apri3}). 

With $\Phi_N$ as in (\ref{eq:PhiN-def}), we define the trial function $\Psi_N \in L^2_s (\L^N)$ by 
\[ \Psi_N ({\bf x}) = \Phi_N ({\bf x}) \cdot \prod_{i<j}^N f_\ell (x_i - x_j)\,. \] 
The presence of the Jastrow factor guarantees that $\Psi_N$ satisfies the hard-sphere condition \eqref{eq:hard}.  Combining (\ref{eq:en-psi}), Prop.  \ref{lm:3body} and Prop. \ref{prop:eff}, we obtain 
\begin{equation}\label{eq:en-Psi1} \begin{split} 
&\frac{\langle \Psi_N , \sum_{j=1}^N -\Delta_{x_j} \Psi_N \rangle}{\| \Psi_N \|^2}  \\ & \leq \langle \Phi_N , H_N^\text{eff} \Phi_N \rangle - \frac{N(N-1)}{2} \langle \Phi_N, \Big\{  \big[ H_{N-2}^\text{eff} - 4 \pi \frak{a} N \big] \otimes u_\ell (x_{N-1} - x_N) \Big\} \Phi_N \Big\rangle + C N^{-\eps} \,.\end{split} \end{equation}
Here we used (\ref{eq:apri1}), (\ref{eq:apri2}) and (\ref{eq:apri3}) to verify the assumption (\ref{eq:3D-assum}) of Prop. \ref{lm:3body} and the assumption (\ref{eq:Delta-phi}) for Prop. \ref{prop:eff}. Moreover, we used (\ref{eq:en-Phi}) to verify the condition $\langle \Phi_N , H_N^\text{eff} \Phi_N \rangle \leq 4 \pi \frak{a} N + C$ in Prop. \ref{prop:eff}. 

Inserting (\ref{eq:en-Phi}) on the r.h.s. of (\ref{eq:en-Psi1}), we arrive at 
\begin{equation}\label{eq:en-Psi2} \begin{split} &\frac{\langle \Psi_N , \sum_{j=1}^N -\Delta_{x_j} \Psi_N \rangle}{\| \Psi_N \|^2} \\ &\hspace{.4cm}  \leq 4\pi \frak{a} (N-1) + e_\Lambda \frak{a}^2 -\frac{1}{2} \sum_{p \in \L^*_+} \Big[ p^2 + 8\pi \frak{a} - \sqrt{|p|^4 + 16 \pi \frak{a} p^2} - \frac{(8\pi \frak{a})^2}{2p^2} \Big] \\ &\hspace{.8cm} - \frac{N(N-1)}{2} \langle \Phi_N, \Big\{  \big[ H_{N-2}^\text{eff} - 4 \pi \frak{a} N \big] \otimes u_\ell (x_{N-1} - x_N) \Big\} \Phi_N \Big\rangle + C N^{-\eps}  \,.
\end{split}
\end{equation}
To conclude the proof of Theorem \ref{thm:main}, we still have to show that the contribution on the last line is negligible, in the limit $N \to \infty$. 

From (\ref{eq:GN-JN-LB}) in Prop. \ref{prop:cGcJ}, we find
\begin{equation} \label{eq:diff-en} H_{N-2}^\text{eff} - 4 \pi \frak{a} N \geq U_{N-2}^* e^{B(\eta)} \Big\{ - C (\cN_+ + 1) - C N^{-\kappa} \cP^{(2+\kappa)} (\cN_+ + 1) \Big\} e^{-B(\eta)} U_{N-2} \end{equation} 
for $0< \kappa < \nu /2$. Notice here that both sides of the equation are operators on the Hilbert space $L^2_s (\L^{N-2})$ describing states with $(N-2)$ particles.  

For $\mu> 0$ to be chosen small enough, we can estimate 
\begin{equation}\label{eq:markov} (\cN_+ + 1) \leq C N^\mu + C (\cN_+ + 1) \chi (\cN_+ \geq N^\mu) \leq C N^\mu + C N^{-m\mu} (\cN_+ + 1)^{m+1} \end{equation} 
for any $m \in \bN$. Thus, the contribution arising from the first term in the parenthesis on the r.h.s. of (\ref{eq:diff-en}) can be bounded by  
\[ \begin{split} &\frac{N (N-1)}{2} \Big\langle \Phi_N ,  \Big\{ \Big[ U_{N-2}^* e^{B(\eta)} (\cN_+ + 1) e^{-B(\eta)} U_{N-2} \Big]  \otimes u_\ell (x_{N-1} - x_N) \Big\} \Phi_N \Big\rangle  \\ &\leq C N^{2+\mu} \langle \Phi_N , u_\ell (x_{N-1} - x_N) \Phi_N \rangle \\ &\hspace{.3cm} + C N^{2-m\mu} \Big\langle \Phi_N , \Big\{ \Big[ U_{N-2}^* e^{B(\eta)} (\cN_+ + 1)^{m+1} e^{-B(\eta)} U_{N-2} \Big] \otimes u_\ell (x_{N-1} - x_N) \Big\} \Phi_N \Big\rangle\,. \end{split}  \]
Using $\| u_\ell \|_1 \leq C\ell^2 / N$ and \eqref{eq:W1-L1} in the first and $\| u_\ell \|_\infty \leq C$ in the second term (by Lemma \ref{lm:hardcorescatt}), we obtain 
\[ \begin{split} \frac{N (N-1)}{2} \Big\langle \Phi_N ,  \Big\{ \Big[ U_{N-2}^* e^{B(\eta)} &(\cN_+ + 1) e^{-B(\eta)} U_{N-2} \Big] \otimes u_\ell (x_{N-1} - x_N) \Big\} \Phi_N \Big\rangle  \\ &\leq C N^{1+\mu} \ell^2 \langle \Phi_N , (1-\Delta_{x_1}) ( 1-\Delta_{x_2}) \Phi_N \rangle \\ &\hspace{.3cm} + C N^{2-m\mu} \Big\langle e^{B(\eta)} e^{B(\tau)} \Omega, (\cN_+ + 1)^{m+1} e^{B(\eta)} e^{B(\tau)} \Omega \Big\rangle \,.
 \end{split}  \]
Here we used Lemma \ref{lm:N-Pgrow} to control the growth of $(\cN_+ +1)^{m+1}$ under the action of $B(\eta)$. Moreover, with $\frak{q} = 1 - |\ph_0 \rangle \langle \ph_0|$ denoting the projection onto the orthogonal complement to the condensate wave function $\ph_0$ in $L^2 (\L)$ and with $\frak{q}_j = 1 \otimes \dots \otimes \frak{q} \otimes \dots \otimes 1$ acting as $\frak{q}$ on the $j$-th particle, we estimated, on the $N$-particle space $L^2_s (\L^N)$,  
\begin{equation}\label{eq:exte} U_{N-2}^* \, \cN_+ U_{N-2} \otimes 1 = \sum_{j=1}^{N-2} \frak{q}_j \leq \sum_{j=1}^N \frak{q}_j = U^*_N \cN_+ U_N \end{equation} 
(with a slight abuse of notation, $\cN_+$ denotes the number of particles operators on 
$\cF^{\leq (N-2)}_+$ on the l.h.s. and the number of particles operator on $\cF^{\leq N}_+$ on the r.h.s.). Using again Lemma \ref{lm:N-Pgrow} (and Lemma \ref{lm:action-tau}, for the action of $B(\tau)$), together with the bounds in (\ref{eq:apri1}), we conclude that 
\begin{equation} \label{eq:UNUfin} \begin{split} \frac{N (N-1)}{2} \Big\langle \Phi_N ,  \Big\{ \Big[ U_{N-2}^* e^{B(\eta)} (\cN_+ + 1) &e^{-B(\eta)} U_{N-2} \Big] \otimes u_\ell (x_{N-1} - x_N) \Big\} \Phi_N \Big\rangle \\ &\leq 
N^{1+\mu} \ell^2  \Big( 1 + \frac{1}{N^2\ell^3} \Big) + C N^{2-m\mu} \leq C N^{-\eps} \end{split} \end{equation} 
choosing first $\mu > 0$ small enough and then $m \in \bN$ large enough.

Let us now focus on the contribution of the second term in the parenthesis on the r.h.s. of (\ref{eq:diff-en}). Also here, we use (\ref{eq:markov}) to estimate
\[ \begin{split} 
&\frac{N^{1-\kappa} (N-1)}{2} \Big\langle \Phi_N , \Big\{ \Big[ U_{N-2}^* e^{B(\eta)} \cP^{(2+\kappa)} (\cN_+ + 1) 
e^{-B(\eta)} U_{N-2} \Big] \otimes u_\ell (x_{N-1} - x_N) \Big\} \Phi_N \Big\rangle \\ &\leq C N^{2-\kappa+\mu} \big\langle \Phi_N,  \big\{ \big[ U_{N-2}^* e^{B(\eta)} \cP^{(2+\kappa)} e^{-B(\eta)} U_{N-2} \big] \otimes u_\ell (x_{N-1} - x_N) \big\} \Phi_N \big\rangle \\ &\hspace{.2cm} + C N^{2-\kappa -m \mu} \big\langle \Phi_N, \big\{ \big[ U_{N-2}^* e^{B(\eta)} \cP^{(2+\kappa)} (\cN_+ + 1)^{m+1} e^{-B(\eta)} U_{N-2} \big] \otimes u_\ell (x_{N-1} - x_N) \big\} \Phi_N \big\rangle \\ &= \text{R}_1 + \text{R}_2 \,.\end{split} \]
To bound $\text{R}_2$, we can estimate $\| u_\ell \|_\infty \leq C$, we can apply Lemma \ref{lm:N-Pgrow} to control the growth of $\cP^{(2+\kappa)} (\cN_+ + 1)^{m+1}$ under conjugation with $e^{B(\eta)}$ and we can proceed similarly as in (\ref{eq:exte}) to replace $U_{N-2}$ with $U_N$. 
We find 
\[ \text{R}_2 \leq CN^{2-\kappa-m\mu} \big\langle e^{B(\eta)} e^{B(\tau)} \Omega, \big( \cP^{(2+\kappa)} + \ell^{-1-\kappa} \big) (\cN_+ + 1)^{m+1} e^{B(\eta)} e^{B(\tau)} \Omega \big\rangle \,. \]
Applying again Lemma \ref{lm:N-Pgrow} (and then Lemma \ref{lm:action-tau} for the action of $B(\tau)$), we conclude that 
\begin{equation}\label{eq:R2fin} \text{R}_{2} \leq C N^{2-\kappa-m\mu} \ell^{-1-\kappa}\,. \end{equation} 
As for the term $\text{R}_1$, we first use (\ref{eq:W1-L1}) in Lemma \ref{lemma:sobolev_bounds} to estimate, for $\delta > 0$ small enough, 
\begin{equation}\label{eq:R1112} \begin{split} 
\text{R}_1 \leq \; &C N^{2-\kappa+\mu} \| u_\ell \|_1 \big\langle \Phi_N , \big\{ \big[ 
U_{N-2}^* e^{B(\eta)} \cP^{(2+\kappa)} e^{-B(\eta)} U_{N-2} \big] \otimes 1 \big\} \Phi_N \big\rangle 
\\ &+  C N^{2-\kappa+\mu} \| u_\ell \|_1 \\ &\hspace{.5cm} \times 
\big\langle \Phi_N , \big\{ \big[ 
U_{N-2}^* e^{B(\eta)} \cP^{(2+\kappa)} e^{-B(\eta)} U_{N-2} \big] \otimes (\Delta_{x_{N-1}}\Delta_{x_N})^{3/4+\delta/2}  \big\} \Phi_N \big\rangle  \\ = \; &\text{R}_{11} + \text{R}_{12} .\end{split} \end{equation}
To control $\text{R}_{12}$, we apply Lemma \ref{lm:N-Pgrow} to bound 
\[ \begin{split} U_{N-2}^* e^{B(\eta)} \cP^{(2+\kappa)} e^{-B(\eta)} U_{N-2} &\leq C U_{N-2}^* \big[ \cP^{(2+\kappa)} + \ell^{-1-\kappa} \big] U_{N-2} \\ &= C \Big[ \sum_{j=1}^{N-2} (-\Delta_{x_j})^{1+\kappa/2}  + \ell^{-1-\kappa} \Big] \,. \end{split} \]
Thus 
\[ \begin{split} 
\text{R}_{12} \leq C N^{-1-\kappa+\mu} \ell^2 \big\langle e^{B(\eta)} e^{B(\tau)} \Omega, \big[ S^{(\kappa, \delta)}_3 + \ell^{-1-\kappa}  \cA^{(\delta)}_2 \big]  e^{B(\eta)} e^{B(\tau)} \Omega \big\rangle \,.\end{split} \]
With (\ref{eq:Adn}) and with (\ref{eq:impr-S3}) from Lemma \ref{lm:regu}, we conclude that 
\begin{equation}\label{eq:R12fin} \text{R}_{12} \leq  C N^\mu N^{-\wt{\eps}} \end{equation}
for some $\wt{\eps} > 0$, if $\delta$ is chosen small enough, and $0< \kappa < \nu/2$. 

It turns out that the term $\text{R}_{11}$ is more subtle; here we cannot afford the error arising from conjugation of $\cP^{(2+\kappa)}$ with $e^{-B(\eta)}$. Instead, we have to use the fact that we conjugate back with $e^{B(\eta)}$ when we take expectation in the state $\Phi_N = e^{B(\eta)} e^{B(\tau)} \Omega$. The two generalized Bogoliubov transformations do not cancel identically (because one acts on $(N-2)$ particles, the other on $N$), but of course their combined action produce a much smaller error. We will make use of the following lemma. 
\begin{lemma} \label{lemma:Pr-ex}
For $r\in(1;4]$ we have 
\begin{equation} \label{eq:exact_action_P}
e^{-B(\eta)} \mathcal{P}^{(r)}e^{B(\eta)}=  \mathcal{P}^{(r)}+ \sum_{p \in \L^*_+} |p|^r \eta_p\left( b^*_p b^*_{-p}+\mathrm{h.c.} \right)+\sum_{p \in \L^*_+} |p|^r \eta_p^2 + \mathcal{X}_1
\end{equation}
with 
\begin{equation*}\label{eq:cX}
\pm \mathcal{X}_1 \le C (\cN_+ + 1) + C N^{-1}  \big( \cP^{(r)} + \ell^{1-r} \big) (\cN_+ + 1) .
\end{equation*}
Moreover, 
\begin{equation} \label{eq:exact_action_offdiag}
e^{-B(\eta)} \sum_{p \in \L^*_+} |p|^r \eta_p(b^*_p b_{-p}^* + \mathrm{h.c.}) e^{B(\eta)} = \sum_{p \in 
\L^*_+} |p|^r \eta_p \big[ b^*_p b_{-p}^*+b_p b_{-p} \big]+2 \sum_{p \in \L^*_+} |p|^r \eta_p^2+ 
\mathcal{X}_2
\end{equation}
with 
\begin{equation*}
\pm \mathcal{X}_2 \le C (\cN_+ + 1) + C N^{-1}  \big( \cP^{(r)} + \ell^{1-r} \big) (\cN_+ + 1). \end{equation*}
\end{lemma}

We defer the proof of Lemma \ref{lemma:Pr-ex} to the end of the section, showing first how it can be used to estimate the error $\text{R}_{11}$ and to conclude the proof of Theorem \ref{thm:main}. Notice first that $2+\kappa\le4$ since $\kappa<\nu/2$ and $\nu$ is small enough. We can therefore apply Lemma \ref{lemma:Pr-ex} to find 
\[ \begin{split} 
 \langle \Phi_N , &\big\{ \big[ U_{N-2}^* e^{B(\eta)} \cP^{(2+\kappa)} e^{-B(\eta)} U_{N-2} \big] \otimes 1 \big\} \Phi_N \rangle \\ \leq \; &\big\langle \Phi_N \big\{ U_{N-2}^* \cP^{(2+\kappa)} U_{N-2} \otimes 1 \big\} \Phi_N \rangle + \sum_{p \in \L^*_+} |p|^{2+\kappa} \eta_p^2 \\ &+ \sum_{p \in \L^*_+} |p|^{2+\kappa} \eta_p  \big\langle \Phi_N \big\{ U_{N-2}^* \big[ b_p^* b_{-p}^* + \text{h.c.} \big] U_{N-2} \otimes 1 \big\} \Phi_N \rangle \\ &+ C   \langle \Phi_N , \big\{  U_{N-2}^* \big[ 1+N^{-1}\big(\mathcal{P}^{(2+\k)} + \ell^{-1-\k}\big)\big] (\mathcal{N}_++1)  U_{N-2} \otimes 1\big\} \Phi_N \rangle\,. 
 \end{split} \]
We observe that 
\[ U_{N-2}^* \cP^{(2+\kappa)} U_{N-2} \otimes 1 = \sum_{j=1}^{N-2} (-\Delta_{x_j})^{2+\kappa} \leq \sum_{j=1}^{N} (-\Delta_{x_j})^{2+\kappa} = U_N^* \cP^{(2+\kappa)} U_N \]
and that, similarly,  
\[\begin{split}  U_{N-2}^*\big[ 1+N^{-1}\big( \mathcal{P}^{(2+\k)} &+ \ell^{-1-\k}\big)\big](\mathcal{N}_++1) U_{N-2}\\\leq\;& U_N^*\big[ 1+N^{-1}\big(\mathcal{P}^{(2+\k)} + \ell^{-1-\k} \big)\big](\mathcal{N}_++1) U_N \,.\end{split} \]

Moreover, we find 
\[ \begin{split} \sum_{p \in \L^*_+} |p|^{2+\kappa} \eta_p U_{N-2}^* \big[ b_p^* b_{-p}^* + \text{h.c.} \big] U_{N-2}  &= \frac{1}{N-2} \sum_{p \in \L^*_+} |p|^{2+\kappa} \eta_p \big[ a_p^* a_{-p}^* a_0 a_0 + \text{h.c.} \big] \\ &= \frac{1}{N-2} \sum_{i<j}^{N-2} \big[ \theta (x_i - x_j) (\frak{p}_i \otimes \frak{p}_j) + \text{h.c.} \big] \\ &=  \frac{1}{N-2} \sum_{i<j}^{N} \big[ \theta (x_i - x_j) (\frak{p}_i \otimes \frak{p}_j) + \text{h.c.} \big] \\ &\hspace{.4cm} - \frac{1}{N-2} \sum_{\substack{i< j : \\ j = N-1, N}} \big[ \theta (x_i - x_j) (\frak{p}_i \otimes \frak{p}_j) + \text{h.c.} \big] \end{split} \]
with $\theta$ defined by the Fourier coefficients $\hat\theta_p = |p|^{2+\kappa} \eta_p$, and with $\frak{p}_j$ denoting the orthogonal projection $\frak{p} = |\ph_0 \rangle \langle \ph_0|$ on the condensate wave function acting on the $j$-particle. Rewriting the first term in second quantized form (but now, on the $N$-particle space), we find 
\[ \begin{split} 
\sum_{p \in \L^*_+} |p|^{2+\kappa} \eta_p U_{N-2}^* \big[ &b_p^* b_{-p}^* + \text{h.c.} \big] U_{N-2} \\
&=  \frac{N}{N-2} 
\sum_{p \in \L^*_+} |p|^{2+\kappa} \eta_p U_N^* \big[ b_p^* b_{-p}^* + \text{h.c.} \big] U_N \\ &\hspace{.4cm} 
- \frac{1}{N-2} \sum_{i < j : j=N-1, N}  \big[ \theta (x_i - x_j) (\frak{p}_i \otimes \frak{p}_j) + \text{h.c.} \big]\,.
 \end{split} \]
Therefore, we find 
\begin{equation} \label{eq:P-exp1} \begin{split} 
\big\langle  \Phi_N , \big\{ \big[& U_{N-2}^* e^{B(\eta)} \cP^{(2+\kappa)} e^{-B(\eta)} U_{N-2} \big] \otimes 1 \big\} \Phi_N \big\rangle \\ \leq \; &\big\langle e^{B(\eta)} e^{B(\tau)} \Omega, \cP^{(2+\kappa)} e^{B(\eta)} e^{B(\tau)} \Omega \big\rangle  + \sum_{p \in \L^*_+} |p|^{2+\kappa} \eta_p^2 \\ &- \frac{N}{N-2} \sum_{p \in \L^*_+} |p|^{2+\kappa} \eta_p \big\langle e^{B(\eta)} e^{B(\tau)} \Omega,  \big( b_p^* b_{-p}^* + b_p b_{-p} \big)  e^{B(\eta)} e^{B(\tau)} \Omega \big\rangle \\ & + \frac{1}{N-2}  \sum_{i < j : j=N-1, N}  \big\langle \Phi_N , \big[ \theta (x_i - x_j) (\frak{p}_i \otimes \frak{p}_j) + \text{h.c.} \big] \Phi_N \rangle  \\
& +C  \langle e^{B(\eta)} e^{B(\tau)} \Omega ,\big[ 1+N^{-1}\big( \mathcal{P}^{(2+\k)} + \ell^{-1-\k} \big)\big](\mathcal{N}_++1)e^{B(\eta)} e^{B(\tau)} \Omega\rangle \,. \end{split} \end{equation} 
Applying again Lemma \ref{lemma:Pr-ex} to the first and third terms on the r.h.s. of (\ref{eq:P-exp1}), and Lemma \ref{lm:N-Pgrow} to the last, we obtain 
\begin{equation}\label{eq:P-exp2} \begin{split} 
\big\langle \Phi_N , \big\{ \big[ U_{N-2}^* e^{B(\eta)} &\cP^{(2+\kappa)} e^{-B(\eta)} U_{N-2} \big] \otimes 1 \big\} \Phi_N \big\rangle \\ \leq \; &\big\langle e^{B(\tau)} \Omega, \cP^{(2+\kappa)} 
e^{B(\tau)} \Omega \big\rangle + \left[ 2 - \frac{2N}{N-2} \right] \sum_{p \in \L^*_+} |p|^{2+\kappa} \eta_p^2 \\ &- \frac{2}{N-2} \sum_{p \in \L^*_+} |p|^{2+\kappa} \eta_p 
\langle e^{B(\tau)} \Omega, \big[ b_p^* b_{-p}^* + b_p b_{-p} \big] e^{B(\tau)} \Omega \rangle \\
&+ \frac{1}{N-2}  \sum_{i < j : j=N-1, N}  \big\langle \Phi_N , \big[ \theta (x_i - x_j) (\frak{p}_i \otimes \frak{p}_j) + \text{h.c.} \big] \Phi_N \big\rangle  \\
& + C \langle e^{B(\tau)} \Omega, \big[ 1+N^{-1}\big(\mathcal{P}^{(2+\k)} + \ell^{-1-\k} \big)\big](\mathcal{N}_++1)  e^{B(\tau)} \Omega\rangle \,. \end{split} \end{equation} 
With the properties of $\tau$ (see Lemma \ref{lm:action-tau}) it is easy to check that all expectations in the state $e^{B(\tau)} \Omega$ are bounded, uniformly in $N,\ell$. Moreover, by (\ref{eq:etaHr}), we find 
\[ \left[ 2 - \frac{2N}{N-2} \right] \sum_{p \in \L^*_+} |p|^{2+\kappa} \eta_p^2 \leq \frac{C}{N \ell^{1+\kappa}} \,. \]
Finally, we can estimate the term on the fourth line in (\ref{eq:P-exp2}) by  
\[ \Big| \frac{1}{N-2}  \sum_{i < j : j=N-1, N}  \big\langle \Phi_N , \big[ \theta (x_i - x_j) (\frak{p}_i \otimes \frak{p}_j) + \text{h.c.} \big] \Phi_N \rangle \Big| \leq C \big| \langle \Phi_N , \theta (x_1 - x_2) (\frak{p}_1 \otimes \frak{p}_2)  \Phi_N \rangle \big| \,.\]
Since $\frak{p}_1 \theta (x_1 - x_2) \frak{p}_1 = \frak{p}_1 \hat{\theta}_0 = 0$ and, similarly, $\frak{p}_2 \theta (x_1 - x_2) \frak{p}_2 = 0$, we have 
\[ \begin{split} \big| \langle \Phi_N , \theta (x_1 - x_2) (\frak{p}_1 \otimes \frak{p}_2)  \Phi_N \rangle \big| &= \big| \langle \Phi_N , (\frak{q}_1 \otimes \frak{q}_2) \theta (x_1 - x_2) (\frak{p}_1 \otimes \frak{p}_2)  \Phi_N \rangle \big| \\ &\leq \| \theta \|_2 \| (\frak{q}_1 \otimes \frak{q}_2) \Phi_N \|  \| \Phi_N \| \,.\end{split} \]
With $\| \theta \|_2 = \| \hat{\theta} \|_2 \leq C \ell^{-3/2 -\kappa}$, for $0< \kappa < 1/2$, and with 
\[  \begin{split} 
\| (\frak{q}_1 \otimes \frak{q}_2) \Phi_N \|^2 &\leq C N^{-2} \big\langle \Phi_N, \big[ \sum_{i=1}^N \frak{q}_i \big]^2 \Phi_N \big\rangle \\ &= C N^{-2} \langle e^{B(\eta)} e^{B(\tau)} \Omega, (\cN_+ + 1)^2 e^{B(\eta)} e^{B(\tau)} \Omega \rangle \leq C N^{-2} \end{split} \]
we conclude that 
\[  \Big| \frac{1}{N-2}  \sum_{i < j : j=N-1, N}  \big\langle \Phi_N , \big[ \theta (x_i - x_j) (\frak{p}_i \otimes \frak{p}_j) + \text{h.c.} \big] \Phi_N \rangle \Big| \leq \frac{C}{N \ell^{3/2+\kappa}}\,.\ \]
Therefore, we obtain 
\[ \Big| \big\langle \Phi_N , \big\{ \big[ U_{N-2}^* e^{B(\eta)} \cP^{(2+\kappa)} e^{-B(\eta)} U_{N-2} \big] \otimes 1 \big\} \Phi_N \big\rangle \Big| \leq \frac{C}{N\ell^{3/2+\kappa}} \] 
for $\ell \leq N^{-2/3}$. Since $\| u_\ell \|_1 \leq C \ell^2 /N$ by Lemma \ref{lm:hardcorescatt}, the error term $\text{R}_{11}$ introduced in (\ref{eq:R1112}) is bounded by 
\[ \text{R}_{11} \leq C N^{-\kappa + \mu} \ell^{1/2-\kappa} \leq C N^\mu \ell^{1/2}\,. \]  
With (\ref{eq:R12fin}), we find 
\[ \text{R}_1 \leq C N^\mu N^{-\wt{\eps}} \]
for $\wt{\eps} > 0$ small enough. Combining this bound with (\ref{eq:R2fin}) we conclude, choosing first $\mu > 0$ small enough and then $m \in \bN$ sufficiently large, that  
\[ \frac{N(N-1)}{2}  \Big| \big\langle \Phi_N , \big\{ \big[ U_{N-2}^* e^{B(\eta)} \cP^{(2+\kappa)} (\cN_+ + 1) e^{-B(\eta)} U_{N-2} \big] \otimes u_\ell (x_{N-1} - x_N) \big\} \Phi_N \big\rangle \Big| \leq C N^{-\eps} \]
for a sufficiently small $\eps > 0$. Together with (\ref{eq:diff-en}) and (\ref{eq:UNUfin}), this estimate implies that 
\[ - \frac{N(N-1)}{2} \big\langle \Phi_N, \big\{  \big[ H_{N-2}^\text{eff} - 4 \pi \frak{a} N \big] \otimes u_\ell (x_{N-1} - x_N) \big\} \Phi_N \big\rangle \leq C N^{-\eps} \,.\]
From (\ref{eq:en-Psi2}), we obtain 
\[ \begin{split} &\frac{\langle \Psi_N , \sum_{j=1}^N -\Delta_{x_j} \Psi_N \rangle}{\| \Psi_N \|^2} \\ &\hspace{.4cm}  \leq 4\pi \frak{a} (N-1) + e_\Lambda \frak{a}^2 -\frac{1}{2} \sum_{p \in \L^*_+} \Big[ p^2 + 8\pi \frak{a} - \sqrt{|p|^4 + 16 \pi \frak{a} p^2} - \frac{(8\pi \frak{a})^2}{2p^2} \Big] + C N^{-\eps} \,.\end{split} \] 

We conclude the proof of Theorem \ref{thm:main} by giving the proof of Lemma \ref{lemma:Pr-ex}.
\begin{proof}[Proof of Lemma \ref{lemma:Pr-ex}]
With (\ref{eq:commaa}), we can compute $[\mathcal{P}^{(r)},B(\eta)]$ to show that  
\begin{equation} \label{eq:first_commut_diag}
e^{-B(\eta)} \mathcal{P}^{(r)} e^{B(\eta)} = \mathcal{P}^{(r)}+ \int_0^1 ds \,e^{-sB(\eta)} \sum_{p \in \L^*_+} |p|^r \eta_p \Big[ b^*_p b_{-p}^*+b_p b_{-p} \Big]e^{sB(\eta)}\,.
	\end{equation}
Furthermore, expanding the integrand on the r.h.s. of (\ref{eq:first_commut_diag}), we write 
	\begin{equation} \label{eq:first_commut_off_diag}
	\begin{split}
	e^{-sB(\eta)} \sum_{p \in \L^*_+} |p|^r& \eta_p \Big[ b^*_p b_{-p}^*+b_p b_{-p} \Big]e^{sB(\eta)} \\
	=\;&\sum_{p \in \L^*_+} |p|^r \eta_p \Big[ b^*_p b_{-p}^*+b_p b_{-p} \Big] \\
	&+ \int_0^s dt\, e^{-tB(\eta)} \sum_{p \in \L^*_+} |p|^r \eta_p \Big[ b^*_p b^*_{-p}+b_p b_{-p},B(\eta) \Big]e^{tB(\eta)}\,.
	\end{split}
	\end{equation}
	Let us compute the last commutator. With \eqref{eq:comm-bp},
we find
	\begin{equation*}
	\begin{split}
	\sum_{p \in \L^*_+} |p|^r \eta_p \Big[ b^*_p b^*_{-p} + b_p b_{-p},B(\eta) \Big]=\;& \frac{1}{2} \sum_{p,q \in \L^*_+} |p|^r \eta_p \eta_q \big[ b_p b_{-p},b^*_q b^*_{-q} \big]+\mathrm{h.c.}\\
	=\;&2\sum_{p \in \L^*_+}|p|^r \eta_p^2+ \Xi
	\end{split}
	\end{equation*}
	with
	\begin{equation*}
	\begin{split}
	\Xi=\;&4 \sum_{p\in \L^*_+} |p|^r \eta_p^2 a^*_p \left[ \left( 1-\frac{\mathcal{N}_++2}{N} \right)\left( 1-\frac{ \mathcal{N}_++1}{N} \right) -\frac{1}{2N^2}\right]a_p\\
	&+ 2\sum_{p \in \L^*_+}|p|^r \eta_p^2 \left[\left( 1-\fra{\mathcal{N}_++1}{N} \right) \left( 1-\fra{\mathcal{N}_+}{N} \right)-1\right]\\
	&-\frac{1}{N} \sum_{p,q \in \L^*_+}|p|^r \eta_p \eta_q a^*_q a^*_{-q} \left[ 2\left(1-\frac{\mathcal{N}_+}{N}\right) -\frac{3}{N}\right] a_p a_{-p}+\mathrm{h.c.}\,.
	\end{split}
	\end{equation*}
To control the last term, we write $a_q^* a_{-q}^* a_p a_{-p} = a_q^* a_p a_{-q}^* a_{-p} - \delta_{-q,p} a_q^* a_{-p}$ and  we bound, for an arbitrary $\xi \in \cF_+^{\leq N}$,  
\[ \begin{split} \Big| \frac{1}{N} \sum_{p,q \in \L^*_+} |p|^r &\eta_p \eta_q \langle \xi, a_q^* a_p a_{-q}^* a_{-p} \xi \rangle \Big| \\ &\leq \frac{1}{N} \sum_{p,q \in \L^*_+} |p|^r |\eta_p| |\eta_q| \| a_p^* a_q \xi \| \| a_{-q}^* a_{-p} \xi \| \\ &\leq \frac{1}{N} \sum_{p,q \in \L^*_+} |p|^r |\eta_p| |\eta_q| \big[ \| a_p a_q \xi \| + \| a_q \xi \| \big] \big[ \| a_{-q} a_{-p} \xi \| + \| a_{-p} \xi \| \big] .\end{split} \]
With Cauchy-Schwarz's inequality and with the bounds $r \leq 4$, $|\eta_p|\le C|p|^{-2}$, we find
\begin{equation*}
	\pm \Xi \le C (\cN_+ + 1) + C N^{-1}  \left( \cP^{(r)} + \ell^{1-r} \right) (\cN_+ + 1) .
\end{equation*}
Inserting this back in \eqref{eq:first_commut_off_diag} and using \eqref{eq:P_r} we obtain 
	\begin{equation*}
	e^{-sB(\eta)} \sum_{p \in \L^*_+} |p|^r \eta_p \Big[ b^*_p b_{-p}^*+b_p b_{-p} \Big]e^{sB(\eta)}=\sum_{p \in \L^*_+} |p|^r \eta_p \Big[ b^*_p b_{-p}^*+b_p b_{-p} \Big]+2s \sum_{p \in \L^*_+} |p|^r \eta_p^2+ \widetilde{\Xi}
	\end{equation*}
	again with
	\begin{equation*}
	\pm \widetilde{\Xi} \le C (\cN_+ + 1) + C N^{-1}  \left( \cP^{(r)} + \ell^{1-r} \right) (\cN_+ + 1) .	\end{equation*}
Setting $s=1$, this proves \eqref{eq:exact_action_offdiag}. Plugging now \eqref{eq:exact_action_offdiag} in \eqref{eq:first_commut_diag} and integrating over $s$ we find \eqref{eq:exact_action_P}.
\end{proof}

\appendix

\section{Properties of one-particle scattering equations} \label{App:scatt}\label{sec:app}

In this section we provide the proof of Lemma \ref{lm:hardcorescatt} and Lemma \ref{lm:driftscatt}. We start with Lemma \ref{lm:hardcorescatt}, where we describe  properties of the solution of the eigenvalue equation (\ref{eq:fell}).

\begin{proof}[Proof of Lemma \ref{lm:hardcorescatt}] 
By standard arguments, the ground state solution of (\ref{eq:ev}) is radial. Thus, we consider the ansatz $f_\ell (x) = m_\ell (|x|) / |x|$, which leads to the equation
\[ m'' (r) + \lambda_\ell m (r) = 0 \]
for $r \in [\aa /N ; \ell]$, with the boundary conditions $m (\aa /N) = 0$, $m' (\ell) = 1$ and $m(\ell) = \ell$. From $m (\aa /N) = 0$ and $m (\ell) = \ell$, we obtain 
\[ 
m (r)= \frac{\ell \: \sin(\sqrt \l_\ell (r-\aa/N))}{\sin(\sqrt \l_\ell (\ell-\aa/N))} 
\]
for all $r \in [\aa /N ; \ell]$. This proves (\ref{eq:fell-x}). Imposing $m' (\ell) = 1$, we arrive at 
\begin{equation}\label{eq:tan2} \tan \big( \sqrt{\l_\ell}\, (\ell- \aa/N) \big) = \sqrt{\l_\ell} \,\ell \end{equation} 
which shows (\ref{eq:lambdaell}). This equation allows us to estimate the eigenvalue $\lambda_\ell$. As already shown in \cite[Lemma A.1]{DGPHD-BEC}, we find 
\begin{equation}\label{eq:lam1} \lambda_\ell = \frac{3\aa}{N \ell^3} \big( 1 + \cO (\aa / N \ell) \big) \end{equation} 
which implies that $\sqrt{\lambda_\ell} (\ell -\aa /N) \simeq \sqrt{\lambda_\ell}  \ell  \simeq (N \ell)^{-1/2} \ll 1$. With $\tan s = s + s^3 /3 + 2 s^5 /15 + \cO (s^7)$, we obtain
\[ \sqrt{\lambda_\ell}  \,\ell =  \sqrt{\l_\ell} \big( \ell - \aa /N \big) + \frac{1}{3} \l_\ell^{3/2} \big( \ell -  \aa /N \big)^3 + \frac{2}{15} \l_\ell^{5/2} \big( \ell -  \aa /N \big)^5 + \cO \big( (N\ell)^{-7/2} \big)  \]
which leads to (\ref{eq:lambdaell-exp}). 

With (\ref{eq:lam1}) for $\lambda_\ell$, we can expand the expression (\ref{eq:fell-x}). We find, for $\aa/N \leq |x| \leq \ell$,  
\begin{align} \label{espansione-f-ell}
f_\ell(x) &= 1  - \frac{\aa}{N|x|} +\frac{3\aa}{2N\ell} - \frac{\aa^2}{2N^2 \ell |x|} - \frac{\aa |x|^2}{2N\ell^3} +\cO \Big(\frac {\aa^2} {N^2 \ell^2}\Big)\\
\partial_r f_\ell (x) & =  \frac \aa {N |x|^2} - \frac {\aa |x|} {N \ell^3} +\cO \Big(\frac {\aa^2} {N^2 \ell^2 |x|}\Big)\,.
\end{align}
With these approximations, we obtain (\ref{eq:Vell-zero}), (\ref{eq:w-bds}), \eqref{eq:normsomega} and \eqref{eq:Lp-norms}. Finally, we show (\ref{eq:omegap}). An explicit computation (using also the eigenvalue equation (\ref{eq:tan2})) gives
\begin{equation*}\label{eq:wp1}
\begin{split}  \widehat{\o}_p &= \widehat{\chi}_\ell (p) - \frac{2\pi \ell}{\sin \big( \sqrt{\lambda_\ell} \, (\ell - \aa/N) \big)}  \int_{\aa / N}^\ell  dr \, r  \sin \big( \sqrt{\lambda_\ell} \, (r-\aa/N) \big) \int_0^\pi  d\theta \sin \theta \, e^{-i|p| r \cos \theta} \\ &=  \frac{\lambda_\ell}{\lambda_\ell - p^2} \frac{4\pi \ell}{p^2} \Big[ \frac{\sin (|p| \ell)}{|p| \ell} - \cos |p| \ell \Big] - \frac{4\pi}{|p| (\lambda_\ell - p^2)} \frac{\sin (|p| \aa /N)}{\cos (\sqrt{\lambda_\ell} (\ell - \aa_N))}. \end{split} \end{equation*} 
For $|p| \geq \ell^{-1}$, we have $|\lambda_\ell - p^2| \geq c p^2$. With (\ref{eq:lam1}), we easily find $|\widehat{\omega}_p| \leq C / (Np^2)$, if $\ell^{-1} \leq |p| \leq N$, and $|\widehat{\omega}_p| \leq C / |p|^3$, if $|p| > N$. From \eqref{eq:normsomega}, we also have $|\widehat{\omega}_p| \leq \| \omega \|_1 \leq C \ell^2 / N$ for all $p \in \L^*$; this implies (\ref{eq:omegap}). 
\end{proof}

Next, we show Lemma \ref{lm:driftscatt}, devoted to the properties of the solution of (\ref{eq:g0}). 

\begin{proof}[Proof of Lemma \ref{lm:driftscatt}] 
We begin with (\ref{eq:checketa}). From the definition $g_{\ell_0}(x)=f_{\ell_0}(x)/f_\ell(x)$ and the explicit expression \eqref{eq:fell-x} we have 
 \be \label{eq:gell0-explicit}
g_{\ell_0}(x)=  \frac{\ell_0}{\ell}\, \frac{\sin(\sqrt \l_{\ell_0} (|x|-\aa/N))}{\sin(\sqrt \l_{\ell_0} (\ell_0-\aa/N))}  \, \frac{\sin(\sqrt \l_\ell (\ell-\aa/N))}{\sin(\sqrt \l_\ell (|x|-\aa/N))} 
\ee
for all $\aa / N \leq |x| \leq \ell$. Expanding, we find $g_{\ell_0} (x) = 1 + \cO (\aa / N\ell)$ and thus $|\check{\eta} (x)| \leq C \aa / \ell \leq C \aa / (|x| + \ell)$ for all $\aa / N \leq |x| \leq \ell$. For $|x| \geq \ell$, $g_{\ell_0} (x) = f_{\ell_0} (x)$ and (\ref{eq:w-bds}) implies that $|\check{\eta} (x)| \leq C \aa / |x| \leq C \aa /(|x| + \ell)$. Finally, for $|x| \leq \aa / N$, we defined 
\[ g_{\ell_0} (x) = \lim_{|y| \downarrow \aa / N} g_{\ell_0} (y) = 1 - \frac{3\aa}{2N \ell} + \cO \Big( \frac{1}{N^2 \ell^2} \Big) \]
which gives $|\check{\eta} (x)| \leq C \aa / \ell \leq C \aa  / (|x| + \ell)$. This shows the first estimate in (\ref{eq:checketa}). To bound $\nabla \check{\eta}$, we proceed similarly. For $\aa / N \leq |x| \leq \ell$, we find 
\be \label{eq:dpr-f}
\partial_r f_\ell (x) = f_\ell(x) \left( \frac{\sqrt{\l_\ell}}{\tan \big(\sqrt{\l_\ell}(r -\aa/N)\big)} - \frac 1 r \right)
\ee
and thus 
\be \label{eq:dpr-eta}
\dpr_r \check{\eta}(x) = N g_{\ell_0}(x)  \left( \frac{\sqrt{\l_\ell}}{\tan \big(\sqrt{\l_\ell}(|x| -\aa/N)\big)}- \frac{\sqrt{\l_{\ell_0}}}{\tan \big(\sqrt{\l_{\ell_0}}(|x| -\aa/N)\big)} \right)\,.
\ee
With $|g_{\ell_0}(x)| \leq C$ and expanding $\tan s = s + \cO (s^3)$, we find $|\nabla \check{\eta} (x)| \leq C \aa / \ell^2 \leq C \aa / (|x| + \ell)^2$, for all $\aa / N \leq |x| \leq \ell$. For $|x| \geq \ell$, we have $g_{\ell_0} (x) = f_{\ell_0} (x)$ and the estimate $|\nabla \check{\eta} (x)| \leq C \aa / (|x| + \ell)^2$ follows from (\ref{eq:w-bds}). 

Next, we show \eqref{eq:intVellg0}. With \eqref{eq:g0} (noticing that the flux of $f_\ell^2 \nabla g_{\ell_0}$ through the spheres $|x| = \aa / N$ and $|x| = \ell_0$ vanishes), 
we have
\[\begin{split}
2 N \lambda_\ell \int \chi_{\ell}(x) f_\ell^2 (x) g_{\ell_0}(x) dx&\; =  2N \l_{\ell_0}  \int \chi_{\ell_0}(x) f_\ell^2(x)  g_{\ell_0}(x) dx \\
&\;  = 2N \l_{\ell_0}  \int \chi_{\ell_0}(x) dx +  2N \l_{\ell_0}  \int \chi_{\ell_0}(x)  (f_\ell f_{\ell_0}-1) (x) 
\end{split}\]
since $g_{\ell_0}(x)=f_{\ell_0}(x)/f_\ell(x)$.  With Lemma \ref{lm:hardcorescatt} we have
\[
\Big| 2N \l_{\ell_0}  \int \chi_{\ell_0}(x) dx - 8  \pi \aa\Big | \leq C N^{-1}
\]
and
\[ \begin{split}
\Big| 2N \l_{\ell_0}  \int \chi_{\ell_0}(x)  f_{\ell}(x) (f_{\ell_0}(x)-1)  dx\Big | & \leq  C \| \o_{\ell_0}\|_1 \leq C /N \\
\Big| 2N \l_{\ell_0}  \int \chi_{\ell_0}(x) (f_\ell(x)-1)  dx\Big |& \leq  C \| \o_\ell\|_1 \leq C \ell^2 /N \,.
\end{split}\]
This proves the first bound in \eqref{eq:intVellg0} and also in (\ref{eq:intchiell0}). 
To show the second 
bound in \eqref{eq:intVellg0}, we compute (with a slight abuse of notation we write here, for $r > 0$, $f_\ell (r), g_{\ell_0} (r)$ to indicate the values of $f_\ell (x), g_{\ell_0} (x)$, for $|x| = r$)  
\[ \begin{split} 
2N \l_\ell \int &\chi_\ell (x) f_\ell^2 (x) g_{\ell_0} (x) e^{-i p \cdot x} dx \\ &=\;  \frac{4\pi  N \l_\ell }{|p|}\int_{\aa/N}^\ell r  f_\ell^2(r) g_{\ell_0}(r) \sin(|p|r)dr \\
& = \frac{4\pi N \l_\ell}{|p|^2} \int_{\aa/N}^\ell    \Big[ (f_\ell^2(r) + 2r f_\ell(r)\, \dpr_r f_\ell(r))  g_{\ell_0}(r) + r f_\ell^2(r) \dpr_r g_{\ell_0}(r)\Big] \cos (|p|r) dr  \\
& \hskip 2cm - \frac{4 \pi \ell N \l_\ell}{|p|^2}  g_{\ell_0}(\ell) \cos(|p|\ell)\,.
\end{split}\]
From Lemma \ref{lm:hardcorescatt}, we have $f_\ell (r) , r |\partial_r f_\ell (r)| \leq C$. From 
(\ref{eq:checketa}), we find (recalling that $g_{\ell_0} = 1 + \check{\eta} / N$) that $|\partial_r g_{\ell_0} (r)| \leq C / (N\ell^2)$. With the bound (\ref{eq:lambdaell-exp}) (or (\ref{eq:lam1})) for $\lambda_\ell$, we conclude that 
\[ \Big| 2N \l_\ell \int \chi_\ell (x) f_\ell^2 (x) g_{\ell_0} (x) e^{-i p \cdot x} dx \Big| \leq \frac{C}{\ell^2 p^2} \,.\]
The second bound in (\ref{eq:intchiell0}) can be proven analogously (on the r.h.s. $\ell$ is then replaced by $\ell_0$, which is chosen of order one). 

Eqs. (\ref{eq:eta-scat0}), (\ref{eq:eta-scat}) follow directly from (\ref{eq:g0}). As for \eqref{eq:eta0}, we rewrite
\be\label{eq:eta0bis}
	\eta_0=\int \check{\eta}(x)\, dx=\int_{|x|<\ell} \check{\eta}(x)\, dx-N\int_{|x|>\ell}\omega_{\ell_0}(x)\, dx\,.
\ee
Using \eqref{eq:normsomega},\eqref{espansione-f-ell} and the fact that $g_{\ell_0} (x)=1+O(\aa/N\ell)$ for $|x|<\ell$,  we obtain 
\[
	\eta_0=-N \int \omega_{\ell_0}(x)\, dx + \int_{|x|<\ell} \check{\eta}(x)\, dx+N\int_{|x|<\ell}\omega_{\ell_0}(x)\, dx=-\frac25\pi\aa\ell_0^2+\cO\Big(\frac{\aa^2\ell_0}{N}\Big)+\cO(\aa\ell^2) \,.
\]

To prove (\ref{eq:etap}), we consider the Fourier coefficients $D_p$ defined in (\ref{eq:Dp-def}) and the corresponding function $\check{D} (x) = -\nabla \cdot  \big[ (f_\ell^2 (x) - 1) \nabla \check{\eta} (x) \big]$. For any $p \in \L^*$, we have 
\begin{equation}\label{eq:Dp1} |D_p| \leq \int_{\aa / N \leq |c| \leq \ell} | \check{D} (x)| dx\,. \end{equation} 
For $\aa / N \leq |x| \leq \ell$, we find 
\be \begin{split} \label{eq:Dr}
\check{D}(x) & =  2 f_\ell(x)  (\dpr_r f_\ell)(x) (\dpr_r \check{\eta})(x) + (f^2_\ell(x)-1) \D \check{\eta}(x)  \\ & = N g_{\ell_0}(x) (\l_\ell - \l_{\ell_0}) (f^2_\ell(x)-1) -2 (\dpr_r \check{\eta})(x) \frac{(\dpr_r f_\ell)(x)}{f_\ell(x)}
\end{split}\ee
where in the second line we used the definition $\check{\eta} = N (g_{\ell_0} -1) = N (f_{\ell_0} / f_{\ell} - 1)$ and the scattering equation (\ref{eq:ev}) for $f_\ell$ and $f_{\ell_0}$ to replace 
\[ \label{eq:eta-scat-pos}
 \D \check \eta (x) = -2 N \frac{\nabla f_\ell (x)}{f_\ell (x)} \cdot \nabla g_{\ell_0} (x) + N  (\l_\ell -\l_{\ell_0}) g_{\ell_0} (x) \,.
\]
Using (\ref{eq:gell0-explicit}) to bound $|g_{\ell_0} (x) | \leq C$, (\ref{eq:w-bds}) to show $|f_\ell^2 (x) - 1| \leq C \aa / (N |x|)$ and using \eqref{eq:dpr-f}, \eqref{eq:dpr-eta} to control  the second term on the r.h.s. of (\ref{eq:Dr}), we find 
\begin{equation}\label{eq:Dx0}  \big| \check{D} (x) \big| \leq \frac{C \aa}{|x|} (\l_\ell - \l_{\ell_0}) \leq \frac{C \aa^2}{N |x| \ell^3} \end{equation} 
for all $\aa/ N \leq |x| \leq \ell$. Inserting \eqref{eq:Dx0} in \eqref{eq:Dp1}, we arrive at \begin{equation} \label{eq:Dp-bd1} 
|D_p| \leq C / (N \ell). \end{equation} 
From the scattering equation (\ref{eq:eta-scat0}), we can estimate 
\be \label{eq:etap-proof}
|\eta_p| \leq \frac{C}{|p|^2} \Big ( |D_p| +  |(\widehat V_\ell \ast \widehat{g}_{\ell_0})(p) |+|N\l_{\ell_0}( \widehat{\chi_{\ell_0} f_\ell^2} \ast \widehat{g}_{\ell_0})(p)|\Big) \,.
\ee
Combining (\ref{eq:Dp-bd1}) with the first bounds in (\ref{eq:intVellg0}), (\ref{eq:intchiell0}), we immediately conclude that $|\eta_p| \leq C / p^2$. To prove the remaining bounds in (\ref{eq:etap}), we write 
\be \begin{split} \label{eq:Dp-improved}
 D_p &\;= \frac{4\pi}{|p|} \int_{\aa/N}^\ell  r \,  \check{D}(r) \sin (|p|r) dr \\ 
&\; = \frac{4\pi}{|p|^2}  \, \left[ r \check{D} (r) \cos (|p| r) \right] |_{r = \frak{a}/N} + \frac{4\pi}{|p|^2} \int_{\aa/N}^\ell  \Big[ \check D(r) + r \, \dpr_r  \check{D}(r) \Big] \cos (|p|r) dr  \,.
\end{split}\ee
With \eqref{eq:Dx0} and $N \ell \gg 1$, we can estimate the boundary term by 
\begin{equation}\label{eq:boundary}  \left| \frac{4\pi}{|p|^2}  \, \left[ r \check{D} (r) \cos (|p| r) \right] |_{r = \frak{a}/N} \right| \leq \frac{C}{p^2 \ell^2} \, . \end{equation}
From \eqref{eq:Dr} we get
\be \begin{split} \label{eq:dpr-Dr}
\dpr_r \check D(r) =\;&  \dpr_r\check{\eta}(r) (\l_\ell - \l_{\ell_0})(f_\ell^2(r)-1) + N g_{\ell_0}(r) (\l_\ell - \l_{\ell_0}) 2 f_\ell(r) \dpr_r f_\ell(r)   \\
& -2 \dpr_r \Big((\dpr_r \check{\eta})(r) \frac{(\dpr_r f_\ell)(r)}{f_\ell(r)}\Big)\,.
\end{split}\ee
Using the bounds  $|\dpr_r\check{\eta}(r) |\leq C r^{-2}$, $|\dpr_r f_\ell(r)|\leq (C N r^2)^{-1}$, the boundness of $f_\ell$ and $g_{\ell_0}$ and \eqref{eq:lambdaell-exp} we easily see that the first line of \eqref{eq:dpr-Dr} is bounded by $\cO \big( (N \ell^3 r^2)^{-1} \big)$. As for the second line of \eqref{eq:dpr-Dr}, we find, using \eqref{eq:dpr-f} and  \eqref{eq:dpr-eta}, 
\[ \begin{split}
& - \dpr_r \Big((\dpr_r \check{\eta})(r) \frac{(\dpr_r f_\ell)(r)}{f_\ell(r)}\Big) \\
&= N g_{\ell_0}(r) \Bigg[  \bigg( \frac{\sqrt{\l_\ell}}{\tan \big(\sqrt{\l_\ell}(r -\aa/N)\big)}- \frac{\sqrt{\l_{\ell_0}}}{\tan \big(\sqrt{\l_{\ell_0}}(r -\aa/N)\big)}  \bigg)^2  \bigg( \frac{\sqrt{\l_\ell}}{\tan \big(\sqrt{\l_\ell}(r -\aa/N)\big)} - \frac 1 r \bigg) \\
& \hskip 1cm + \bigg( \frac{\l_\ell}{\sin^2\big(\sqrt{\l_\ell}(r -\aa/N)\big)}  - 
\frac{\l_{\ell_0}}{\sin^2\big(\sqrt{\l_{\ell_0}}(r -\aa/N)\big)}\bigg)  \bigg( \frac{\sqrt{\l_\ell}}{\tan \big(\sqrt{\l_\ell}(r -\aa/N)\big)} - \frac 1 r \bigg) \bigg) \\
& \hskip 1cm - \bigg( \frac{\sqrt{\l_\ell}}{\tan \big(\sqrt{\l_\ell}(r -\aa/N)\big)}- \frac{\sqrt{\l_{\ell_0}}}{\tan\big (\sqrt{\l_{\ell_0}}(r -\aa/N)\big)}  \bigg)\bigg( \frac 1 {r^2}- \frac{\l_\ell}{\sin^2\big(\sqrt{\l_\ell}(r -\aa/N)\big)} \bigg)\Bigg]\,.
\end{split}\]
Expanding $1/ \tan (s)= 1/s + s/3 + \cO(s^3)$ and $1/\sin^2 (s) = 1/ s^{2} +1/ 3 + \cO (s^2)$, we obtain 
\[
\Big| \dpr_r \Big((\dpr_r \check{\eta})(r) \frac{(\dpr_r f_\ell)(r)}{f_\ell(r)}\Big) \Big| \leq  \frac{C}{r^2} (\l_\ell -\l_{\ell_0}) \leq \frac{C}{N \ell^3 r^2}\,.
\]
Thus, $| r \partial_r \check{D} (r)| \leq C / (N \ell^3 r) \leq C / \ell^3$ for all $\aa / N \leq |x| \leq \ell$. Combined with (\ref{eq:Dx0}), (\ref{eq:Dp-improved}) and (\ref{eq:boundary}), we conclude that 
\[ |D_p| \leq \frac{C}{p^2 \ell^2}\,. \]
Inserting this estimate in (\ref{eq:etap-proof}), together with the second bounds in (\ref{eq:intVellg0}), (\ref{eq:intchiell0}), we obtain $|\eta_p| \leq C/ (\ell^2 |p|^4)$, which finishes the proof of (\ref{eq:etap}). Eq. (\ref{eq:etaHr}) is a simple consequence of (\ref{eq:etap}). 
\end{proof}

\bigskip

{\it Acknowledgment.} B. S. gratefully acknowledges partial support from the NCCR SwissMAP, from the Swiss National Science Foundation through the Grant ``Dynamical and energetic properties of Bose-Einstein condensates'' and from the European Research Council through the ERC-AdG CLaQS.  G.B. acknowledges support through the project ``Progetto Giovani GNFM 2020: Emergent Features in Quantum Bosonic Theories and Semiclassical Analysis''. G.B., S.C., and A.O. warmly acknowledge support of  the GNFM Gruppo Nazionale per la Fisica Matematica - INDAM.

%%%%%%%%%%%%%%%%%%%%%%%%%%%%%%%%%
%%%%%%%%%%%%%%%%%%%%%%%%%%%%%%%%%
%%%%%%%%%%%%%%%%%%%%%%%%%%%%%%%%%

\def\bskip{\\[-0.6cm]}

\end{document}